\documentclass[twocolumn,prx,nofootinbib,superscriptaddress, aps, 10pt]{revtex4-2}
\usepackage{algorithm2e}

\usepackage{amssymb,amsxtra,amsmath,amsfonts, amsthm}
\usepackage{enumerate}
\usepackage{hyperref}
\usepackage[toc,page]{appendix}
\usepackage[caption=false]{subfig}
\usepackage{tikz}
\usetikzlibrary{external}
\usetikzlibrary{fit}
\usetikzlibrary{positioning}
\usetikzlibrary{calc}

\def\aa{\mathfrak{a}}
\def\bb{\mathfrak{b}}
\def\cc{\mathfrak{c}}
\def\g{\mathfrak{g}}
\def\uu{\mathfrak{u}}

\def\su{\mathfrak{su}}
\def\so{\mathfrak{so}}
\def\sp{\mathfrak{sp}}

\def\mcA{\mathcal{A}}
\def\mcB{\mathcal{B}}

\def\mcP{\mathcal{P}}

\def\ad{\mathrm{ad}}
\def\Span{\mathrm{span}}

\newcommand\Lie[1]{\langle #1 \rangle_{\mathrm{Lie}}}

\newtheorem{theorem}{Theorem}
\newtheorem{definition}[theorem]{Definition}

\newtheorem{lemma}[theorem]{Lemma}
\newtheorem{proposition}[theorem]{Proposition}

\newtheorem{corollary}[theorem]{Corollary}
\numberwithin{theorem}{section}

\theoremstyle{definition}
\newtheorem{example}{Example}[section]

\newcommand{\LK}[1]{\textcolor{magenta}{[LEX: #1]}}

\newcommand{\EK}[1]{\textcolor{blue}{[EFEKAN: #1]}}

\usepackage[english]{babel}
\usepackage{adjustbox}
\usepackage{amsmath}
\usepackage{graphicx}
\usepackage{lipsum}
\usepackage{xcolor}
\hypersetup{colorlinks=true, linkcolor=blue, citecolor=blue, urlcolor=blue, unicode=true}
\usepackage{cleveref}
\definecolor{fgraph}{RGB}{20, 79, 181}
\definecolor{highlight}{RGB}{235, 101, 52}
\definecolor{oldgraph}{RGB}{181,20, 20}

\usepackage{enumitem}
\newlist{abbrv}{itemize}{1}
\setlist[abbrv,1]{label=,labelwidth=1.1in,align=parleft,itemsep=0.1\baselineskip,leftmargin=!}
\newif\ifTikz
\Tikztrue
\begin{document}

\title{Classification of dynamical Lie algebras generated\\
by spin interactions  
on undirected graphs}

\author{Efekan K\"okc\"u}
\affiliation{Department of Physics, North Carolina State University, Raleigh, NC 27695, USA}
\address{Computational Research Division,
            Lawrence Berkeley National Laboratory,
            Berkeley, CA 94720, USA}
\email{ekokcu@lbl.gov}

\author{Roeland Wiersema}
\affiliation{Vector Institute, MaRS  Centre,  Toronto,  Ontario,  M5G  1M1,  Canada}
\affiliation{Department of Physics and Astronomy, University of Waterloo, Ontario, N2L 3G1, Canada}
\affiliation{Xanadu, Toronto, ON, M5G 2C8, Canada}

\author{Alexander F. Kemper}
\affiliation{Department of Physics, North Carolina State University, Raleigh, NC 27695, USA}

\author{Bojko N. Bakalov}
\affiliation{Department of Mathematics, North Carolina State University, Raleigh, NC 27695, USA}

\begin{abstract}

We provide a classification of all dynamical Lie algebras generated by 2-local spin interactions 
on undirected graphs. Building on our previous work in~\cite{wiersema2023classification}, 
where we provided such a classification for spin chains, here we consider the more general case of undirected graphs. As it turns out, 
the one-dimensional case is special; for any other graph, the dynamical Lie algebra solely depends on whether the graph is bipartite or not. 
An important consequence of this result 
is that the cases where the dynamical Lie algebra is polynomial in size are special 
and restricted to one dimension.

\end{abstract}

\date{September 29, 2024}

\maketitle

\section{Introduction}

A \emph{dynamical Lie algebra}~\cite{zeier2011symmetry,zimboras2015symmetry,dalessandro2021}, or DLA, refers to the Lie algebra generated by the terms of a Hamiltonian describing a quantum system. The name comes from the fact that the dynamics of a quantum system are controlled by the exponential map of the Hamiltonian; the potential dynamical ``directions''
are given by linear combinations of nested commutators of the individual (non-commuting) terms in the Hamiltonian, as can be seen from the
Baker--Campbell--Hausdorff formula~\cite{knapp2013lie}.
Originally, DLAs were explored in the domain of quantum control, where the primary focus was on assessing the controllability of a quantum system~\cite{dalessandro2021}.
They have recently gained renewed attention due to their significant connections with variational quantum computing, particularly in relation to the trainability of quantum circuits~\cite{mcclean2018barren,marrero2020entanglement,ragone2023,fontana2023theadjoint,cerezo2023simul,larocca2023theory}. In the early 2000s, work of Somma~\cite{somma2005quantum,somma2006efficient} connected DLAs to classical simulatibility by proposing efficient algorithms for simulating DLAs whose dimension scaled polynomially with the system size. Furthermore, in the realm of condensed matter physics, DLAs have been employed to construct path integrals for many-body quantum systems~\cite{galitski2011hubbard}. Finally, DLAs appear under a different nomenclature in discussions surrounding Hamiltonian symmetries and bond algebras~\cite{Moudgalya2022fragment}.

Given their broad relevance across various areas of physics, a comprehensive understanding of DLAs is essential. In our previous work~\cite{wiersema2023classification}, we classified DLAs that arise from a specific class of generators: 1- and 2-local Pauli spin-$1/2$ operators in one dimension. Here, we generalize this classification to encompass DLAs that are generated by 1- and 2-local Pauli spin operators where the operators lie on the vertices and edges of an arbitrary undirected graph, which we call the \emph{interaction graph}. Our earlier classification provides a crucial component---the DLAs for the complete graph $K_n$---which facilitated this extension to general graphs. Surprisingly, expanding the classification to more complex topologies is more straightforward than one would
naively expect from such an increase in complexity; rather, the extension is more straightforward. The inverse is true:
due to the restrictive nature of one dimension, a variety
of DLAs arise that do not in more complex topologies.
Intuitively, the increased connectivity leads somewhat
inevitably to a bigger DLA, with fewer symmetries.

Recently, Aguilar and coauthors~\cite{aguilar2024full}
have presented a similar work, which provides a full classification of Pauli Lie algebras, i.e., Lie algebras generated by Pauli strings. They do so by use
of the \emph{frustration graph}~\cite{chapman2020characterization},
a graph whose vertices are the generators and edges
exist if two vertices do not commute (see Appendix \ref{sec:frustr}). Aguilar et al.\ show
that any frustration graph can, in principle, be reduced via certain operations to one of several fundamental graphs, and provide the corresponding fundamental dynamical Lie
algebras.

While Ref.~\cite{aguilar2024full} is quite general
and shows that the dynamical Lie algebra generated by any set of Pauli strings is equivalent to that generated by one of their fundamental graphs, it is not immediately clear which one.  In contrast, although it is restricted to algebras generated by 1 and 2-local Pauli strings, we do provide an explicit answer.
Moreover, the approach of this work is based on interaction graphs, which can be more convenient in a number of ways, e.g., when dealing with subgraphs and their corresponding DLAs. In particular, with interaction graphs we can apply equivalence relations on a subgraph in a local manner (c.f. Lemma~\ref{lem:loc}), whereas this cannot be necessarily done for frustration graphs. Although our results primarily make use of interaction graphs, we note that frustration graphs are a
useful tool and that a number of our results can be obtained using them. For completeness, we have included
a discussion of frustration graphs and some associated proofs in the Appendix. 

\subsection{Summary of the Main Results}

Here we give a brief summary of our main results, which extend the classification of DLAs generated by 1- and 2-local operators from
line, cycle, and complete graphs ($L_n$, $C_n$, $K_n$)
\cite{wiersema2023classification} 
to DLAs generated by the same operators on an arbitrary undirected graph $G$. 
The vertices of the interaction graph $G$ correspond to qubits, and its edges to 2-local interactions between qubits.
Note that when $G$ is a disjoint union of connected components, the DLA is a direct sum of commuting subalgebras corresponding to the components.

In the following, we will consider connected graphs $G$ with $n$ vertices and $E$ edges.  
Because we have already considered the line and cycle graphs in~\cite{wiersema2023classification}, 
we will be focusing on graphs with at least one vertex of degree $>2$.
As the interaction graph is undirected, we further restrict the classification to the case where the DLA is symmetric under exchange of the qubits, which means that whenever we have a
generator such as $XY$ between two sites we must also have $YX$.

As in Ref.~\cite{wiersema2023classification},
we have two types of Lie algebras: $\aa$-type, which can
be generated by a set formed of only 2-local interactions (not including the identity matrix),
and $\bb$-type for which this is not possible.
The list of these Lie algebras, together with their generators and bases, is given in Appendix \ref{sec:list_dla}. When restricted to symmetric DLAs, we are left with 9 $\aa$-type and 3 $\bb$-type Lie algebras~\cite{wiersema2023classification}:
$\aa_k$ for $k = 0,2,4,6,7,14,16,20,22$, and $\bb_l$ for $l = 0,1,3$. 
The corresponding DLAs will be denoted as $\aa_k^G$ and $\bb_l^G$.
For example, $\aa_0=\Span_{\mathbb R}\{iXX\}$, gives that $\aa^G_0$ is Abelian and has a basis over $\mathbb R$ given by a copy of $iXX$ acting on every edge of $G$.
Similarly, $\bb$-type Lie algebras can be analyzed easily: $\bb_0^G, \bb_1^G$ are Abelian, and $\bb_3^G$ only contains 1-qubit operators.

The DLAs that arise in this way strongly depend on whether
the graph $G$ is bipartite (BP) or non-bipartite (NBP),
and we will consider these two cases separately.
Recall that a graph is called \emph{bipartite} if its vertices can be colored in two colors
so that every edge connects vertices of different colors.

\ifTikz
\begin{figure}[htp!]
    \centering
    \includegraphics[width=\columnwidth]{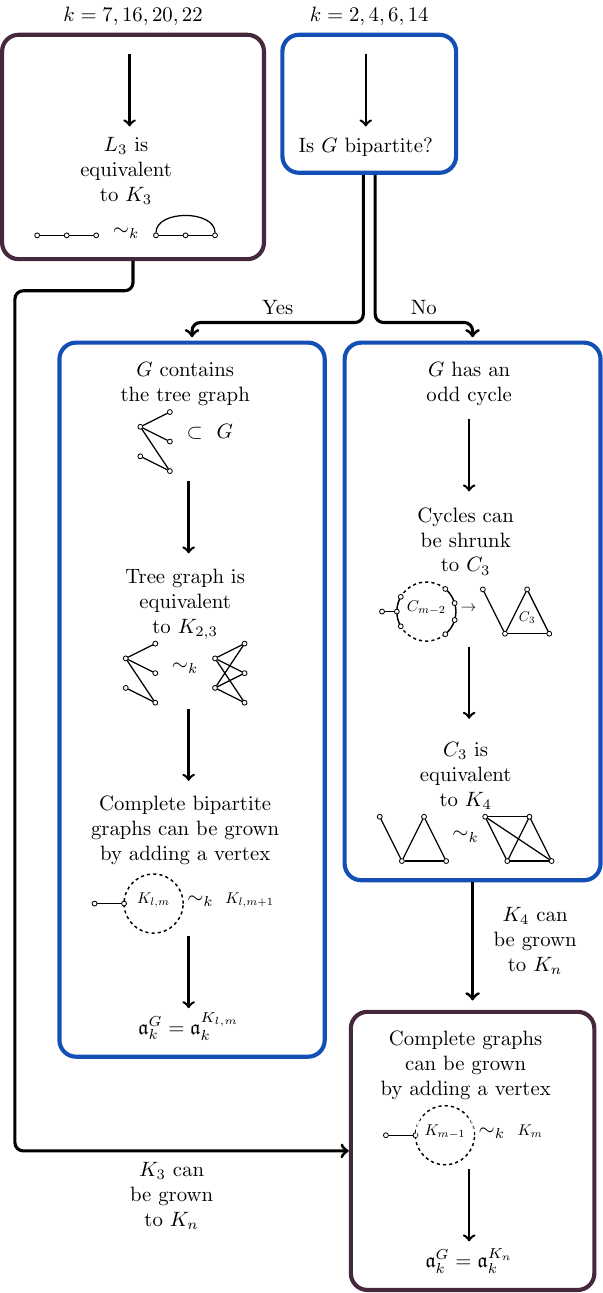}
    \caption{\textbf{Summary of the proof.} We consider two different cases. In the first case, $k=7,16,20,22$, we show that the graph is equivalent to the complete graph, and then use the results of~\cite{wiersema2023classification} to quickly identify the DLAs. In the second case,
    $k=2,4,6,14$, we analyze separately bipartite and non-bipartite graphs.$\vspace{5cm}$}
    \label{fig:summary}
\end{figure}
\fi

Our main result is as follows:
\begin{theorem}[Classification of dynamical Lie algebras of 2-local spin interactions on connected undirected graphs]\label{the:classification}
    For any connected undirected graph $G$ with $n$  
    vertices and $E$ edges, which has at least one vertex of degree $>2$, we have:
    \allowdisplaybreaks
    \begin{align*}
    \aa_0^G      &\cong	\uu(1)^{\oplus E},\\
    \aa_2^G&\cong  \begin{cases}
                \mathrm{BP}:
                \begin{cases}
                \su(2^{n-2})^{\oplus 2}, &\text{$l,m$ odd,} \\
                \so(2^{n-2})^{\oplus 4}, &\text{$l,m$ even,} \\
                \so(2^{n-1}), &\text{$n=l+m$ odd,} 
            \end{cases}\\
        \mathrm{NBP}:\so(2^{n-1})^{\oplus2}, 
    \end{cases}\\
    \aa_4^G&\cong \begin{cases}
         \mathrm{BP}:\aa_2^G,\\
         \mathrm{NBP}:\begin{cases} 
            \su(2^{n-1}),  & n \;\;\text{odd}, \\
            \su(2^{n-2})^{\oplus 4},  & n \;\;\text{even},
            \end{cases} \\
    \end{cases}\\
    \aa_6^G&\cong \begin{cases}
         \mathrm{BP}:\aa_7^G, \\
         \mathrm{NBP}:\su(2^{n-1})^{\oplus2}, \\
    \end{cases}\\
    \aa^G_7 &\cong \begin{cases} \su(2^{n-1}), & n \;\;\text{odd}, \\
        \su(2^{n-2})^{\oplus 4}, & n \;\;\text{even},
        \end{cases} \\
        \aa_{14}^G&\cong \begin{cases}
         \mathrm{BP}:\begin{cases}
            \sp(2^{n-2})^{\oplus 2}, &\text{$l,m$ odd,} \\
            \so(2^{n-1})^{\oplus 2}, &\text{$l,m$ even,} \\
            \su(2^{n-1}), &\text{$n=l+m$ odd,} \\
        \end{cases}\\
         \mathrm{NBP}:\aa_6^G, \\   
         \end{cases}      \\
    \aa^G_{16} &\cong \so(2^n),  \\
    \aa^G_{20} &\cong\su(2^{n-1})^{\oplus2}, \\
    \aa^G_{22} &= \su(2^n), \\
    \bb_0^G      &\cong	\uu(1)^{\oplus n},\\
    \bb_1^G      &\cong	\uu(1)^{\oplus (n+E)},\\
    \bb_3^G      &\cong	\su(2)^{\oplus n}.
    \end{align*}
    Here $l$ and $m$ represent the number of vertices in each of the two colors of the bipartite graph $G$.
\end{theorem}

A direct consequence \Cref{the:classification} is the following.

\begin{corollary}[Dimension scalings of the DLAs on an arbitrary graph] The dimension of the DLA generated by 2-local spin interactions on an arbitrary connected undirected graph, which has a vertex of degree $>2$,  will scale as $\mathcal{O}(n)$, $\mathcal{O}(n^2)$ or $\mathcal{O}(4^n)$ 
where $n$ is the number of vertices. 
The $\text{Poly}(n)$ scalings occur only when the DLA is Abelian $(\aa_0, \bb_0, \bb_1)$ or consists of only 1-qubit operators $(\bb_3)$. 
\end{corollary}
A consequence of this corollary is that the DLAs generated by 2-local interactions that can be simulated efficiently are limited to free fermion models~\cite{somma2005quantum, somma2006efficient}.

\subsection{Outline of the Proof}

The core ingredient of the proof of \Cref{the:classification} is the surprising observation that the 2-colorability of the graph $G$ completely determines the associated DLA $\aa_k^G$. As we show, if $G$ is not a line or cycle, the DLA is always isomorphic to either the DLA on the complete graph $K_n$ or the DLA on the complete bipartite graph $K_{l,m}$. 
The proof that $\aa_k^G \cong\aa_k^{K_n}$ or $\aa_k^G \cong\aa_k^{K_{l,m}}$ relies on two inductive techniques. The first technique involves showing that if the graph $K_{l,m}$ is connected to a single vertex, then the DLA on that graph will be equivalent to the DLA on the graph $K_{l+1,m}$ or $K_{l,m+1}$. The second technique allows us to go from a graph $K_{n-1}$ connected to a single vertex to the complete graph $K_n$. The proof then comes down to showing that for all $k$, one can generate the DLA of $K_{2,3}$ or $K_4$ with elements of $G$. The last step is to determine the DLAs $\aa_k^G$ for $G=K_{l,m}$ and $G=K_n$ for $k = 2,4,6,7,14,16,20$. We summarize the outline of the proof in \Cref{fig:summary}.

\section{Preliminaries}\label{sec:preliminaries}
\ifTikz
\begin{figure*}[htb!]
\begin{tikzpicture}[every node/.style={align=center,anchor=base},baseline=0pt]
    \node (A1) at (0,0.)[anchor=south] {
    \begin{tikzpicture}[scale=1.4] 
    \begin{scope}[every node/.style={circle,thick,draw,scale=0.5}]
        \node (A) at (0,0) {};
        \node (B) at (0.5,0) {};
        \node (C) at (1,0) {};
        \node (D) at (1.5,0) {};
    \end{scope}
    \begin{scope}[every edge/.style={draw=black,very thick}]
        \path [-] (A) edge (B);
        \path [-] (B) edge (C);
        \path [-] (C) edge (D);=
    \end{scope}
    \end{tikzpicture}
    };
    \node (A2) at (2.5,0)[anchor=south] {
    \begin{tikzpicture}
        \begin{scope}[every node/.style={circle, thick,draw, scale=0.5}]
            \draw[line width=.5mm] (0.75,0) arc (0:360:0.75);
            \foreach \i in {0,60,120,180,240,300}{
                \draw (\i:0.75) node[circle, fill=white,anchor=center] {};
            }
        \end{scope}
    \end{tikzpicture}
    };
    \node (A3) at (5.,0)[anchor=south] {
    \begin{tikzpicture}[scale=1.4] 
    \def\n{6}
    
    \def\r{0.65}

    \begin{scope}[every node/.style={circle, thick, draw, scale=0.5}]
    \foreach \i in {1,...,\n} {
        \pgfmathsetmacro\angle{360*\i/\n}
        \coordinate (P\i) at (\angle:\r);
        \node[fill=white] (C\i) at (P\i)[anchor=center] {};
    }
    \end{scope}
    \foreach \i in {1,...,\n} {
        \foreach \j in {1,...,\n} {
            \ifnum\i<\j
                \draw[thick] (C\i) -- (C\j);
            \fi
        }
    }
    \end{tikzpicture}
    };
    \node (Ag0) at (0, 2.) {};
    \node (Ag1) at (0,-0.5) {};
    
    \node (A4) at (8,0)[anchor=south] {
    \begin{tikzpicture}[scale=1.4] 
    \begin{scope}[every node/.style={circle,thick,draw,scale=0.5}]
        \node (A) at (0,0.) {};
        \node (B) at (0.5,0.0) {};
        \node (C) at (1,0) {};
        \node (D) at (0.33,1) {};
        \node (E) at (0.66,1) {};
        \node (F) at (1.,1) {};
        \node (G) at (0,1) {};
    \end{scope}
    \begin{scope}[every edge/.style={draw=black,very thick}]
        \path [-] (A) edge (D);
        \path [-] (A) edge (E);
        \path [-] (B) edge (F);
        \path [-] (B) edge (E);
        \path [-] (C) edge (D);
        \path [-] (C) edge (E);
        \path [-] (A) edge (G);
    \end{scope}
    \end{tikzpicture}
    };
    \node (A5) at (11.0,0)[anchor=south]{
    \begin{tikzpicture}[scale=1.4] 
    \begin{scope}[every node/.style={circle,thick,draw,scale=0.5}]
        \node (A) at (0,0.) {};
        \node (B) at (0.3,0.5) {};
        \node (C) at (1.6,0) {};
        \node (D) at (0.2,1) {};
        \node (E) at (0.6, 0.7) {};
        \node (F) at (1.1,0.9) {};
        \node (G) at (1, 0.7) {};
    \end{scope}
    \begin{scope}[every edge/.style={draw=black,very thick}]
        \path [-] (A) edge (B);
        \path [-] (A) edge (D);
        \path [-] (B) edge (C);
        \path [-] (C) edge[bend left=20] (G);
        \path [-] (D) edge (E);
        \path [-] (E) edge (F);
        \path [-] (G) edge (E);
        \path [-] (G) edge (F);
        \path [-] (F) edge (C);
        \path [-] (E) edge (B);
    \end{scope}
    \end{tikzpicture}
    };
    \node (A6) at (14,0)[anchor=south]{
    \begin{tikzpicture}[scale=1.4] 
    \begin{scope}[every node/.style={circle,thick,draw,scale=0.5}]
        \node (A) at (0,0.) {};
        \node (B) at (0.5,0.0) {};
        \node (C) at (1,0) {};
        \node (D) at (0.33,1) {};
        \node (E) at (0.66,1) {};
    \end{scope}
    \begin{scope}[every edge/.style={draw=black,very thick}]
        \path [-] (A) edge (D);
        \path [-] (A) edge (E);
        \path [-] (B) edge (D);
        \path [-] (B) edge (E);
        \path [-] (C) edge (D);
        \path [-] (C) edge (E);
    \end{scope}
    \end{tikzpicture}
    };
    \node (Ag2) at (6.5,-0.5) {};
    \node (Ag3) at (6.5, 2.) {};
    \node[below] at (A1.south) {(a) $L_4$};
    \node[below] at (A2.south) {(b) $C_6$};
    \node[below] at (A3.south) {(c) $K_6$};
    \node[below] at (A4.south) {(d) Bipartite $G$};
    \node[below] at (A5.south) {(e) Non-bipartite $G$};
    \node[below] at (A6.south) {(f) $K_{2,3}$};

    \begin{scope}[every node/.style={inner sep=1pt,line width=2pt, rounded corners=8pt}]
        \node[draw, fit=(A1) (A2) (A3) (Ag0) (Ag1), color=oldgraph] (border1) {};
        \node[draw, fit=(A4) (A5) (A6) (Ag2) (Ag3), color=fgraph] (border2) {};
    \end{scope}
    \node[below] at (border1.south) {Previous work};
    \node[below] at (border2.south) {This work};
\end{tikzpicture}
    \caption{\textbf{Examples of interaction graphs $\boldsymbol{G}$. }(a-c) The line, cycle and complete graph were the focus of our work in~\cite{wiersema2023classification}. (d-e) We extend our classification to arbitrary graphs, which will require studying the DLAs of complete bipartite graphs $K_{l,m}$.}
    \label{fig:graph_examples}
\end{figure*}
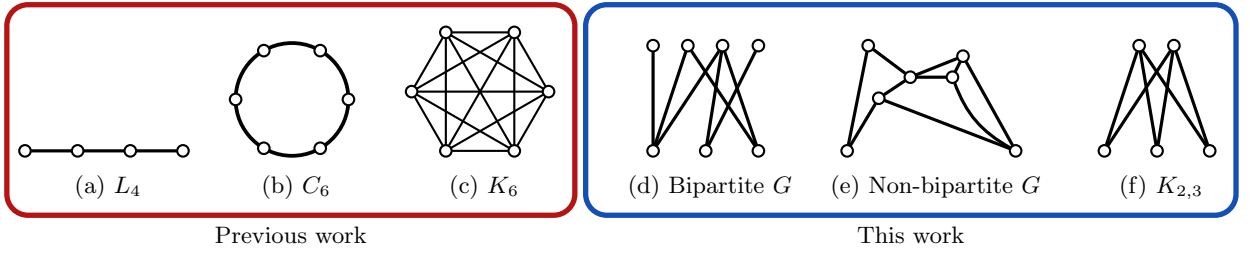
\fi
Consider a Hamiltonian 
on $n$ qubits, which is a sum of Hermitian terms,
\begin{align*}
    H = \sum_{m=1}^M a_m, \qquad a_m\in i\,\uu(2^n).
\end{align*}
Note that the space $\uu(2^n)$ of all skew-Hermitian $2^n\times 2^n$ matrices is a Lie algebra under the commutator
$[a,b] = ab - ba$.
The elements $ia_m\in \uu(2^n)$ generate a Lie subalgebra of $\uu(2^n)$, called the \emph{dynamical Lie algebra} (DLA) ~\cite{dalessandro2021, Albertini2001dynlie}; this is the smallest (under inclusion) subalgebra containing $ia_1, \ldots, ia_{M}$ (see, e.g.,~\cite[Supplemental Material A II]{wiersema2023classification} for more details).

We will denote the DLA generated by the set $\mcA=\{ a_1,\ldots, a_{M}\}$ as $\Lie{\mcA}$. Explicitly, $\Lie{\mcA}$ consists of all real linear combinations of nested commutators of elements of $i\mcA$:
\begin{equation}\label{eq:adj}
\begin{split}
\ad_{ia_{m_1}} &\cdots \ad_{ia_{m_r}} (ia_{m_{r+1}}) \\
&= i^{r+1} [a_{m_1},[a_{m_2},[\cdots[a_{m_r},a_{m_{r+1}}]\cdots]]].
\end{split}
\end{equation}
Here $r\ge0$, $1\le m_1,\dots,m_{r+1} \le M$, and we used the standard notation $\ad_a(b)= [a,b]$.

We will assume that the terms of the Hamiltonian $H$ correspond to 2-local spin interactions that are determined by an undirected graph $G$ with $n$ vertices. More precisely, we consider a subset $\mcA_2 \subset \{I,X,Y,Z\}^{\otimes 2}$ of length-2 Pauli strings, not containing $I\otimes I$ (cf.\ Appendix \ref{secpauli1}). Then $H$ will have the form
\begin{align}
    H= \sum_{(i,j)\in G} \sum_{a\otimes a'\in \mcA_2} J_{i,j,a,a'} a_{i}\otimes a'_{j}, \label{eq:hamiltonian}
\end{align}
where the first sum is over the edges of $G$,
and $J_{i,j,a,a'}$
are arbitrary real coefficients that are possibly time-dependent. 
In Eq.\ \eqref{eq:hamiltonian}, $a_i$ denotes the action of the matrix $a\in\{I,X,Y,Z\}$ on the $i$-th qubit, i.e.,
on the $i$-th factor in the tensor power $(\mathbb{C}^2)^{\otimes n}$.

Observe that if we replace the generating set $\mcA_2$ by the Lie subalgebra $\Lie{\mcA_2} \subset\su(4)$, then the DLA generated by the terms of $H$  remains the same. Hence, we can make use of the classification of all Lie subalgebras of $\su(4)$ that are generated by length-2 Pauli strings,
due to~\cite[Supplemental Material B I]{wiersema2023classification}.
As the graph $G$ is \emph{undirected}, for every edge $(i,j)\in G$, we also have an edge $(j,i)\in G$.
Because of this symmetry, we will assume that the
set $\mcA_2$ is \emph{symmetric} under the flip of the two qubits, i.e., $XY \in\mcA_2$ if and only if $YX \in\mcA_2$. Here and further, we use the shorthand notation $XY:=X\otimes Y$.
In particular, we assume that
whenever we have an element of the form $XI\in\mcA_2$, we also have $IX \in\mcA_2$ and vice versa, so that
we get an action of $X$ on each qubit. Thus, we can exclude the $\cc$-type Lie algebras from~\cite{wiersema2023classification}. 

By inspecting the bases of all Lie subalgebras of $\su(4)$ of $\aa$ and $\bb$ type, given in Appendix \ref{sec:list_dla}, we find that the list
of symmetric subalgebras consists of
$\aa_k$ for $k = 0,2,4,6,7,14,16,20,22$, and $\bb_l$ for $l = 0,1,3$.
The resulting DLAs, generated by the terms of a Hamiltonian $H$ of the form \eqref{eq:hamiltonian}, will be denoted as $\aa_k^G$ and $\bb_l^G$.
Explicitly, we have:
\allowdisplaybreaks
\begin{align*}
\aa_0       =&	\Lie{XX},\\
\aa_2       =&	\Lie{XY, YX},	\\
\aa_4       =&	\Lie{XX, YY},\\
\aa_6   	=&	\Lie{XX, YZ, ZY}, \\
\aa_7   	=&	\Lie{XX, YY, ZZ}, \\
\aa_{14}	=&	\Lie{XX, YY, XY, YX},\\
\aa_{16}	=&	\Lie{XY, YX, YZ, ZY}, \\
\aa_{20}	=&	\Lie{XX, YY, YZ, ZY}, \\	
\aa_{22}    =&	\Lie{XX, XY, YX, XZ, ZX} = \su(4), \\
\bb_0       =&	\Lie{XI, IX},\\
\bb_1       =&	\Lie{XX, XI, IX},\\
\bb_3       =&	\Lie{XI, YI, IX, IY}.
\end{align*}
Note that some of Lie algebras $\aa_k$ can be generated by non-symmetric sets, but we have provided symmetric generators for each of them.
For example, $\aa_{14}$ can be generated by the minimal set $\{XX, YY, XY\}$, and also by the symmetric set $\{XX, YY, XY, YX\}$.

In~\cite{wiersema2023classification}, we classified the DLAs generated by 2-local spin Hamiltonians on line, cycle and complete graphs. 
The notation of~\cite{wiersema2023classification} compares to our present notation
as follows (see \Cref{fig:graph_examples}):
\begin{enumerate}
    \item Line graph $L_n$: $\aa_k(n) = \aa_k^{L_n}$;
    \item Cycle graph $C_n$: $\aa_k^\circ(n) = \aa_k^{C_n}$;
    \item Complete graph $K_n$: $\aa_k^\pi(n) = \aa_k^{K_n}$.
\end{enumerate}

\begin{example}
For the Lie algebra $\aa_0=\Lie{XX}$, example Hamiltonians on $L_n$, $C_n$, $K_n$, and arbitrary graph $G$ are given respectively by:
\begin{align*}
H=&	\sum_{i=1}^{n-1} X_i X_{i+1},\\
H=&	\sum_{i=1}^n X_i X_{i+1}, \qquad X_{n+1}:=X_1, \\
H=&	\sum_{1 \le i<j \le n} X_i X_j, \\
H=&	\sum_{(i,j) \in G} X_i X_j.
\end{align*}
\end{example}

\section{Graph reduction}
From now on, we consider the values $k = 2,4,6,7,14,16,20,22$.
The approach of~\cite{wiersema2023classification} was to determine $\aa_k(n)= \aa_k^{L_n}$ inductively by increasing the number $n$ of sites in the system. The analogue of this procedure for a general graph $G$ will be to add new vertices and edges to the graph.
This motivates the following definition.

\begin{definition}\label{def:Gsub}
    We say that a graph $G$ is a \emph{subgraph} of a graph $G'$ if all the vertices and edges in $G$ are also in $G'$. We denote this by $G \subset G'$. 
\end{definition}

It is obvious from the definitions that
\begin{equation}\label{eq:order}
    G \subset G' \implies \aa_k^G \subseteq \aa_k^{G'}.
\end{equation}
In particular, $G\subset K_n$ for any graph $G$ with $n$ vertices, which gives the upper bound $\aa_k^G \subseteq \aa_k^{K_n}$.
Similarly, if $G$ is bipartite, with $l$ and $m$ vertices in each of the two colors, then $G\subset K_{l,m}$ and $\aa_k^G \subseteq \aa_k^{K_{l,m}}$.
Our goal will be to show that these upper bounds are in fact equalities, provided that $G$ has at least one vertex of degree $>2$.

It will be convenient to introduce the following notion.
\begin{definition}\label{def:Geq}
    Two graphs $G$ and $G'$ are called \emph{$k$-equivalent} if and only if $\aa_k^G = \aa_k^{G'}$. We denote this equivalence as $G \sim_k G'$. 
\end{definition}
Thus, we want to show that for all $k$ the graph $G$ is $k$-equivalent to either $K_n$ or $K_{l,m}$. 
An important observation is that the notion of equivalence can be applied locally to subgraphs. More precisely, we have:

\begin{lemma}\label{lem:loc}
Consider subgraphs $G_1\subset G_2$ and $G'_1\subset G'_2$, such that $G_i$ and $G'_i$ have the same set of vertices $V_i$ for $i=1,2$. Suppose that:
\begin{enumerate}
    \item All edges of\/ $G_1$ and $G_2$ connecting two vertices from $V_1$ are the same;
    \item All edges of\/ $G'_1$ and $G'_2$ connecting two vertices from $V_1$ are the same;
    \item All edges of\/ $G_2$ and $G'_2$ that have at least one vertex outside $V_1$ are the same.
\end{enumerate}
Then $G_1 \sim_k G'_1$ implies $G_2 \sim_k G'_2$.
\end{lemma}
\begin{proof}
By construction, $\aa_k^{G}$ is generated as a Lie algebra by $a_i \otimes b_j$ for every $a\otimes b\in\mcA_2$ (the generating set of $\aa_k$) and every edge $(i,j)$ of $G$. Let us denote the generators of $\aa_k^{G_i}$ and $\aa_k^{G'_i}$ by $\mcB_i$ and $\mcB'_i$, respectively ($i=1,2$), so that $\aa_k^{G_i} = \Lie{\mcB_i}$ and $\aa_k^{G'_i} = \Lie{\mcB'_i}$.
By assumption, we have
\begin{align*}
\mcB_2 = \mcB_1 \cup \mcB'' , \quad \mcB'_2 = \mcB'_1 \cup \mcB'' ,
\end{align*}
where $\mcB''$ corresponds to edges that are not within $V_1$; such edges are common for $G_2$ and $G'_2$.
The claim now follows, because the process of generating a Lie subalgebra from a subset is transitive:
\begin{align*}
\Lie{\mcB_1 \cup \mcB''} &= \Lie{\Lie{\mcB_1}\cup\mcB''} \\
&= \Lie{\Lie{\mcB'_1}\cup\mcB''} = \Lie{\mcB'_1 \cup \mcB''},
\end{align*}
where we have used that $\Lie{\mcB_1} = \Lie{\mcB'_1}$ by the assumption $G_1 \sim_k G'_1$.
\end{proof}

In the following, we will consider separately two cases for~$k$.

\subsection{The Case $k=7,16,20,22$}
From~\cite[Theorem IV.3]{wiersema2023classification}, we know that for these values of $k$ and $n\geq 3$ we have $\aa_k^{L_n} = \aa_k^{K_n}$. 
Hence, if a graph $G$ with $n$ vertices has a Hamiltonian path (i.e., some $L_n \subset G$), then we obtain from Eq.\ \eqref{eq:order} that $\aa_k^{G} = \aa_k^{K_n}$. We will prove that the last equality holds without any assumptions.

\begin{theorem}\label{thm:complete}
    For $k = 7, 16, 20, 22$ and any connected graph $G$ with $n\geq 3$ vertices, we have $\aa_k^{G} = \aa_k^{K_n}$, i.e.,
    $G \sim_k K_n$.
\end{theorem}

\begin{proof}
For $n=3$, the only connected graphs with $3$ vertices are $L_3$ and $K_3$. As a special case of~\cite[Theorem IV.3]{wiersema2023classification}, we have $L_3 \sim_k K_3$:

    \ifTikz
    \begin{align*}
    \begin{tikzpicture}[baseline=0mm, scale=2] 
    \begin{scope}[every node/.style={circle,thick,draw,scale=0.5}]
        \node (A) at (0.0,0.0) {};
        \node (B) at (0.5,0.0) {};
        \node (C) at (1.0,0.0) {};
    \end{scope}
    \begin{scope}[every edge/.style={draw=black,very thick}]
        \path [-] (A) edge (B);
        \path [-] (B) edge (C);
    \end{scope}
    \end{tikzpicture} 
    \quad \adjustbox{scale=2}{$\sim_k$} \quad 
    \begin{tikzpicture}[baseline=0mm, scale=2]        \begin{scope}[every node/.style={circle,thick,draw,scale=0.5}]
        \node (A) at (0.0,0.0) {};
        \node (B) at (0.5,0.0) {};
        \node (C) at (1.0,0.0) {};
    \end{scope}
    \begin{scope}[every edge/.style={draw=black,very thick}]
        \path [-] (A) edge (B);
        \path [-] (B) edge (C);
        \path [-] (C) edge[bend right=90] (A);
    \end{scope}
    \end{tikzpicture}
    \end{align*}
    \fi
    For $n>3$, since $G$ is connected, we can find three vertices $x_1,x_2,x_3$ that form $L_3\subset G$, so that $G$ has edges $(x_1,x_2)$ and $(x_2,x_3)$. 
    Then Lemma \ref{lem:loc} and $L_3\sim_k K_3$ imply that, if $G$ does not already have an edge $(x_1,x_3)$, we can add it and obtain an equivalent graph:
    \ifTikz
    \begin{align*}
    \begin{tikzpicture}[baseline=3mm]
    \begin{scope}[every node/.style={circle,thick,draw,scale=0.5}]
        \node (A) at (0,0.) {};
        \node (B) at (1,0) {};
        \node (C) at (0,1) {};
        \node (D) at (1,1) {};
        \node (E) at (1.5,1.25) {};
        \node (F) at (-0.5,.75) {};
    \end{scope}
    \path [-] (A) edge[very thick, color=highlight] (B);
    \path [-] (C) edge[very thick] (D);
    \path [-] (D) edge[very thick] (B);
    \path [-] (E) edge[very thick] (B);
    \path [-] (F) edge[very thick] (C);
    \path [-] (C) edge[very thick, color=highlight] (B);
    \node[left] at (A) {$x_1$};
    \node[right] at (B) {$x_2$};
    \node[above] at (C) {$x_3$};
    \end{tikzpicture}
    \quad \adjustbox{scale=2}{$\sim_k$} \quad 
    \begin{tikzpicture}[baseline=3mm]
        \begin{scope}[every node/.style={circle,thick,draw,scale=0.5}]
            \node (A) at (0,0.) {};
            \node (B) at (1,0) {};
            \node (C) at (0,1) {};
            \node (D) at (1,1) {};
            \node (E) at (1.5,1.25) {};
            \node (F) at (-0.5,.75) {};
        \end{scope}
        \path [-] (A) edge[very thick, color=highlight] (B);
        \path [-] (C) edge[very thick] (D);
        \path [-] (D) edge[very thick] (B);
        \path [-] (E) edge[very thick] (B);
        \path [-] (F) edge[very thick] (C);
        \path [-] (C) edge[very thick, color=highlight] (B);
        \path [-] (C) edge[very thick, color=highlight] (A);
        \node[left] at (A) {$x_1$};
        \node[right] at (B) {$x_2$};
        \node[above] at (C) {$x_3$};
        \node[above] at (F) {};
        \end{tikzpicture}
    \end{align*}
    \fi

The idea of the proof is to keep doing this for any subgraph $L_3\subset G$, adding more and more edges to $G$ until it becomes a complete graph.
In more detail, as $G$ is connected,
    we can find a vertex $x_4$ in $G$ that is connected to at least one of $x_1,x_2,x_3$. 
    Without loss of generality, suppose that $x_4$ is connected to $x_3$. Then from the line $\{x_1,x_3,x_4\}$, we see that we can add the edge $(x_1,x_4)$ (if it is not already in $G$) to obtain an equivalent graph:
    \ifTikz
    \begin{align*}
    \begin{tikzpicture}[baseline=3mm]
        \begin{scope}[every node/.style={circle,thick,draw,scale=0.5}]
            \node (A) at (0,0.) {};
            \node (B) at (1,0) {};
            \node (C) at (0,1) {};
            \node (D) at (1,1) {};
            \node (E) at (1.5,1.25) {};
            \node (F) at (-0.5,.75) {};
        \end{scope}
        \path [-] (A) edge[very thick] (B);
        \path [-] (C) edge[very thick] (D);
        \path [-] (D) edge[very thick] (B);
        \path [-] (E) edge[very thick] (B);
        \path [-] (C) edge[very thick] (B);
        \path [-] (C) edge[very thick, color=highlight] (A);
        \path [-] (C) edge[very thick, color=highlight] (F);
        \node[left] at (A) {$x_1$};
        \node[right] at (B) {$x_2$};
        \node[above] at (C) {$x_3$};
        \node[above] at (F) {$x_4$};
    \end{tikzpicture}
    \quad \adjustbox{scale=2}{$\sim_k$} \quad 
    \begin{tikzpicture}[baseline=3mm]
        \begin{scope}[every node/.style={circle,thick,draw,scale=0.5}]
            \node (A) at (0,0.) {};
            \node (B) at (1,0) {};
            \node (C) at (0,1) {};
            \node (D) at (1,1) {};
            \node (E) at (1.5,1.25) {};
            \node (F) at (-0.5,.75) {};
        \end{scope}
        \path [-] (A) edge[very thick] (B);
        \path [-] (C) edge[very thick] (D);
        \path [-] (D) edge[very thick] (B);
        \path [-] (E) edge[very thick] (B);
        \path [-] (C) edge[very thick] (B);
        \path [-] (C) edge[very thick, color=highlight] (A);
        \path [-] (C) edge[very thick, color=highlight] (F);
        \path [-] (A) edge[very thick, color=highlight] (F);
        \node[left] at (A) {$x_1$};
        \node[right] at (B) {$x_2$};
        \node[above] at (C) {$x_3$};
        \node[above] at (F) {$x_4$};
    \end{tikzpicture}
    \end{align*}
    \fi
    Similarly, we can add the edge $(x_2,x_4)$. Now the four vertices $\{x_1,x_2,x_3,x_4\}$ form a complete subgraph $K_4\subset G$.
    We repeat this process until we reach $K_n$; the induction is formalized in Lemma \ref{lem:kn1kn} below. This completes the proof of the theorem.
\end{proof}

\begin{lemma}\label{lem:kn1kn}
    For $k = 7, 16, 20, 22$ and all $m\ge 2$, we have:
    \ifTikz
    \begin{align*}
    \begin{tikzpicture}[baseline=0mm]
        \begin{scope}[every node/.style={circle, thick,draw, scale=0.5}]
            \node[circle] (A) at (-2,0) {};
            \path [-] (A) edge[draw=black,very thick] (-1,0);
            \draw[line width=.5mm, dash pattern=on 3pt off 3pt] (1,0) arc (0:360:1);
            \node[circle, fill=white] (B) at (-1,0) {};
            \node[draw=white] (C) at (0,0) {};
            \node[draw=white, scale=3] at (C) {$K_{m-1}$};
        \end{scope}
    \end{tikzpicture}
    \quad \adjustbox{scale=2}{$\sim_k$} 
    \begin{tikzpicture}[baseline=0mm]
        \begin{scope}[every node/.style={circle, thick,draw, scale=0.5}]
            \node[draw=white, scale=3] at (C) {$K_{m}.$};
        \end{scope}
    \end{tikzpicture}
	\end{align*}
    \fi
\end{lemma}
\begin{proof}
    For $m = 2$, the two graphs are the same. 
    For $m\ge3$, let us label the leftmost vertex as $x$ and the vertex in $K_{m-1}$ connected to it as $y$.    
    Then for any other vertex $z$ in $K_{m-1}$, the vertices $x, y,z$ form a line $L_3$.
By Lemma \ref{lem:loc} and $L_3\sim_k K_3$, we can add the edge $(x,z)$ to get an equivalent graph:
    \ifTikz
    \begin{align*}
        \begin{tikzpicture}[baseline=0mm]
        \begin{scope}[every node/.style={circle, thick,draw, scale=0.5}]
            \node[circle] (A) at (-2,0) {};
            \path [-] (A) edge[very thick, color=highlight] (-1,0);\node[draw=white] (C) at (0,0) {};
			\node[draw=white, scale=2.5] at (C) {$K_{m-1}$};
            \draw[line width=.5mm, dash pattern=on 3pt off 3pt] (-150:1) arc (-150:120:1);
            \draw[line width=.5mm] (120:1) arc (120:210:1);
            \path [-] (60:1) edge[very thick, color=highlight] (180:1);
            \foreach \i in {60, 120, 150, 180, 210}{
                \draw (\i:1) node[circle, fill=white] {};
            }
        \end{scope}
        \node[above] at ($(120:1) + (-0.1, 0.05)$) {};
        \node[above] at ($(150:1) + (-0.1, 0.075)$) {};
        \node[below] at ($(180:1) - (0.125, 0.075)$) {$y$};
        \node[above] at ($(210:1) - (0.6, 0.4)$) {};
        \node[above] at ($(60:1) + (0.1, 0.1)$) {$z$};
        \node[below] at (A.south) {$x$};
        \end{tikzpicture}
        \quad \adjustbox{scale=2}{$\sim_k$} \quad 
        \begin{tikzpicture}[baseline=0mm]
        \begin{scope}[every node/.style={circle, thick,draw, scale=0.5}]
            \node[draw=white] (C) at (0,0) {};
			\node[draw=white, scale=2.5] at (C) {$K_{m-1}$};
            \path [-] (60:1) edge[very thick, color=highlight] (-2,0);
            \node[circle, fill=white] (A) at (-2,0) {};
            \path [-] (A) edge[very thick, color=highlight] (-1,0);
            \draw[line width=.5mm, dash pattern=on 3pt off 3pt] (-150:1) arc (-150:120:1);
            \draw[line width=.5mm] (120:1) arc (120:210:1);
            \path [-] (60:1) edge[very thick, color=highlight] (180:1);
            \foreach \i in {60, 120, 150, 180, 210}{
                \draw (\i:1) node[circle, fill=white] {};
            }
        \end{scope}
        \node[above] at ($(120:1) + (-0.1, 0.05)$) {};
        \node[above] at ($(150:1) + (-0.1, 0.075)$) {};
        \node[below] at ($(180:1) - (0.125, 0.075)$) {$y$};
        \node[above] at ($(210:1) - (0.6, 0.4)$) {};
        \node[above] at ($(60:1) + (0.1, 0.1)$) {$z$};
        \node[below] at (A.south) {$x$};
        \end{tikzpicture}   
    \end{align*}
    \fi
    In this way, we can connect $x$ to all vertices of $K_{m-1}$ and obtain $K_m$.
\end{proof}

From \Cref{thm:complete} and~\cite[Theorem IV.3]{wiersema2023classification}
(see also Appendix \ref{sec:list_dla}), we deduce the cases $k=7,16,20,22$ in \Cref{the:classification}.

\subsection{The Case \texorpdfstring{$k=2,4,6,14$}{}}
The second case for $k$ is more involved. We will further split it into two cases: $G$ bipartite or $G$ non-bipartite. 
Recall that a graph is \emph{bipartite} if its vertices can be colored in two colors
so that every edge connects vertices of different colors. 
Equivalently, a graph is bipartite if and only if it does not contains any odd cycles.
Since we have already analyzed line and cycle graphs in~\cite{wiersema2023classification},
from now on we will assume that $G$ has at least one vertex of degree $>2$.

\subsubsection{Non-bipartite case}
Here we assume that the interaction graph $G$ is not bipartite, i.e., it contains an odd cycle.
We present a sequence of lemmas, which will allow us to simplify certain subgraphs of $G$.

\begin{lemma}\label{lem:deltak23}
    For $k = 2, 4, 6, 14$, the following graphs are $k$-equivalent:
    \ifTikz
    \begin{align*}
    \adjustbox{scale=1.5}{$\Sigma:=$} \quad
    \begin{tikzpicture}[baseline=0mm, scale=1] 
    \begin{scope}[every node/.style={circle,thick,draw,scale=0.5}]
        \node (A) at (1,1.) {};
        \node (B) at (0.0,0.5) {};
        \node (C) at (0,-0.5) {};
        \node (D) at (1, -1) {};
        \node (E) at (1, 0.) {};
    \end{scope}
    \node[right] at (A.east) {$1$};
    \node[above] at (B.north) {$2$};
    \node[above] at (C.north) {$4$};
    \node[right] at (D.east) {$3$};
    \node[right] at (E.east) {$5$};
    \begin{scope}[every edge/.style={draw=black,very thick}]
        \path [-] (A) edge (B);
        \path [-] (C) edge (D);
        \path [-] (B) edge (D);
        \path [-] (E) edge (B);
    \end{scope}
    \end{tikzpicture}
    \quad\adjustbox{scale=1.5}{$\sim_k$} \quad
    \begin{tikzpicture}[baseline=0mm, scale=1] 
    \begin{scope}[every node/.style={circle,thick,draw,scale=0.5}]
        \node (A) at (1,1.) {};
        \node (B) at (0.0,0.5) {};
        \node (C) at (0,-0.5) {};
        \node (D) at (1, -1) {};
        \node (E) at (1, 0.) {};
    \end{scope}
    \node[right] at (A.east) {$1$};
    \node[above] at (B.north) {$2$};
    \node[above] at (C.north) {$4$};
    \node[right] at (D.east) {$3$};
    \node[right] at (E.east) {$5$};
    \begin{scope}[every edge/.style={draw=black,very thick}]
        \path [-] (A) edge (B);
        \path [-] (C) edge (D);
        \path [-] (B) edge (D);
        \path [-] (E) edge (B);
        \path [-] (A) edge (C);
         \path [-] (C) edge (E);
    \end{scope}
    \end{tikzpicture}
    \quad \adjustbox{scale=1.5}{$=K_{2,3}$}.
    \end{align*}
    \fi
\end{lemma}

\begin{proof}

We consider each value of $k$ separately. 

\paragraph*{Case $k = 2$.} The generators of $\aa_2$ are $XY, YX$.
One checks that the following nested commutator of the generators of $\aa_2^\Sigma$ is equal to $X_1 Y_4$
up to a nonzero scalar:
\begin{equation}\label{eq:X1Y4}
\begin{split}
    X_1 Y_4 \equiv\,& \ad_{Y_1 X_2} \ad_{X_3Y_2} \ad_{X_5Y_2} \ad_{Y_1 X_2} \\
    &\ad_{X_2 Y_3} \ad_{X_5Y_2} \ad_{X_3 Y_4} \ad_{X_2 Y_3} (X_1 Y_2).
\end{split}
\end{equation}
By swapping $X \leftrightarrow Y$ on every vertex, we obtain the recipe of how to construct $X_4Y_1$ too. 
This means that $\aa_2^\Sigma$ contains the generators of $\aa_2$ corresponding to the edge $(1,4)$; hence, if add this edge we get the same DLA. Due to symmetry of the graph $\Sigma$, we can also add the edge $(4,5)$, which gives $\Sigma \sim_2 K_{2,3}$. See also Appendix \ref{frust-eq4} for an alternative derivation of Eq.~\eqref{eq:X1Y4} using frustration graphs.

\paragraph*{Case $k = 4$.} 
The generators of $\aa_4$ are $XX, YY$.
Both graphs $\Sigma$ and $K_{2,3}$ are bipartite, with the first color consisting of vertices $\{2,4\}$. If we swap 
$X_2 \leftrightarrow Y_2$ and $X_4 \leftrightarrow Y_4$, the generators of $\aa_2^\Sigma$ transform into the generators of $\aa_4^\Sigma$, and similarly with $K_{2,3}$ in place of $\Sigma$. Therefore, $\aa_4^\Sigma \cong \aa_2^\Sigma = \aa_2^{K_{2,3}} \cong \aa_4^{K_{2,3}}$.

\paragraph*{Case $k = 6$.} The generators of $\aa_6$ are $XX, YZ, ZY$, while the generators of $\aa_7$ are $XX, YY, ZZ$.
As above, if we swap $Y_2 \leftrightarrow Z_2$ and $Y_4 \leftrightarrow Z_4$, we can transform the generators of
$\aa_6^\Sigma$ into those of $\aa_7^\Sigma$, and similarly with $K_{2,3}$ in place of $\Sigma$. Therefore, $\aa_6^\Sigma \cong \aa_7^\Sigma$ and  $\aa_6^{K_{2,3}} \cong \aa_7^{K_{2,3}}$.
But $\aa_7^\Sigma = \aa_7^{K_{2,3}} = \aa_7^{K_5}$, by \Cref{thm:branch_complete}, which gives
$\aa_6^\Sigma = \aa_6^{K_{2,3}}$.

\paragraph*{Case $k = 14$.} It is easy to see that $\aa_{14} = \Lie{XY, YX, ZI, IZ}$. For any interaction graph $G$,
we can generate $\aa_{14}^G$ by placing $XY,YX$ on every edge and $Z$ on every vertex. 
Since $\aa_2\subset\aa_{14}$, from the $k=2$ case above we get that
$X_1 Y_4, X_4Y_1, X_4 Y_5, Y_4 X_5 \in  \aa^{K_{2,3}}_2 = \aa^\Sigma_2 \subset \aa^\Sigma_{14}$.
Hence, if we add the edges $(1,4)$ and $(4,5)$, we will get the same DLA:
$\aa^\Sigma_{14} = \aa^{K_{2,3}}_{14}$.
In Appendix \ref{frust-eq4}, we give an alternative proof using
frustration graphs that the edge $(1,4)$ can be added to the interaction graph $\Sigma$
without changing the DLA. 

This concludes the proof of the lemma.
\end{proof}

\begin{lemma}\label{lem:k3plusk4}
For $k = 2, 4, 6, 14$, the following graphs are $k$-equivalent:
\ifTikz
\begin{align*}
\adjustbox{scale=1.5}{$\Omega\: :=$}&\quad
\begin{tikzpicture}[baseline=0mm, scale=0.75]
	\begin{scope}[every node/.style={circle, thick,draw, scale=0.5}]
        \node[circle, fill=white] (A) at (-2,1) {};
        \node[circle, fill=white] (B) at (-1,-1) {};
		\node[circle, fill=white] (C) at (0.,1) {};
		\node[circle, fill=white] (D) at (1,-1) {};
        \path [-] (A) edge[draw=black,very thick] (B);
        \path [-] (B) edge[draw=black,very thick] (C);
        \path [-] (C) edge[draw=black,very thick] (D);
        \path [-] (D) edge[draw=black,very thick] (B);
	\end{scope}
    \node[left] at (A.west) {$3$};
    \node[left] at (B.west) {$2$};
    \node[right] at (C.east) {$1$};
    \node[right] at (D.east) {$4$};
\end{tikzpicture} \\
\adjustbox{scale=1.5}{$\sim_k$}&\quad
\begin{tikzpicture}[baseline=0mm, scale=0.75]
	\begin{scope}[every node/.style={circle, thick,draw, scale=0.5}]
        \node[circle, fill=white] (A) at (-2,1) {};
        \node[circle, fill=white] (B) at (-1,-1) {};
		\node[circle, fill=white] (C) at (0.,1) {};
		\node[circle, fill=white] (D) at (1,-1) {};
        \path [-] (A) edge[draw=black,very thick] (B);
        \path [-] (B) edge[draw=black,very thick] (C);
        \path [-] (C) edge[draw=black,very thick] (D);
        \path [-] (D) edge[draw=black,very thick] (B);
        \path [-] (D) edge[draw=black,very thick] (A);
        \path [-] (C) edge[draw=black,very thick] (A);
	\end{scope}
    \node[left] at (A.west) {$3$};
    \node[left] at (B.west) {$2$};
    \node[right] at (C.east) {$1$};
    \node[right] at (D.east) {$4$};
\end{tikzpicture} 
\quad\adjustbox{scale=1.5}{$=\:K_4$.}
\end{align*}
\fi
\end{lemma}

\begin{proof}
The proof is similar to that of Lemma \ref{lem:deltak23}.
For each value of $k$, we will show that we can add the edge $(1,3)$; then the edge $(3,4)$ can be added by symmetry.

\paragraph*{Case $k = 2$.} Generators of $\aa_2$ are $XY, YX$. 
One checks that, up to a nonzero scalar,
\begin{equation}\label{eq:Y1X3}
    Y_1 X_3 \equiv \ad_{X_1 Y_2} \ad_{X_2 Y_4} \ad_{X_1 Y_4} \ad_{Y_2 X_3} (Y_1 X_2 ).
\end{equation} 
Thus, $Y_1 X_3 \in \aa^\Omega_2$. Due to the $X \leftrightarrow Y$ symmetry of the generators, we also have $X_1 Y_3 \in \aa^\Omega_2$. 
Hence, $\Omega \sim_2 K_4$.

\paragraph*{Case $k = 4$.} The generators of $\aa_4$ are $XX, YY$. The following shows that $X_1 X_3 $ can be obtained from the generators on the edges of $\Omega$:
\begin{align}\label{eq:X1X3}
    X_1 X_3 \equiv \ad_{Y_1 Y_2} \ad_{X_2 X_4} \ad_{X_2 X_3} \ad_{X_1 X_4} (Y_1 Y_2).
\end{align}
Thus, $X_1 X_3 \in \aa_4^\Omega$, and by the $X \leftrightarrow Y$ symmetry, we also get $Y_1 Y_3 \in \aa_4^\Omega$. Hence, $\Omega \sim_4 K_4$.

\paragraph*{Case $k = 6$.} The generators of $\aa_6$ are $XX, YZ, ZY$, which are equivalent to $ZZ, XY, YX$ under $X \leftrightarrow Z$ exchange. Considering the latter generators, from the $k = 2$ case of the lemma, we know that 
$X_1 Y_3$, $Y_1 X_3$, $X_3 Y_4$, $X_4Y_3 \in \aa_2^\Omega \subset \aa_6^\Omega$.
Then we can generate $Z_1 Z_3$ as follows:
\begin{align}\label{eq:Z1Z3}
    Z_1 Z_3 \equiv \ad_{X_3 Y_4} \ad_{Z_2 Z_3} \ad_{Z_2 Z_4} \ad_{Z_1 Z_4} (X_3 Y_4).
\end{align}
This proves that we can add the edge $(1,3)$, and 
$\Omega \sim_6 K_4$.

\paragraph*{Case $k = 14$.} 
As $\aa_{14}$ is generated by $XY, YX, ZI, IZ$, from the $k=2$ case above we get that
$X_1 Y_3, Y_1 X_3 \in \aa_2^\Omega \subset \aa_{14}^\Omega$.
Thus, we can add the edge $(1,3)$, and conclude that $\Omega \sim_{14} K_4$. We again provide an alternative derivation of Eqs.\ \eqref{eq:Y1X3}, ~\eqref{eq:X1X3}, ~\eqref{eq:Z1Z3}, and the fact that $X_1 X_3 \in \aa^\Omega_{14}$, by using frustration graphs in Appendix~\ref{frust-eq567}.
\end{proof}

Now we can prove a result similar to Lemma \ref{lem:kn1kn}.

\begin{lemma}\label{lem:kn1kn_again}
    For $k = 2, 4, 6, 14$ and $m \geq 4$, we have:
    \ifTikz
    \begin{align*}
    \begin{tikzpicture}[baseline=0mm]
        \begin{scope}[every node/.style={circle, thick,draw, scale=0.5}]
            \node[circle] (A) at (-2,0) {};
            \path [-] (A) edge[draw=black,very thick] (-1,0);
            \draw[line width=.5mm, dash pattern=on 3pt off 3pt] (1,0) arc (0:360:1);
            \node[circle, fill=white] (B) at (-1,0) {};
            \node[draw=white] (C) at (0,0) {};
            \node[draw=white, scale=3] at (C) {$K_{m-1}$};
        \end{scope}
    \end{tikzpicture}
    \quad \adjustbox{scale=2}{$\sim_k$} \quad 
    \begin{tikzpicture}[baseline=0mm]
        \begin{scope}[every node/.style={circle, thick,draw, scale=0.5}]
            \node[draw=white, scale=3] at (C) {$K_{m}$.};
        \end{scope}
    \end{tikzpicture}
	\end{align*}   
    \fi

\end{lemma}

\begin{proof}
    Let us label the vertices of $K_{m-1}$ clockwise as $1,2,\dots, m-1$, the extra vertex as $x$, and assume that $x$ is connected to $1$. Then, for any $1<i<j<m-1$, the vertices $x,1,i,j$ will form the subgraph $\Omega$ from Lemma \ref{lem:k3plusk4}: 
    \ifTikz
    \begin{align*}
        \begin{tikzpicture}[baseline=0mm]
        \begin{scope}[every node/.style={circle, thick,draw, scale=0.5}]
            \node[circle] (A) at (-2,0) {};
            \path [-] (A) edge[very thick] (-1,0);
            \draw[line width=.5mm, dash pattern=on 3pt off 3pt] (-150:1) arc (-150:120:1);
            \draw[line width=.5mm] (120:1) arc (120:210:1);
            \foreach \i in {120, 150, 180, 210}{
                \draw (\i:1) node[circle, fill=white] {};
            }
            \node[draw=white] (C) at (0,0) {};
			\node[draw=white, scale=2.5] at (C) {$K_{m-1}$};
        \end{scope}
        \node[above] at ($(120:1) + (-0.1, 0.05)$) {$3$};
        \node[above] at ($(150:1) + (-0.1, 0.075)$) {$2$};
        \node[above] at ($(180:1) + (-0.125, 0.1)$) {$1$};
        \node[above] at ($(210:1) - (0.6, 0.4)$) {$m-1$};
        \node[above] at (A.north) {$x$};
        \end{tikzpicture}
        \quad\adjustbox{scale=2}{$\supset$}\quad
        \begin{tikzpicture}[baseline=0mm]
        \begin{scope}[every node/.style={circle, thick,draw, scale=0.5}]
            \node[circle] (A) at (-2,0) {};
            \path [-] (A) edge[very thick, color=highlight] (-1,0);
            \node[draw=white] (C) at (0,0) {};
			\node[draw=white, scale=2.5] at (C) {$K_{m-1}$};
            \draw[line width=.5mm, dash pattern=on 3pt off 3pt] (-150:1) arc (-150:120:1);
            \draw[line width=.5mm] (120:1) arc (120:210:1);
            \path [-] (180:1) edge[draw=highlight, very thick, bend left=30] (60:1);
            \path [-] (180:1) edge[draw=highlight, very thick, bend right=30] (-60:1);
             \path [-] (60:1) edge[draw=highlight, very thick, bend left=30] (-60:1);
            \foreach \i in {60,-60, 120, 150, 180, 210}{
                \draw (\i:1) node[circle, fill=white] {};
            }  
        \end{scope}
        \node[above] at ($(120:1) + (-0.1, 0.05)$) {$3$};
        \node[above] at ($(150:1) + (-0.1, 0.075)$) {$2$};
        \node[above] at ($(180:1) + (-0.125, 0.1)$) {$1$};
        \node[above] at ($(210:1) - (0.6, 0.4)$) {$m-1$};
        \node[above] at ($(60:1) + (0.0, 0.1)$) {$i$};
        \node[below] at ($(-60:1) - (0.0, 0.1)$) {$j$};
        \node[above] at (A.north) {$x$};
    \end{tikzpicture}
    \end{align*}
    Since $\Omega\sim_k K_4$, we can connect $x$ to both $i$ and $j$:
    \begin{align*}
        \begin{tikzpicture}[baseline=0mm]
        \begin{scope}[every node/.style={circle, thick,draw, scale=0.5}]
            \node[circle] (A) at (-2,0) {};
            \path [-] (A) edge[very thick, color=highlight] (-1,0);
            \node[draw=white] (C) at (0,0) {};
			\node[draw=white, scale=2.5] at (C) {$K_{m-1}$};
            \draw[line width=.5mm, dash pattern=on 3pt off 3pt] (-150:1) arc (-150:120:1);
            \draw[line width=.5mm] (120:1) arc (120:210:1);
            \path [-] (180:1) edge[draw=highlight, very thick, bend left=30] (60:1);
            \path [-] (180:1) edge[draw=highlight, very thick, bend right=30] (-60:1);
             \path [-] (60:1) edge[draw=highlight, very thick, bend left=30] (-60:1);
            \foreach \i in {60,-60, 120, 150, 180, 210}{
                \draw (\i:1) node[circle, fill=white] {};
            }  
        \end{scope}
        \node[above] at ($(120:1) + (-0.1, 0.05)$) {};
        \node[above] at ($(150:1) + (-0.1, 0.075)$) {};
        \node[above] at ($(180:1) + (-0.125, 0.1)$) {$1$};
        \node[above] at ($(210:1) - (0.6, 0.4)$) {};
        \node[above] at ($(60:1) + (0.0, 0.1)$) {$i$};
        \node[below] at ($(-60:1) - (0.0, 0.1)$) {$j$};
        \node[above] at (A.north) {$x$};
        \end{tikzpicture}
        \quad\adjustbox{scale=2}{$\sim_k$}\quad
        \begin{tikzpicture}[baseline=0mm]
        \begin{scope}[every node/.style={circle, thick,draw, scale=0.5}]
            \node[draw=white] (C) at (0,0) {};
			\node[draw=white, scale=2.5] at (C) {$K_{m-1}$};
            \path [-] (180:1) edge[draw=highlight, very thick, bend left=30] (60:1);
            \path [-] (180:1) edge[draw=highlight, very thick, bend right=30] (-60:1);
             \path [-] (60:1) edge[draw=highlight, very thick, bend left=30] (-60:1);
            \path [-] (60:1) edge[draw=highlight, very thick, bend right=50] (-2,0);
            \path [-] (-60:1) edge[draw=highlight, very thick, bend left=50] (-2,0);
            \node[circle, fill=white] (A) at (-2,0) {};
            \path [-] (A) edge[very thick, color=highlight] (-1,0);
            \draw[line width=.5mm, dash pattern=on 3pt off 3pt] (-150:1) arc (-150:120:1);
            \draw[line width=.5mm] (120:1) arc (120:210:1);
            \foreach \i in {60,-60, 120, 150, 180, 210}{
                \draw (\i:1) node[circle, fill=white] {};
            }  
        \end{scope}
        \node[above] at ($(120:1) + (-0.1, 0.05)$) {};
        \node[above] at ($(150:1) + (-0.1, 0.075)$) {};
        \node[above] at ($(180:1) + (-0.125, 0.1)$) {$1$};
        \node[above] at ($(210:1) - (0.6, 0.4)$) {};
        \node[above] at ($(60:1) + (0.0, 0.1)$) {$i$};
        \node[below] at ($(-60:1) - (0.0, 0.1)$) {$j$};
        \node[above] at (A.north) {$x$};
    \end{tikzpicture}
    \end{align*}
    \fi
    Repeating this process for all $i,j$ allows us connect all vertices of $K_{m-1}$ to $x$, which gives $K_m$.
\end{proof}

Using the previous lemmas, we can conclude the non-bipartite case as follows.

\begin{theorem}\label{thm:branch_complete}
    If\/ $G$ is a connected non-bipartite graph with $n$ vertices that has at least one vertex of degree $>2$, then $G \sim_k K_n$ for $k = 2, 4, 6, 14$.
\end{theorem}

\begin{proof}
Since $G$ is not bipartite, it must contain a cycle $C_m$ with odd $m$.
Because $G$ is connected and has at least one vertex of degree $>2$, one of the vertices in this odd cycle should be attached to some other vertex, say $x$. Let us label the vertices of $C_m$ as $1,2,\dots,m$, and assume that $x$ is connected to $1$. Then if $m\ge5$, vertices $x,m,1,2,3$ form a subgraph $\Sigma$, and Lemma \ref{lem:deltak23} leads to the following $k$-equivalence:
\ifTikz
\begin{align*}
    \begin{tikzpicture}[baseline=0mm]
        \begin{scope}[every node/.style={circle, thick,draw, scale=0.5}]
            \node[circle] (A) at (-1.5,0) {};
            \path [-] (A) edge[very thick, color=highlight] (-1,0);
            \draw[line width=.5mm, dash pattern=on 3pt off 3pt] (-150:1) arc (-150:120:1);
            \draw[line width=.5mm, color=highlight] (120:1) arc (120:210:1);
            \foreach \i in {120, 150, 180, 210}{
                \draw (\i:1) node[circle, fill=white] {};
            }
            \node[draw=white] (C) at (0,0) {};
			\node[draw=white, scale=3] at (C) {$C_{m}$};
        \end{scope}
        \node[above] at ($(120:1) + (-0.1, 0.05)$) {$3$};
        \node[above] at ($(150:1) + (-0.1, 0.075)$) {$2$};
        \node[above] at ($(180:1) + (-0.125, 0.1)$) {$1$};
        \node[above] at ($(210:1) - (0.3, 0.3)$) {$m$};
        \node[above] at (A.north) {$x$};
    \end{tikzpicture}
    \quad\adjustbox{scale=2}{$\sim_k$}\quad
    \begin{tikzpicture}[baseline=0mm]
        \begin{scope}[every node/.style={circle, thick,draw, scale=0.5}]
            \node[circle] (A) at (-1.5,0) {};
            \path [-] (A) edge[very thick, color=highlight] (-1,0);
            \node[draw=white] (C) at (0,0) {};
			\node[draw=white, scale=3] at (C) {$\;\;C_{m-2}$};
            \draw[line width=.5mm, dash pattern=on 3pt off 3pt] (-150:1) arc (-150:120:1);
            \draw[line width=.5mm, color=highlight] (120:1) arc (120:210:1);
            \path [-] (A) edge[draw=highlight, very thick, bend left=30] (120:1);
            \path [-] (120:1) edge[draw=highlight, very thick, bend left=20] (210:1);
            \foreach \i in {120, 150, 180, 210}{
                \draw (\i:1) node[circle, fill=white] {};
            }
        \end{scope}
        \node[above] at ($(120:1) + (-0.1, 0.05)$) {$3$};
        \node[above] at ($(150:1) + (-0.1, 0.075)$) {$2$};
        \node[above] at ($(180:1) + (-0.125, 0.1)$) {$1$};
        \node[above] at ($(210:1) - (0.3, 0.3)$) {$m$};
        \node[above] at (A.north) {$x$};
    \end{tikzpicture}
    .
\end{align*}
\fi

Hence, if $G$ contains a $C_m$ with $m\ge5$, then $G$ is $k$-equivalent to another graph that contains a $C_{m-2}$.   
Repeating this process, we conclude that $G \sim_k G'$ for some connected graph $G'$ containing a $C_3$ attached to some other vertex, i.e., $\Omega\subset G'$ (cf.\ Lemma \ref{lem:k3plusk4}).

Now, by Lemma \ref{lem:k3plusk4}, $G' \sim_k G''$ for some connected graph $G''$ such that $K_4 \subset G''$. If $n>4$, since $G''$ is connected, its $K_4$ subgraph is attached to some other vertex. By Lemma \ref{lem:kn1kn_again}, we have a $k$-equivalent $K_5$ subgraph. Repeating this process, we eventually obtain the complete graph $K_n$.
\end{proof}

\subsubsection{Bipartite case}

Now we suppose that the interaction graph $G$ is bipartite. This means that its set of vertices is a disjoint union of two subsets, which we will denote as $U$ and $V$, so that all edges connect a vertex from $U$ to a vertex from $V$.
As before, we further assume that $G$ has at least one vertex of degree $>2$.

We start with the following obvious lemma about such graphs.

\begin{lemma}\label{lem:treesubG}
    Let $G$ be a connected bipartite graph with $|U|>1$ and $|V|>1$,
    which has at least one vertex of degree $>2$. Then $G$ contains $\Sigma$ as a subgraph:
        \ifTikz
    \begin{align*}
    \adjustbox{scale=1.5}{$\Sigma \: :=$} \quad
    \begin{tikzpicture}[baseline=0mm, scale=1]
    \begin{scope}[every node/.style={circle,thick,draw,scale=0.5}]
        \node (A) at (1,1.) {};
        \node (B) at (0.0,0.5) {};
        \node (C) at (0,-0.5) {};
        \node (D) at (1, -1) {};
        \node (E) at (1, 0.) {};
    \end{scope}
    \begin{scope}[every edge/.style={draw=black,very thick}]
        \path [-] (A) edge (B);
        \path [-] (C) edge (D);
        \path [-] (B) edge (D);
        \path [-] (E) edge (B);
    \end{scope}
    \end{tikzpicture}
    \quad\adjustbox{scale=1.5}{$\subset G$.} \quad
    \end{align*}
    \fi
\end{lemma}
\begin{proof}
By assumption, $G$ has a vertex $v_1$ connected to three other vertices $u_1, u_2, u_3$. Without loss of generality, we assume that $v_1\in V$.
Since $G$ is connected, there exists another vertex $v_2 \in V$ that is connected to at least one of the vertices $u_1, u_2, u_3$.
After relabeling, let $v_2$ be connected to $u_3$. Then we have the following subgraph of $G$:
\ifTikz
\begin{align*}
    \adjustbox{scale=1.5}{$\Sigma$} \quad
    \adjustbox{scale=1.5}{$=$} \quad
    \begin{tikzpicture}[baseline=0mm, scale=1] 
    \begin{scope}[every node/.style={circle,thick,draw,scale=0.5}]
        \node (A) at (1,1.) {};
        \node (B) at (0.0,0.5) {};
        \node (C) at (0,-0.5) {};
        \node (D) at (1, -1) {};
        \node (E) at (1, 0.) {};
    \end{scope}
    \node[right] at (A.east) {$u_1$};
    \node[above] at (B.north) {$v_1$};
    \node[above] at (C.north) {$v_2$};
    \node[right] at (D.east) {$u_3$};
    \node[right] at (E.east) {$u_2$};
    \begin{scope}[every edge/.style={draw=black,very thick}]
        \path [-] (A) edge (B);
        \path [-] (C) edge (D);
        \path [-] (B) edge (D);
        \path [-] (E) edge (B);
    \end{scope}
    \end{tikzpicture}
\end{align*}
\fi
thus completing the proof. 
    
\end{proof}

Recall that the \emph{complete bipartite graph} $K_{l,m}$ is a bipartite graph with $|U|=l$, $|V|=m$, where every vertex in $U$ is connected to every vertex in $V$.
We have shown in Lemma \ref{lem:deltak23} that $\Sigma \sim_k K_{2,3}$.
Next, we have an analogue of Lemma \ref{lem:kn1kn_again}.

\begin{lemma}\label{lem:knm1kn1m}
    For $k = 2,4,6,14$ and $l>1$, $m>1$, we have:
    \ifTikz
    \begin{align*}
    \begin{tikzpicture}[baseline=0mm]
        \begin{scope}[every node/.style={circle, thick,draw, scale=0.5}]
            \node[circle] (A) at (-2,0) {};
            \path [-] (A) edge[draw=black,very thick] (-1,0);
            \draw[line width=.5mm, dash pattern=on 3pt off 3pt] (1,0) arc (0:360:1);
            \node[circle, fill=white] (B) at (-1,0) {};
            \node[draw=white] (C) at (0,0) {};
            \node[draw=white, scale=3] at (C) {$K_{l,m}$};
        \end{scope}
    \end{tikzpicture}
    \quad \adjustbox{scale=2}{$\sim_k$} \quad 
    \begin{tikzpicture}[baseline=0mm]
        \begin{scope}[every node/.style={circle, thick,draw, scale=0.5}]
            \node[draw=white, scale=3] at (C) {$K_{l,m+1}$};
        \end{scope}
    \end{tikzpicture}
	\end{align*}  
    \fi
    where the extra vertex is attached to the $l$ vertices of $K_{l,m}$.
\end{lemma}

\begin{proof}
Let us label the vertices of $K_{l,m}$ as $U=\{u_1,\ldots,u_l\}$ and $V=\{v_1,\ldots,v_m\}$, and the extra vertex as $x$. 
    The graph formed by $x, u_1, u_2, v_1, v_2$ contains $\Sigma$ as a subgraph. By Lemma \ref{lem:deltak23}, we can then connect $x$ to $u_2$:
   
    \ifTikz
    \begin{align*}
        \begin{tikzpicture}[baseline=7mm, scale=2] 
        \begin{scope}[every node/.style={circle,thick,draw,scale=0.5}]
            \node (A) at (0,0.) {};
            \node (B) at (0.5,0.0) {};
            \node (C) at (1,0) {};
            \node (D) at (0.33,1) {};
            \node (E) at (0.66,1) {};
            \node (F) at (1,1) {};
            \node (G) at (0,1) {};
            \node (X) at (-0.5,0.5) {};
            \node[coordinate] (U) at (1.1, 1.2) {};
            \node[coordinate] (V) at (1.1, -0.2) {};
        \end{scope}
        \node at (U.center) {$U$};
        \node at (V.center) {$V$};
        \node[below=1.5mm] at (A.south) {$v_1$};
        \node[below=1.5mm] at (B.south) {$v_2$};
        \node[above=1.5mm] at (G.north) {$u_1$};
        \node[above=1.5mm] at (D.north) {$u_2$};
        \node[below=1.5mm] at (X) {$x$};
        \draw[dashed] (0.5,1) ellipse (8mm and 1.3mm);
        \draw[dashed] (0.5,0) ellipse (7mm and 1.3mm);
        \begin{scope}[every edge/.style={draw=black,very thick}]
            \path [-] (A) edge (F);
            \path [-] (D) edge (C);
            \path [-] (B) edge (E);
            \path [-] (B) edge (F);
            \path [-] (B) edge (D);
            \path [-] (B) edge[color=highlight] (G);
            \path [-] (C) edge (E);
            \path [-] (C) edge (F);
            \path [-] (C) edge (G);
            \path [-] (A) edge[color=highlight] (D);
            \path [-] (A) edge[color=highlight] (G);
            \path [-] (G) edge[color=highlight] (X);
            \path [-] (A) edge (E);
        \end{scope}
        \end{tikzpicture}
        \:\adjustbox{scale=2}{$\sim_k$}\:
        \begin{tikzpicture}[baseline=7mm, scale=2]
        \begin{scope}[every node/.style={circle,thick,draw,scale=0.5}]
            \node (A) at (0,0.) {};
            \node (B) at (0.5,0.0) {};
            \node (C) at (1,0) {};
            \node (D) at (0.33,1) {};
            \node (E) at (0.66,1) {};
            \node (F) at (1,1) {};
            \node (G) at (0,1) {};
            \node (X) at (-0.5,0.) {};
            \node[coordinate] (U) at (1.1, 1.2) {};
            \node[coordinate] (V) at (1.1, -0.2) {};
        \end{scope}
        \node at (U.center) {$U$};
        \node at (V.center) {$V$};
        \node[below=1.5mm] at (A.south) {$v_1$};
        \node[below=1.5mm] at (B.south) {$v_2$};
        \node[above=1.5mm] at (G.north) {$u_1$};
        \node[above=1.5mm] at (D.north) {$u_2$};
        \node[below=1.5mm] at (X) {$x$};
        \draw[dashed] (0.5,1) ellipse (8mm and 1.3mm);
        \draw[dashed] (0.25,0) ellipse (9mm and 1.3mm);
        \begin{scope}[every edge/.style={draw=black,very thick}]
            \path [-] (A) edge (E);
            \path [-] (A) edge[color=highlight] (D);
            \path [-] (A) edge (F);
            \path [-] (A) edge[color=highlight] (G);
            \path [-] (D) edge (C);
            \path [-] (B) edge (E);
            \path [-] (B) edge (F);
            \path [-] (B) edge (D);
            \path [-] (B) edge[color=highlight] (G);
            \path [-] (C) edge (E);
            \path [-] (C) edge (F);
            \path [-] (C) edge (G);
            \path [-] (X) edge[color=highlight] (D);
            \path [-] (X) edge[color=highlight] (G);
        \end{scope}
        \end{tikzpicture}
    \end{align*}
    \fi
By symmetry,  we can connect $x$ to all vertices in $U$, which gives $K_{l,m+1}$.
\end{proof}

Combining these lemmas gives us the following theorem:

\begin{theorem}\label{thm:klm}
Let\/ $G$ be a connected bipartite graph with $l$ and $m$ vertices in each color, which has at least one vertex of degree $>2$. Then $G \sim_k K_{l,m}$ for $k = 2,4,6,14$.
\end{theorem}
\begin{proof}
If $l=1$, then since $G$ is connected, the single vertex of the first color must be connected to each of the other $m$ vertices; hence $G=K_{1,m}$.
Suppose now that $l>1$ and $m>1$.
By Lemma \ref{lem:treesubG}, we know that $G$ must contain $\Sigma$ as a subgraph. By Lemma \ref{lem:deltak23}, $\Sigma \sim_k K_{2,3}$. 
Using Lemma \ref{lem:knm1kn1m}, we can grow $K_{2,3}$ to all of $G$, which eventually gives $G \sim_k K_{l,m}$. 
\end{proof}

\section{Identifying the DLAs\label{sec:identify}}
In this section, we finish the proof of \Cref{the:classification}.

\subsection{Initial Observations}

First, as already mentioned before, the case $k=0$ trivially gives the Abelian Lie algebra 
\begin{align}\label{eq:aa0}
\aa_0^G = i\,\Span_{\mathbb R} \{ X_i X_j \,|\, (i,j) \in G\},
\end{align}
spanned by $XX$ placed on every edge $(i,j)$ of the interaction graph $G$.

Second, for $k=7,16,20,22$, \Cref{thm:complete} tells us that
\begin{align}\label{eq:aakpi}
\aa_k^G = \aa_k^{K_n} = \aa_k^\pi(n),
\end{align}
where $n$ is the number of vertices of $G$ and in the right-hand side we used the notation of~\cite{wiersema2023classification}.
Now~\cite[Theorem IV.3]{wiersema2023classification} provides the answer for $\aa_k^G$ (see also Appendix \ref{sec:list_dla}).

Third, for $k=2,4,6,14$, Eq.\ \eqref{eq:aakpi} still holds when the graph $G$ is not bipartite (and not a cycle), by \Cref{thm:branch_complete}.
Thus, we are only left to consider the case where $k=2,4,6,14$ and $G$ is bipartite. Provided that $G$ is not a line or circle,
which were already considered in~\cite{wiersema2023classification}, we can apply \Cref{thm:klm} to conclude that
$\aa_k^G = \aa_k^{K_{l,m}}$. For the remainder of this section, we will assume that $k=2,4,6,14$ and $G=K_{l,m}$.

The number of cases can be reduced due to the following result, which is a straightforward generalization of 
\cite[Lemmas C.3, C.4]{wiersema2023classification}.

\begin{proposition}\label{prop:iso}
For every bipartite interaction graph $G$, we have\/ $\aa_2^G \cong \aa^G_4$ and\/ $\aa_6^G \cong \aa^G_7$.
\end{proposition}
\begin{proof}
Recall that $\aa_2=\Lie{XY, YX}$ and $\aa_4=\Lie{XX, YY}$. 
By construction, the DLA $\aa_2^G$ is generated by $X_iY_j$ for every edge $(i,j)\in G$.
As before, let us denote by $U$ and $V$ the sets of vertices in $G$ for each of the two colors, so that every edge connects a vertex from $U$ to a vertex from $V$. Thus, the generators of $\aa_2^G$ have the form $X_uY_v, Y_uX_v$ for edges $(u,v)$ with $u\in U$, $v\in V$.
If we swap $X_v \leftrightarrow Y_v$ for all $v\in V$, these generators will transform into the generators $X_uX_v, Y_uY_v$ of $\aa_4^G$.
Therefore, $\aa_2^G \cong \aa^G_4$.

The proof of $\aa_6^G \cong \aa^G_7$ is similar, by swapping $Y_v \leftrightarrow Z_v$ for all $v\in V$.
\end{proof}

\subsection{Upper Bounds for $\aa^{K_{l,m}}_4$ and $\aa^{K_{l,m}}_{14}$}\label{sec:upb}

Due to Proposition \ref{prop:iso}, because $\aa^G_7$ is determined by Eq.\ \eqref{eq:aakpi}, 
we are only left with the cases $k = 4$ and $k = 14$, with the interaction graph
$G=K_{l,m}$ a complete bipartite graph.
The first step in determining the corresponding DLAs is to find upper bounds for them.
For convenience, recall that
\begin{equation}\label{eq:gen414}
\begin{split}
\aa_4       =&	\Lie{XX, YY},\\
\aa_{14}	=&	
\Lie{XX, YY, ZI, IZ}.
\end{split}
\end{equation}

Let us denote the vertices of the first color of $K_{l,m}$ as $v_1, \dots, v_l$, and those of the second color as $u_1, \dots, u_m$.
We introduce the operator
\begin{align}\label{eq:q}
Q = \prod_{i=1}^l Y_{v_i} \prod_{j=1}^m X_{u_j} = Y^{\otimes l} \otimes X^{\otimes m}.
\end{align}
Since $Q$ is a Pauli string, it satisfies $Q^\dagger=Q$ and $Q^T=(-1)^l Q$ (see Appendix \ref{secpauli1}).
Then the formula
\begin{align}\label{eq:theta}
\theta_{l,m}(g) = -Q g^T Q, \quad g\in\su(2^n), \;\; n=l+m,
\end{align}
defines an involution of $\su(2^n)$ (cf.\ \cite[Corollary A.3]{wiersema2023classification}).
Recall that an \emph{involution} of a Lie algebra $\g$ is a linear operator $\theta$ such that
$\theta^2=I$ and $\theta([a,b]) = [\theta(a), \theta(b)]$; then the set of \emph{fixed points},
\begin{align}\label{eq:fixed}
\g^\theta = \{ g\in\g \,|\, \theta(g)=g \},
\end{align}
is a subalgebra of $\g$.

\begin{lemma}\label{thm:upperbounds}
    The Lie algebras $\aa^{K_{l,m}}_4$ and $\aa^{K_{l,m}}_{14}$ are invariant under the involution $\theta_{l,m}$ defined by \eqref{eq:theta}, i.e.,
    \begin{align}\label{eq:upperbounds}
        \aa^{K_{l,m}}_k \subseteq \left( \aa^{K_{l+m}}_k\right)^{\theta_{l,m}}, \qquad k=4,14.
    \end{align}
\end{lemma}

\begin{proof}
The generators \eqref{eq:gen414} of both $\aa_4$ and $\aa_{14}$ are
invariant under transpose. Recall that for any interaction graph $G$, the DLA $\aa_4^G$ is generated by $X_iX_j, Y_iY_j$ for every edge $(i,j)$ of $G$,
while $\aa_{14}^G$ is generated by $X_iX_j, Z_i$ for every vertex $i$ and edge $(i,j)$ of $G$.
Because any two Pauli strings either commute or anti-commute (see \eqref{lemp1}), 
any generator $g$ of $\aa_k^{K_{l+m}}$ satisfies $\theta_{l,m}(g) = \pm g$. Hence, $\theta_{l,m}$ restricts to an involution of $\aa_k^{K_{l+m}}$.

For $G=K_{l,m}$, note that all generators of $\aa^{K_{l,m}}_k$
anti-commute with $Q$, because $XX,YY,ZI,IZ$ anti-commute with $YX$ and every edge has the form $(v_i,u_j)$ where $1\le i\le l$, $1\le j\le m$.
Thus, all generators are fixed under the involution $\theta_{l,m}$.
\end{proof}

In the next two lemmas, we identify the upper bounds from Eq.\ \eqref{eq:upperbounds}.
It is easier to start with $k=14$.

\begin{lemma}\label{lem:upb1}
We have
    \begin{align}
        \left( \aa^{K_{n}}_{14}\right)^{\theta_{l,m}} \cong \begin{cases}
            \sp(2^{n-2})^{\oplus 2}, &\text{$l,m$ odd,} \\
            \so(2^{n-1})^{\oplus 2}, &\text{$l,m$ even,} \\
            \su(2^{n-1}), &\text{$n=l+m$ odd.}
        \end{cases}
    \end{align}
\end{lemma}

\begin{proof}
First, let us describe more explicitly the Lie algebra $\aa^{K_n}_{14}$.
By Theorem C.1 and Lemma C.38 in \cite{wiersema2023classification}, we have
\begin{align}\label{eq:Kn14}
    \aa^{K_{n}}_{14} = \aa^\pi_{14}(n) = \su(2^{n})^{P_Z}/i \mathbb{R} P_Z,
\end{align}
where $P_Z=Z^{\otimes n}$ is the product of $Z$'s acting on every qubit.
As in \eqref{eq:fixed}, $\su(2^{n})^{P_Z}$ denotes the subalgebra of $\su(2^{n})$ fixed under $P_Z$;
in this case, commuting with $P_Z$.

In \cite[Lemma C.26]{wiersema2023classification}, we showed that
\begin{align*}
\su(2^{n})^{P_Z}/i \mathbb{R} P_Z \cong \su(2^{n-1}) \oplus \su(2^{n-1}).
\end{align*}
As in the proof of that lemma, consider the isomorphism $\varphi(g) = UgU^\dagger$, with the unitary matrix
\begin{align*}
U = e^{i\frac{\pi}{4}X \otimes Z^{\otimes n-1}} e^{-i\frac{\pi}{4}X_1} .
\end{align*}
Since $U P_Z U^\dagger = Z_1$ by \eqref{equab}, $\varphi$ sends $\su(2^{n})^{P_Z}/i \mathbb{R} P_Z$ onto
$\su(2^{n})^{Z_1}/i \mathbb{R} Z_1$. 
The latter has a basis consisting of all Pauli strings that commute with $Z_1$ excluding $Z_1$ itself, and is easily identified with
\begin{align*}
\{I&,Z \} \otimes \su(2^{n-1}) \\
&= \left(\frac{I+Z}{2} \otimes \su(2^{n-1})\right) \oplus \left(\frac{I-Z}{2} \otimes \su(2^{n-1})\right) \\
&\cong \su(2^{n-1}) \oplus \su(2^{n-1}).
\end{align*}  
In the last line above, the two $\su(2^{n-1})$ summands correspond to the eigenspaces of the operator $Z_1$.

According to \cite[Lemmas A.3, A.4]{wiersema2023classification}, 
$\varphi$ sends the fixed points $(\aa^{K_{n}}_{14})^{\theta_{l,m}}$ to the fixed points of
$\varphi(\aa^{K_{n}}_{14}) \cong \su(2^{n-1}) \oplus \su(2^{n-1})$ under the involution
$\Bar{\theta}_{l,m}(g) = -\Bar{Q}g^T\Bar{Q}^\dagger$ where $\Bar{Q}=UQU^T$.
Assuming that vertex 1 is $v_1$, we compute:
\begin{align*}
    \Bar{Q}&=UQU^T \\
    &= e^{i\frac{\pi}{4}X \otimes Z^{\otimes n-1}} e^{-i\frac{\pi}{4}X_1} \cdot Q \cdot e^{-i\frac{\pi}{4}X_1} e^{i\frac{\pi}{4}X \otimes Z^{\otimes n-1}} \\
    &= e^{i\frac{\pi}{4}X \otimes Z^{\otimes n-1}} \cdot Q \cdot e^{i\frac{\pi}{4}X \otimes Z^{\otimes n-1}}.
\end{align*}

\paragraph*{Case $n = l+m$ even.} 
In this case, $Q = Y^{\otimes l} \otimes X^{\otimes m}$ and $X \otimes Z^{\otimes n-1}$ commute, which leads to 
\begin{align*}
    \Bar{Q} &= i(X \otimes Z^{\otimes n-1}) \cdot (Y^{\otimes l} \otimes X^{\otimes m}) \\
    &\equiv Z \otimes X^{\otimes l-1} \otimes Y^{\otimes m}.
\end{align*}
The last equality is up to a phase, which has no effect on the involution itself. Since $[\Bar{Q},Z_1] = 0$, the involution $\Bar{\theta}_{l,m}$ will not mix the eigenspaces of $Z_1$ and will effect them separately. Taking the transpose, we find
\begin{align*}
    \Bar{Q}^T = (Z \otimes X^{\otimes l-1} \otimes Y^{\otimes m})^T = (-1)^m \Bar{Q}.
\end{align*}
By \cite[Corollary A.3]{wiersema2023classification}, we obtain
\begin{align*}
    \left( \su(2^{n-1})\right)^{\Bar{\theta}_{l,m}} \cong \begin{cases}
        \sp(2^{n-2}), &\text{$m$ odd,} \\
        \so(2^{n-1}), &\text{$m$ even.}
    \end{cases}
\end{align*}

\paragraph*{Case $n = l+m$ odd.} In this case, $Q = Y^{\otimes l} \otimes X^{\otimes m}$ and $X \otimes Z^{\otimes n-1}$ anti-commute, which gives $\Bar{Q} = Q$ and $\Bar{\theta}_{l,m} = \theta_{l,m}$.
As $Q$ and $Z_1$ do not commute, the involution mixes the eigenspaces of $Z_1$ together. 
When restricted to $I\otimes \su(2^{n-1})$, $\theta_{l,m}$ acts as $I \otimes \theta_{l-1,m}$,
while its restriction to $Z\otimes \su(2^{n-1})$ acts as $-I \otimes \theta_{l-1,m}$.
Hence, the following map is an isomorphism from $\su(2^{n-1})$ onto $(\aa^{K_{n}}_{14})^{\theta_{l,m}}$:
\begin{align*}
    g \mapsto I \otimes \frac{g + \theta_{l-1,m}(g)}{2} + Z \otimes \frac{g - \theta_{l-1,m}(g)}{2}.
\end{align*}
Therefore, 
\begin{align}\label{eq:a_14_odd}
    \left( \aa^{K_{n}}_{14}\right)^{\theta_{l,m}} \cong \su(2^{n-1}), \quad n=l+m \;\text{ odd,}
\end{align}
completing the proof of the lemma.
\end{proof}

\begin{lemma}\label{lem:upb2}
We have
    \begin{align}
        \left( \aa^{K_{n}}_4\right)^{\theta_{l,m}} \cong \begin{cases}
            \su(2^{n-2})^{\oplus 2}, &\text{$l,m$ odd,} \\
            \so(2^{n-2})^{\oplus 4}, &\text{$l,m$ even,} \\
            \so(2^{n-1}), &\text{$n=l+m$ odd.}
        \end{cases}
    \end{align}
\end{lemma}
\begin{proof}
We will continue to use the same notation as in the proof of Lemma \ref{lem:upb1}.
By Theorem C.1 and Lemma C.39 in \cite{wiersema2023classification}, we have for odd $n$:
\begin{align*}
\aa^{K_{n}}_4 = \aa_4^\pi(n) = \su(2^{n})^{\{P_X,P_Y,P_Z\}},
\end{align*}
which is the set of all matrices from $\su(2^{n})$ that commute with
$P_X=X^{\otimes n}$, $P_Y=Y^{\otimes n}$ and $P_Z=Z^{\otimes n}$.
When $n$ is even, $P_X,P_Y,P_Z$ commute with each other, and we have to quotient by them:
\begin{align}\label{eq:Kn4}
\aa^{K_{n}}_4 = \su(2^{n})^{\{P_X,P_Y,P_Z\}} / i\,\Span_{\mathbb R} \{P_X,P_Y,P_Z\}.
\end{align}
In \cite[Lemma C.28]{wiersema2023classification}, we identified these Lie algebras as
\begin{align*}
    \aa^{K_{n}}_4 \cong \begin{cases}
            \su(2^{n-2})^{\oplus 4}, &\text{$n$ \, even,} \\
            \su(2^{n-1}), &\text{$n$ \, odd.} 
        \end{cases}
\end{align*}

\paragraph*{Case $n=l+m$ even.} In this case, the 4 copies of $\su(2^{n-2})$ correspond to the mutual eigenspaces of $P_X$ and $P_Y$. 
As in the proof of \cite[Lemma C.28]{wiersema2023classification}, we apply the isomorphism $\varphi(g) = UgU^\dagger$, with the unitary matrix
\begin{align*}
    U = e^{i\frac{\pi}{4}X_2} e^{i\frac{\pi}{4} I\otimes X\otimes Z^{\otimes n-2}} e^{i\frac{\pi}{4}Y\otimes X^{\otimes n-1}}.
\end{align*}
It has the property that
\begin{align*}
    U P_X U^\dagger &= Z_1, \\
    U P_Y U^\dagger &= (-1)^{(n+2)/2} Z_2, \\
    U P_Z U^\dagger &= - Z_1 Z_2,
\end{align*}
hence, 
\begin{align*}
\varphi(\aa^{K_{n}}_4) = \{I,Z\} \otimes \{I,Z\} \otimes \su(2^{n-2}).
\end{align*}

\paragraph*{Case $l,m$ even.} 
Now $[Q, P_X] = [Q, P_Y] = 0$; therefore, $\theta_{l,m}$ does not mix the 4 different eigenspaces and acts on each of them separately. 
Due to \cite[Lemmas A.3, A.4]{wiersema2023classification}, under the isomorphism $\varphi$, the involution $\theta_{l,m}$ transforms into
$\Bar{\theta}_{l,m}(g) = -\Bar{Q}g^T\Bar{Q}^\dagger$ where $\Bar{Q}=UQU^T$.

By symmetry, without loss of generality, we suppose that the first two vertices of $K_{l,m}$ are $v_1,v_2$.
Then $Q$ anti-commutes with $X_2$, and we find using \eqref{equab}:
\begin{align*}
    \Bar{Q} &= UQU^T = e^{i\frac{\pi}{4}X_2} e^{i\frac{\pi}{4} I\otimes X\otimes Z^{\otimes n-2}} e^{i\frac{\pi}{4}Y\otimes X^{\otimes n-1}} Q \\
    &\;\;\cdot e^{-i\frac{\pi}{4}Y\otimes X^{\otimes n-1}}
    e^{i\frac{\pi}{4} I\otimes X\otimes Z^{\otimes n-2}} 
    e^{i\frac{\pi}{4}X_2} \\
    &= e^{i\frac{\pi}{4}X_2} e^{i\frac{\pi}{4} I\otimes X\otimes Z^{\otimes n-2}} \cdot i (Y\otimes X^{\otimes n-1}) \cdot Q\\
    &\;\;\cdot e^{i\frac{\pi}{4} I\otimes X\otimes Z^{\otimes n-2}} 
    e^{i\frac{\pi}{4}X_2}\\
    &= i (Y\otimes X^{\otimes n-1}) \cdot e^{i\frac{\pi}{4}X_2} e^{i\frac{\pi}{4} I\otimes X\otimes Z^{\otimes n-2}} Q \\
    &\;\;\cdot e^{i\frac{\pi}{4} I\otimes X\otimes Z^{\otimes n-2}} 
    e^{i\frac{\pi}{4}X_2} \\
    &= i (Y\otimes X^{\otimes n-1}) \cdot e^{i\frac{\pi}{4}X_2}  Q 
    e^{i\frac{\pi}{4}X_2} \\
    &= i (Y\otimes X^{\otimes n-1}) \cdot Q \\
    &\equiv I \otimes Z^{\otimes l-1} \otimes I^{\otimes m},
\end{align*}
where the last equality is up to a phase.
We see that $\Bar{\theta}_{l,m}$ restricts on each $\su(2^{n-2})$ summand as the involution
$g\mapsto -\tilde{Q}g^T\tilde{Q}$ where $\tilde{Q}=Z^{\otimes l-2} \otimes I^{\otimes m}$.
Its fixed points are $\so(2^{n-2})$, by \cite[Corollary A.3]{wiersema2023classification}, since $\tilde{Q}^T = \tilde{Q}$.

\paragraph*{Case $l,m$ odd.} In this case, $Q$ does not commute with $P_X$ and $P_Y$, and it commutes with $P_Z \equiv P_X P_Y$. 
Thus, $\theta_{l,m}$ will not mix the eigenspaces of $P_Z$, but it will mix the eigenspaces of $P_X$, as in the odd $n$ case for $\aa_{14}^{K_{l,m}}$ (see Eq.\ \eqref{eq:a_14_odd}). The same reasoning yields
\begin{align*}
    \left(\aa^{K_{n}}_4\right)^{\theta_{l,m}} \cong \su(2^{n-2})^{\oplus 2},
\end{align*}
where each $\su(2^{n-2})$ summand corresponds to an eigenspace of $P_Z$ that is invariant under $\theta_{l,m}$.

\paragraph*{Case $n =l+ m$ odd.} In this case, $P_X$ and $P_Y$ do not commute. 
As in the proof of \cite[Lemma C.28]{wiersema2023classification}, we apply the isomorphism $\varphi(g) = UgU^\dagger$, with
\begin{equation*}
U = e^{i\frac{\pi}{4}Z\otimes Y^{\otimes n-1}} e^{i\frac{\pi}{4}Y\otimes X^{\otimes n-1}},
\end{equation*}
which satisfies
\begin{align*}
U P_X U^\dagger = Z_1, \quad U P_Y U^\dagger = X_1. 
\end{align*}
Hence,
\begin{equation*}
\varphi(\aa_4^{K_N}) = \su(2^n)^{\{X_1,Z_1\}} = I \otimes \su(2^{n-1}) \cong \su(2^{n-1}).
\end{equation*}

By symmetry, without loss of generality, we can assume that $l$ is even and the first two vertices of $K_{l,m}$ are $v_1,v_2$.
Then $Q = Y^{\otimes l} \otimes X^{\otimes m}$ anti-commutes with $Y\otimes X^{\otimes n-1}$ and commutes with $Z\otimes Y^{\otimes n-1}$.
We find using \eqref{equab}:
\begin{align*}
    \Bar{Q} &= UQU^T = e^{i\frac{\pi}{4}Z\otimes Y^{\otimes n-1}} e^{i\frac{\pi}{4}Y\otimes X^{\otimes n-1}} Q \\
    &\;\;\cdot e^{-i\frac{\pi}{4}Y\otimes X^{\otimes n-1}} e^{i\frac{\pi}{4}Z\otimes Y^{\otimes n-1}} \\
    &= e^{i\frac{\pi}{4}Z\otimes Y^{\otimes n-1}} \cdot i (Y\otimes X^{\otimes n-1}) \cdot Q e^{i\frac{\pi}{4}Z\otimes Y^{\otimes n-1}} \\
    &= i (Y\otimes X^{\otimes n-1}) \cdot e^{-i\frac{\pi}{4}Z\otimes Y^{\otimes n-1}} Q e^{i\frac{\pi}{4}Z\otimes Y^{\otimes n-1}} \\
    &= i (Y\otimes X^{\otimes n-1}) \cdot Q \\
    &\equiv I \otimes Z^{\otimes l-1} \otimes I^{\otimes m}.
\end{align*}
We see that $\Bar{\theta}_{l,m}$ restricts on $\su(2^{n-1})$ as the involution
$g\mapsto -\tilde{Q}g^T\tilde{Q}$ where $\tilde{Q}=Z^{\otimes l-1} \otimes I^{\otimes m}$.
Its fixed points are $\so(2^{n-1})$, by \cite[Corollary A.3]{wiersema2023classification}.

\end{proof}

\subsection{Tightness of the Upper Bounds}
Here we continue to use the notation of Section \ref{sec:upb}.

\begin{lemma}\label{lem:upbt}
    The upper bounds in \eqref{eq:upperbounds} are tight, i.e.,
    \begin{align}
        \aa_k^{K_{l,m}} &= \left( \aa_k^{K_{l+m}} \right)^{\theta_{l,m}}, \qquad k=4,14.
    \end{align}
\end{lemma}

\begin{proof}
For convenience, let us denote the Lie algebra $(\aa_{k}^{K_{l+m}})^{\theta_{l,m}}$ by $\g^{l,m}_k$.
We have to prove that $\g^{l,m}_k \subseteq \aa_k^{K_{l,m}}$.
Notice that, as in \cite[Corollary C.1]{wiersema2023classification}, both of these Lie algebras are linearly spanned by Pauli strings.
Hence, it will be enough to show that every Pauli string $\sigma \in \g^{l,m}_k$ lies in $\aa_k^{K_{l,m}}$.

We prove this statement by induction on $n=l+m$. 
The base $n=2$ holds because $K_{1,1} = K_2$, and hence 
$\g^{1,1}_k \subseteq \aa_k^{K_2} = \aa_k^{K_{1,1}}$.
Suppose, by induction, that every Pauli string from $\g^{l',m'}_k$ lies in $\aa_k^{K_{l',m'}}$ whenever $l'+m'<n=l+m$. 
Now consider an arbitrary Pauli string $\sigma \in \g^{l,m}_k$.
If any of the factors in $\sigma$ is $I$, then by symmetry we can assume that
either $\sigma \in I \otimes \g^{l-1,m}_k$ or $\sigma \in \g^{l,m-1}_k \otimes I$.
By the inductive assumption, we get
\begin{align*}
I \otimes \g^{l-1,m}_k \subseteq I \otimes \aa_k^{K_{l-1,m}} \subseteq \aa_k^{K_{l,m}},
\end{align*}
and the second case is similar.

When $\sigma$ does not have any factors $I$, our strategy will be to find Pauli strings $\gamma_1,\dots,\gamma_r \in \aa_k^{K_{l,m}}$ such that
\begin{align*}
\ad_{\gamma_1} \cdots \ad_{\gamma_r}(\sigma)
\end{align*}
has an $I$ factor. Then \cite[Lemma C.16]{wiersema2023classification} will ensure that $\sigma\in \aa_k^{K_{l,m}}$.
Now let us consider the cases $k = 14$ and $k=4$ separately.

\paragraph*{Case $k = 14$.} 
If any of the factors of $\sigma$ is a $Y$, we can transform it into an $X$, since all $Z_i$ are in the generating set of $\aa_{14}^{K_{l,m}}$. 
Thus, by symmetry, we can assume that
\begin{align*}
\sigma = X^{\otimes a} \otimes Z^{\otimes l-a} \otimes X^{\otimes m-b} \otimes Z^{\otimes b},
\end{align*}
for some $0\le a\le l$ and $0\le b\le m$.
If $a>0$ and $b > 0$, then $\ad_{\gamma} \sigma$ will have $I$ in the first factor for $\gamma=X_1 X_n$.
Similarly, if $a < l$ and $b < m$, we can take $\gamma=X_l X_{l+1}$.

The only cases left are $\sigma = X^{\otimes n} = P_X$ or $\sigma = Z^{\otimes n} = P_Z$. The second is impossible, because
$P_Z \notin \g^{l,m}_{14}$ by Eq.\ \eqref{eq:Kn14}. When $\sigma = P_X$, we observe that
\begin{align*}
\ad_{X_1 X_2} \ad_{Z_2}(P_X) = -I\otimes Z \otimes X^{\otimes n-2}.
\end{align*}
This proves the lemma for $k = 14$.

\paragraph*{Case $k = 4$.} 
If $\sigma$ has a $Y$ in the first $l$ factors and an $X$ or $Z$ in the last $m$ factors, by symmetry we can assume that these are $Y_1$ and $X_n$ or $Z_n$. Then $\ad_{Y_1 Y_n} \sigma$ has $I$ in the first factor. The case where $\sigma$ has a $Y$ in the last $m$ factors is treated similarly.
Therefore, if $\sigma$ has any $Y$ factors, we are done unless $\sigma=P_Y$. However, $P_Y \notin \g^{l,m}_{4}$ by Eq.\ \eqref{eq:Kn4}.

Similarly, if $\sigma$ has an $X$ in the first $l$ factors and a $Z$ in the last $m$ factors, by symmetry we can assume that these are $X_1$ and $Z_n$.
Then $\ad_{X_1 X_n} \sigma$ has $I$ in the first factor. The only cases left are $\sigma=P_X$ or $\sigma=P_Z$, but these are again impossible by 
Eq.\ \eqref{eq:Kn4}.
\end{proof}

Combining Lemmas \ref{lem:upb1}, \ref{lem:upb2} and \ref{lem:upbt} gives us the answer for the DLAs $\aa^{K_{l,m}}_4$ and $\aa^{K_{l,m}}_{14}$,
thus concluding the proof of \Cref{the:classification}.

\section{Conclusion}

In conclusion, we have extended the classification of dynamical Lie algebras generated by $2$-local spin interaction Hamiltonians from one-dimensional spin chains \cite{wiersema2023classification} to the more general context of undirected graphs. Our analysis reveals that the one-dimensional case is unique, since for all other graphs, the structure of the dynamical Lie algebra is determined by whether the graph is bipartite or non-bipartite. We find that the cases where the dynamical Lie algebra has a polynomial size are exceptional and limited to one-dimensional systems. Given the consequences of our results, we have to perhaps review how much the dynamical Lie algebra tells us about trainablity, since a direct consequence of our work is that any quantum circuit consisting of $2$-local gates (except the one dimensional transverse-field Ising model) is not trainable as a variational quantum circuit due to barren plateau issues~\cite{ragone2023,fontana2023theadjoint}. 

We believe that our results, together with the work of \cite{aguilar2024full} provide two complementary techniques for analyzing dynamical Lie algebras generated by Pauli strings. An important future direction would be to extend these results to linear combinations of Paulis strings, since this will enable one to connect DLAs more directly with physical models. For example, right now we can only determine the DLA on a graph $G$ of a generating set of the form $\{XX,YY,ZZ\}$. However, this does not tell us about the DLA generated by $\{XX + YY + ZZ\}$ or $\{XX+YY, ZZ\}$. 
Choosing the generators as linear combinations rather than single Pauli strings may reduce the dimension of the dynamical Lie algebra significantly if the linear combinations are invariant under a symmetry. This is studied for permutation invariant circuits \cite{mansky2023permutation, mansky2024scaling}, and translationally invariant quantum approximate optimization algorithm (QAOA) circuits, which correspond to the generators of $\aa_{14}$ \cite{dalessandro2021dynamical, allcock2024dynamical}.

\section*{Acknowledgements}

AFK acknowledges financial support from the National Science Foundation under award No. 1818914: PFCQC: STAQ: Software-Tailored Architecture for Quantum co-design and No. 2325080:
PIF: Software-Tailored Architecture for Quantum Co-Design.
BNB was supported in part by a Simons Foundation grant No. 584741.
Resources used in preparing this research were provided, in part, by the Province of Ontario, the Government of Canada through CIFAR, and companies sponsoring the Vector Institute \url{https://vectorinstitute.ai/#partners}.

\section*{Author Contributions}
Theoretical results were derived by EK and BNB, while computer calculations were performed by RW.
The manuscript was written mostly by EK and BNB. The figures were made by RW. All authors contributed to reviewing and editing the manuscript. 

\bibliography{literature}

\renewcommand{\appendixtocname}{Supplemental Material}
\onecolumngrid

\renewcommand\thefigure{S\arabic{figure}}  
\renewcommand\thetable{S.\Roman{table}}  
\setcounter{figure}{0}
\setcounter{table}{0}

\renewcommand{\thesection}{\Alph{section}}
\renewcommand{\thesubsection}{\Roman{subsection}}

\renewcommand{\theequation}{\Alph{section}\arabic{equation}}
\counterwithin*{equation}{section}
\setcounter{section}{0}
\pagebreak
\begin{center}
\textbf{\large Supplemental Material}
\end{center}
\setcounter{equation}{0}
\setcounter{page}{1}
\makeatletter
\section{Pauli strings}\label{secpauli1}

Throughout the paper, we work with the \emph{Pauli matrices}
\begin{equation*}
\sigma_0=I=\begin{pmatrix} 1 & 0 \\ 0 & 1 \end{pmatrix}, \qquad
\sigma_1=X=\begin{pmatrix} 0 & 1 \\ 1 & 0 \end{pmatrix}, \qquad
\sigma_2=Y=\begin{pmatrix} 0 & -i \\ i & 0 \end{pmatrix}, \qquad
\sigma_3=Z=\begin{pmatrix} 1 & 0 \\ 0 & -1 \end{pmatrix},
\end{equation*}
including the identity matrix $I$,
which form a basis for the real vector space of $2\times 2$ Hermitian matrices.
We will denote by $A^T$ the transpose of a matrix, and by $A^\dagger$ its Hermitian conjugate
(which is obtained from $A^T$ by taking complex conjugates of all entries).
Thus, $A^\dagger = A$ for all $A\in\mcP_1 := \{I,X,Y,Z\}$. On the other hand, we have
\begin{equation*}
Y^T=-Y, \qquad A^T=A \quad\text{for}\quad A=I,X,Z.
\end{equation*}

Length-$n$ \emph{Pauli strings} are tensor products of $n$ Pauli matrices of the form
\begin{equation}\label{pauli2}
a = A^1 \otimes A^2 \otimes \dots \otimes A^n, \qquad A^j \in\mcP_1
\end{equation}
(where the superscripts are indices not powers).
We denote the set of all such Pauli strings by $\mcP_n := \{I,X,Y,Z\}^{\otimes n}$.
Every $a\in\mcP_n$ is a linear operator on the Hilbert space $(\mathbb{C}^2)^{\otimes n}$ of $n$ qubits,
so $a$ can be represented as a matrix of size $2^n \times 2^n$ (by the Kronecker product).
In particular, $I^{\otimes n}$ is the $2^n \times 2^n$ identity matrix.
The Hermitian conjugate and transpose of a Pauli string are done componentwise:
\begin{align*}
a^\dagger &= (A^1)^\dagger \otimes (A^2)^\dagger \otimes \dots \otimes (A^n)^\dagger = a, \\
a^T &= (A^1)^T \otimes (A^2)^T \otimes \dots \otimes (A^n)^T = (-1)^{\# \{ A^j=Y \}} a.
\end{align*}
All Pauli strings are Hermitian, and $\mcP_n$ is a basis (over $\mathbb R$) of the vector space
of $2^n \times 2^n$ Hermitian matrices.

To shorten the notation, we will often omit the tensor product signs in Pauli strings, so \eqref{pauli2}
will be written as $a = A^1 A^2 \cdots A^n$. For example, we will write
\begin{equation*}
XX = X \otimes X, \qquad XY = X \otimes Y, \qquad Z \cdots Z = Z^{\otimes n}, \quad\text{etc.}
\end{equation*}
For $A\in\mcP_1$ and $1\le j\le n$, we will denote by
\begin{equation*}
A_j := I^{\otimes (j-1)} \otimes A \otimes I^{\otimes (n-j)}
\end{equation*}
the linear operator $A$ acting on the $j$-th qubit. For example, for $n=3$,
\begin{equation*}
X_1 = XII = X \otimes I \otimes I, \qquad Z_2 = IZI = I \otimes Z \otimes I, \qquad X_1Z_2Y_3 = XZY = X \otimes Z \otimes Y, \quad\text{etc.}
\end{equation*}
With this notation, we distinguish
\begin{equation*}
A_1 A_2 \cdots A_n = AA \cdots A = A \otimes A \otimes\cdots\otimes A = A^{\otimes n}
\end{equation*}
from \eqref{pauli2}, where in the latter the tensor factors $A^1,\dots,A^n$ are allowed to be different.

When there is a danger to confuse the tensor product and the matrix product, we will use $\cdot$ for the product of matrices. We have:
\begin{equation}\label{eq:multiplication_table}
X \cdot Y = i Z = -Y \cdot X, \qquad
Y \cdot Z = i X = -Z \cdot Y, \qquad
Z \cdot X = i Y = -X \cdot Z,
\end{equation}
and each Pauli matrix squares to the identity:
\begin{equation*}
X \cdot X = Y \cdot Y = Z \cdot Z = I.
\end{equation*}
The matrix product of Pauli strings is done componentwise:
\begin{equation}\label{eq:componenwise}
(A^1 \otimes \dots \otimes A^n) \cdot (B^1 \otimes \dots \otimes B^n)
= (A^1 \cdot B^1) \otimes \dots \otimes (A^n \cdot B^n).
\end{equation}
From here, it is easy to deduce the following important property of Pauli strings:
\begin{equation}\label{lemp1}
a \cdot a = I^{\otimes n}, \quad a \cdot b = \pm b \cdot a, \qquad a,b\in \mcP_n.
\end{equation}
Hence, any two Pauli strings either commute or anti-commute, which leads to
\begin{align}\label{eq:2ab}
[a,b] = 2 a\cdot b, \qquad\text{if }\; a,b\in \mcP_n, \;\; [a,b]\ne 0.
\end{align}
Another useful identity, which follows from Euler's formula, is
\begin{equation}\label{equab}
e^{i\frac{\pi}{4} a} \, b \, e^{-i \frac{\pi}{4} a} 
= e^{i\frac{\pi}{2} a} \, b = i a \cdot b,
\qquad\text{if }\; a,b\in \mcP_n, \;\; [a,b]\ne 0.
\end{equation}

\section{Dynamical Lie algebras of one-dimensional spin chains}\label{sec:list_dla}

In~\cite{wiersema2023classification}, we classified all dynamical Lie algebras (DLAs) of 1-dimensional 2-local spin chains. At the core of this classification was the identification of all unique generators on 2 spins and their Lie algebras, which we copy here for reference (see \Cref{tab:algebras_with_names}).

\begin{table}[ht]
\begin{center}
\begin{tabular}{|c|l|l|} 
\hline
\textbf{Label}   &   \textbf{Generating Set}  & \textbf{Basis} 
\\	\hline
$\aa_0$       &	$XX$	& $XX$ 
\\	\hline
$\aa_1$       &	$XY$	& $XY$
\\	\hline
$\aa_2$       &	$XY, YX$	& $XY, YX$
\\ \hline
$\aa_3$       &	$XX, YZ$	& $XX, YZ$
\\	\hline
$\aa_4$       &	$XX, YY$	& $XX, YY$
\\	 \hline
$\aa_5$       &	$XY, YZ$	& $XY, YZ$
\\	\hline
$\aa_6$   	&	$XX, YZ, ZY$	& $XX, YZ, ZY$
\\	\hline
$\aa_7$   	&	$XX, YY, ZZ$	& $XX, YY, ZZ$
\\	\hline
$\aa_8$   	&	$XX, XZ$	& $XX, XZ, IY$
\\	\hline
$\aa_9$   	&	$XY, XZ$	& $XY, XZ, IX$
\\	\hline
$\aa_{10}$    &	$XY, YZ, ZX$	& $XY, YZ, ZX$
\\	\hline
$\aa_{11}$	&	$XY, YX, YZ$	& $XY, YX, YZ, IY$
\\	\hline
$\aa_{12}$	&	$XX, XY, YZ$	& $XX, XY, YZ, IZ$
\\	\hline
$\aa_{13}$	&	$XX, YY, YZ$	& $XX, YY, YZ, IX$
\\	\hline
$\aa_{14}$	&	$XX, XY, YX$	& $XX, YY, XY, YX, ZI, IZ$
\\	\hline
$\aa_{15}$	&	$XX, XY, XZ$	& $XX, XY, XZ, IX, IY, IZ$
\\	\hline
$\aa_{16}$	&	$XY, YX, YZ, ZY$	& $XY, YX, YZ, ZY, YI, IY$
\\	\hline
$\aa_{17}$	&	$XX, XY, ZX$	& $XX, XY, ZX, ZY, YI, IZ$
\\	\hline
$\aa_{18}$	&	$XX, XZ, YY, ZY$	& $XX, YY, XZ, ZY, XI, IY$
\\	\hline
$\aa_{19}$	&	$XX, XY, ZX, YZ$ & $XX, XY, ZX, ZY, YZ, YI, IZ$
\\	\hline
$\aa_{20}$	&	$XX, YY, YZ, ZY$	& $XX, YY, ZZ, YZ, ZY, XI, IX$
 \\	\hline
$\aa_{21}$	&	$XX, YY, XY, ZX$	& $XX, YY, XY, YX, ZX, ZY, XI, YI, ZI, IZ$
\\	\hline
$\bb_0$	    &	$XI, IX$	& $XI, IX$ 
\\	\hline
$\bb_1$   	&	$XX, XI, IX$	& $XX, XI, IX$
\\	\hline
$\bb_2$   	&	$XY, XI, IX$	&  $XY, XZ, XI, IX$
\\	\hline
$\bb_3$   	&	$XI, YI, IX, IY$	& $XI, YI, ZI, IX, IY, IZ$
\\	\hline
$\bb_4$   	&	$XX, XY, XI,IX$ &	$XX, XY, XZ, XI, IX, IY, IZ$
\\ \hline
\end{tabular}
\end{center}
\caption{List of unique proper Lie subalgebras of $\su(4)$ generated by 2-local Pauli strings. Reproduced from Ref.~\cite{wiersema2023classification}.}
\label{tab:algebras_with_names}
\end{table}

For convenience, we also reproduce here the classification of dynamical Lie algebras on complete graphs from \cite[Theorem IV.3]{wiersema2023classification},
where $n\ge3$:
\allowdisplaybreaks

\begin{align*}
\aa_0^\pi(n) &\cong \uu(1)^{\oplus n(n-1)/2}, \\
\aa_2^\pi(n) &\cong \so(2^{n-1})^{\oplus2}, \\
\aa_4^\pi(n) &= \aa_7^\pi(n) \cong \begin{cases} 
\su(2^{n-1}), & n \;\;\mathrm{odd}, \\
\su(2^{n-2})^{\oplus 4}, & n \;\;\mathrm{even},
\end{cases} \\
\aa_6^\pi(n) &= \aa_{20}^\pi(n) \cong \aa_{14}^\pi(n) \cong \su(2^{n-1})^{\oplus2}, \\
\aa_{16}^\pi(n) &= \so(2^n), \\
\aa_{22}^\pi(n) &= \su(2^n).
\end{align*}

\clearpage
\section{Frustration graphs}
\label{sec:frustr}

In this section, we introduce the concept of a frustration graph of a set of Pauli strings $\mcA$, and illustrate its use for producing elements of the dynamical Lie algebra $\Lie{\mcA}$ generated by $\mcA$. 

\subsection{Colored Frustration Graphs and Operations on Them}

In the main text, we introduced the notion of equivalence for interaction graphs, which means that if we place the generators of the DLA on the edges of two graphs, we would obtain the same Lie algebra. We were concerned with the equivalence of an interaction graph with another graph obtained by adding more edges to it. Hence, the DLA of the first interaction graph is a subalgebra of the DLA of the second. To prove that the two DLAs are equal, we had to show that every generator of the second DLA, corresponding to new edges of the second interaction graph, is also contained in the first DLA. To this end, we had to check that particular elements lie in that DLA.

For example, in Lemma \ref{lem:deltak23}, we had to check that $X_1 Y_4$ (corresponding to the new edge $(1,4)$ in the interaction graph) can be obtained from the generators of the DLA. We verified this by presenting an explicit expression \eqref{eq:X1Y4}:
\begin{equation}\label{eq:X1Y4-app}
    X_1 Y_4 \equiv \ad_{Y_1 X_2} \ad_{X_3Y_2} \ad_{X_5Y_2} \ad_{Y_1 X_2}\ad_{X_2 Y_3} \ad_{X_5Y_2} \ad_{X_3 Y_4} \ad_{X_1 Y_2} (X_2 Y_3).
\end{equation}
Here and below, we use the notation $\equiv$ to indicate that the two sides are equal up to a non-zero scalar multiple.
Instead of checking Eq.\ \eqref{eq:X1Y4-app} by a tedious calculation,
in this appendix we develop a graphical calculus for easily manipulating commutators of Pauli strings.
It is based on the notion of frustration graph, originally introduced in \cite{chapman2020characterization} and already successfully employed for studying DLAs in \cite{wiersema2023classification, aguilar2024full}.

\begin{definition}[Frustration graph]
    The \emph{frustration graph} of a set of $n$-qubit Pauli strings $\mcA = \{a_1, \dots, a_M\} \subseteq \mcP_n$ is a graph with $\mcA$ as the set of vertices
    and edges $(a_i, a_j)$ for all pairs $i,j$ such that $[a_i, a_j] \neq 0$.
    We denote the frustration graph of $\mcA$ as $\Gamma(\mcA)$.
\end{definition}

For example, for the Pauli matrices $X,Y,Z$,
we have the frustration graph
\begin{align*}
\Gamma(\{X,Y,Z\}) =\:
\begin{tikzpicture}[baseline=0mm, scale=1.0] 
\begin{scope}[every node/.style={circle,thick,draw,scale=0.5, color=fgraph}]
    \node (A) at (0,-0.35) {};
    \node (B) at (1,-0.35) {};
    \node (C) at (0.5,0.35) {};
\end{scope}
\node[below] at (A.south) {$X$};
\node[below] at (B.south) {$Y$};
\node[above] at (C.north) {$Z$};
\begin{scope}[every edge/.style={draw=black,very thick, color=fgraph}]
    \path [-] (A) edge (B);
    \path [-] (A) edge (C);
    \path [-] (C) edge (B);
\end{scope}
\end{tikzpicture}.
\end{align*}
Each edge is connected since $X$, $Y$ and $Z$ do not commute with each other.
As another example, consider the generators $\{X_1 X_2, X_1 Y_2, Y_1 X_2\}$ of $\aa_{14}^{L_2}$ for the interaction graph $L_2$ (the line graph with 2 vertices). The frustration graph is then given by
\begin{align*}
\Gamma(\aa_{14}^{L_2}) = 
\begin{tikzpicture}[scale=1.0]
\begin{scope}[every node/.style={circle,thick,draw,scale=0.5, color=fgraph}]
    \node (A) at (0,0) {};
    \node (B) at (1,0) {};
    \node (C) at (2,0) {};
\end{scope}
\node[above] at (A.north) {$Y_1 X_2 $};
\node[above] at (B.north) {$X_1 X_2$};
\node[above] at (C.north) {$X_1 Y_2$};
\begin{scope}[every edge/.style={draw=black,very thick, color=fgraph}]
    \path [-] (A) edge (B);
    \path [-] (B) edge (C);
\end{scope}
\end{tikzpicture}.
\end{align*}
In the following, whenever we refer to the frustration graph of one of the DLAs $\aa_k^G$, we will always use the generators given in Table \ref{tab:algebras_with_names}, with these generators applied on every edge of the interaction graph $G$, as explained in Sec.~\ref{sec:preliminaries}.  

When computing nested commutators like in Eq.\ \eqref{eq:X1Y4-app}, we want each commutator to be non-zero, i.e., to correspond to an edge in the frustration graph. 
In this case, Eq.\ \eqref{eq:2ab} allows us to express the commutator as the product (up to a non-zero scalar multiple): $[a,b] \equiv a \cdot b$ whenever $[a,b] \ne 0$.

\begin{lemma}\label{lem:generator_product}
    Given a set of Pauli strings $\mcA = \{a_1,\dots,a_M\}$, every Pauli string in the dynamical Lie algebra $\Lie{\mcA}$ can be written as a product $a_1^{k_1}\cdot a_2^{k_2} \cdots a_M^{k_M}$ up to a non-zero scalar factor, where $k_i \in \{0,1\}$. 
\end{lemma}

\begin{proof}
Recall that, by definition, any element in $\Lie{\mcA}$ is a 
real linear combination of nested commutators of elements of $i\mcA$ 
(cf.\ Eq.\ \eqref{eq:adj}):
\begin{align*}
\ad_{ia_{m_1}} \cdots \ad_{ia_{m_r}} (ia_{m_{r+1}}) = i^{r+1} [a_{m_1},[a_{m_2},[\cdots[a_{m_r},a_{m_{r+1}}]\cdots]]] \qquad (0\le r, \; 1\le m_1,\dots,m_{r+1} \le M).
\end{align*}
By Eq.\ \eqref{eq:2ab}, such a nested commutator is either $0$ or congruent to the product:
\begin{align}\label{eq:basis_element}
\ad_{ia_{m_1}} \cdots \ad_{ia_{m_r}} (ia_{m_{r+1}}) 
\equiv a_{m_1}  \cdots  a_{m_r} \cdot a_{m_{r+1}}.
\end{align}
If two indices are equal, say $m_i=m_j$, then we can simplify the product using that
$a_{m_j} \cdot a_{m_j} = 1$ (cf.\ Eq.\ \eqref{lemp1}). Thus, we obtain a product in which all the factors are distinct, and up to a sign, we can reorder the indices in Eq.\ \eqref{eq:basis_element} in a strictly increasing order.

Finally, note that expression \eqref{eq:basis_element} is a scalar multiple of a Pauli string. As different Pauli strings are linearly independent, every Pauli string in $\Lie{\mcA}$ is obtained as in Eq.\ \eqref{eq:basis_element}.
\end{proof}

Lemma \ref{lem:generator_product} allows us to represent each Pauli string in $\Lie{\mcA}$ as a bit-string $k_1\cdots k_M \in \{0,1\}^M$. However, in general, not every bit-string corresponds to an element of $\Lie{\mcA}$. We can use the frustration graph of $\mcA$ to derive such elements. To this end, it will be convenient to represent bit-strings as colorings of the frustration graph.

\begin{definition}[Colored frustration graph]\label{def:colfr}
    Consider a set of Pauli strings $\mcA = \{a_1, \ldots, a_M\}$ and its corresponding frustration graph $\Gamma(\mcA)$.
    Any subset $\mathcal{C}\subseteq\mcA$ of vertices will be called a
    \emph{coloring} of $\Gamma(\mcA)$, and we say that $(\Gamma(\mcA), \mathcal{C})$ is a \emph{colored frustration graph}. The vertices in $\mathcal{C}$ will be depicted as filled blue circles, while the rest of the vertices will be represented as hollow circles.
    A product of generators $c=a_1^{k_1}\cdot a_2^{k_2} \cdots a_M^{k_M}$, with $k_i \in \{0,1\}$, corresponds to the coloring of $\Gamma(\mcA)$ given by $\mathcal{C} = \{a_i \mid k_i = 1\}$.
\end{definition}
As an example, consider the frustration graph
(which is the same as $\Gamma(\aa_{14}^{L_2})$ above):
\begin{align}\label{eq:example_frustration_graph}
\Gamma(\{a_1,a_2,a_3\}) = \begin{tikzpicture}[scale=1.0]
\begin{scope}[every node/.style={circle,thick,draw,scale=0.5, color=fgraph}]
    \node (A) at (0,0) {};
    \node (B) at (1,0) {};
    \node (C) at (2,0) {};
\end{scope}
\node[above] at (A.north) {$a_1$};
\node[above] at (B.north) {$a_2$};
\node[above] at (C.north) {$a_3$};
\begin{scope}[every edge/.style={draw=black,very thick, color=fgraph}]
    \path [-] (A) edge (B);
    \path [-] (B) edge (C);
\end{scope}
\end{tikzpicture},
\end{align}
where $a_1,a_2,a_3$ are Pauli strings such that $[a_1,a_2]\neq0$, $[a_2,a_3]\neq 0$ and $[a_1,a_3]=0$. 
Then all possible products of the generators $a_1,a_2,a_3$
are represented by the following colored frustration graphs: 
\begin{align*}
1 = 
\begin{tikzpicture}[scale=1.0] 
\begin{scope}[every node/.style={circle,thick,draw,scale=0.5, color=fgraph}]
    \node (A) at (0,0) {};
    \node (B) at (1,0) {};
    \node (C) at (2,0) {};
\end{scope}
\node[above] at (A.north) {$a_1$};
\node[above] at (B.north) {$a_2$};
\node[above] at (C.north) {$a_3$};
\begin{scope}[every edge/.style={draw=black,very thick, color=fgraph}]
    \path [-] (A) edge (B);
    \path [-] (B) edge (C);
\end{scope}
\end{tikzpicture}, \qquad
a_1 = 
\begin{tikzpicture}[scale=1.0] 
\begin{scope}[every node/.style={circle,thick,draw,scale=0.5, color=fgraph}]
    \node[fill = fgraph] (A) at (0,0) {};
    \node (B) at (1,0) {};
    \node (C) at (2,0) {};
\end{scope}
\node[above] at (A.north) {$a_1$};
\node[above] at (B.north) {$a_2$};
\node[above] at (C.north) {$a_3$};
\begin{scope}[every edge/.style={draw=black,very thick, color=fgraph}]
    \path [-] (A) edge (B);
    \path [-] (B) edge (C);
\end{scope}
\end{tikzpicture}, \qquad
a_2 = 
\begin{tikzpicture}[scale=1.0] 
\begin{scope}[every node/.style={circle,thick,draw,scale=0.5, color=fgraph}]
    \node (A) at (0,0) {};
    \node[fill = fgraph] (B) at (1,0) {};
    \node (C) at (2,0) {};
\end{scope}
\node[above] at (A.north) {$a_1$};
\node[above] at (B.north) {$a_2$};
\node[above] at (C.north) {$a_3$};
\begin{scope}[every edge/.style={draw=black,very thick, color=fgraph}]
    \path [-] (A) edge (B);
    \path [-] (B) edge (C);
\end{scope}
\end{tikzpicture}, \qquad
a_3 = 
\begin{tikzpicture}[scale=1.0] 
\begin{scope}[every node/.style={circle,thick,draw,scale=0.5, color=fgraph}]
    \node (A) at (0,0) {};
    \node (B) at (1,0) {};
    \node[fill = fgraph]  (C) at (2,0) {};
\end{scope}
\node[above] at (A.north) {$a_1$};
\node[above] at (B.north) {$a_2$};
\node[above] at (C.north) {$a_3$};
\begin{scope}[every edge/.style={draw=black,very thick, color=fgraph}]
    \path [-] (A) edge (B);
    \path [-] (B) edge (C);
\end{scope}
\end{tikzpicture},
\end{align*}
\begin{align*}
a_1 \cdot a_2 = 
\begin{tikzpicture}[scale=1.0] 
\begin{scope}[every node/.style={circle,thick,draw,scale=0.5, color=fgraph}]
    \node[fill = fgraph] (A) at (0,0) {};
    \node[fill = fgraph] (B) at (1,0) {};
    \node (C) at (2,0) {};
\end{scope}
\node[above] at (A.north) {$a_1$};
\node[above] at (B.north) {$a_2$};
\node[above] at (C.north) {$a_3$};
\begin{scope}[every edge/.style={draw=black,very thick, color=fgraph}]
    \path [-] (A) edge (B);
    \path [-] (B) edge (C);
\end{scope}
\end{tikzpicture}, \;\;
a_2 \cdot a_3 = 
\begin{tikzpicture}[scale=1.0] 
\begin{scope}[every node/.style={circle,thick,draw,scale=0.5, color=fgraph}]
    \node (A) at (0,0) {};
    \node[fill = fgraph] (B) at (1,0) {};
    \node[fill = fgraph] (C) at (2,0) {};
\end{scope}
\node[above] at (A.north) {$a_1$};
\node[above] at (B.north) {$a_2$};
\node[above] at (C.north) {$a_3$};
\begin{scope}[every edge/.style={draw=black,very thick, color=fgraph}]
    \path [-] (A) edge (B);
    \path [-] (B) edge (C);
\end{scope}
\end{tikzpicture}, \;\;
a_1 \cdot a_3 = 
\begin{tikzpicture}[scale=1.0] 
\begin{scope}[every node/.style={circle,thick,draw,scale=0.5, color=fgraph}]
    \node[fill = fgraph] (A) at (0,0) {};
    \node (B) at (1,0) {};
    \node[fill = fgraph]  (C) at (2,0) {};
\end{scope}
\node[above] at (A.north) {$a_1$};
\node[above] at (B.north) {$a_2$};
\node[above] at (C.north) {$a_3$};
\begin{scope}[every edge/.style={draw=black,very thick, color=fgraph}]
    \path [-] (A) edge (B);
    \path [-] (B) edge (C);
\end{scope}
\end{tikzpicture}, \;\;
a_1 \cdot a_2 \cdot a_3 = 
\begin{tikzpicture}[scale=1.0] 
\begin{scope}[every node/.style={circle,thick,draw,scale=0.5, color=fgraph}]
    \node[fill = fgraph] (A) at (0,0) {};
    \node[fill = fgraph] (B) at (1,0) {};
    \node[fill = fgraph]  (C) at (2,0) {};
\end{scope}
\node[above] at (A.north) {$a_1$};
\node[above] at (B.north) {$a_2$};
\node[above] at (C.north) {$a_3$};
\begin{scope}[every edge/.style={draw=black,very thick, color=fgraph}]
    \path [-] (A) edge (B);
    \path [-] (B) edge (C);
\end{scope}
\end{tikzpicture}.
\end{align*}
Note that if we reorder the factors in a product of Pauli strings, the product will stay the same up to a sign; hence it is represented by the same colored frustration graph.

Starting from a generating set $\mcA$, the DLA $\Lie{\mcA}$ contains (up to a scalar) all elements $a_i \in \mcA$; their corresponding colored frustration graphs have just one colored vertex $a_i$. In order to generate other Pauli strings in $\Lie{\mcA}$, we need to apply nested commutators as in Eq.\ \eqref{eq:basis_element}. This operation can be expressed in terms of colored frustration graphs as follows.
%

\begin{definition}[Colored frustration graph operations]\label{def:frust_ops}
    Consider a set of Pauli strings $\mcA = \{a_1, \ldots, a_M\}$ and a colored frustration graph $(\Gamma(\mcA),\mathcal{C})$ corresponding to a Pauli string $c$, as in Definition \ref{def:colfr}.
    For any vertex $a_i$ connected with an odd number of edges to the colored vertices $\mathcal{C}$, 
    %
    we can perform one of the following two operations.
    \begin{enumerate}
    \item \textbf{Adding $a_i$.} If $a_i\not\in\mathcal{C}$ is a hollow vertex, we can add it to the colored vertices $\mathcal{C}{:}$
    \begin{align*} 
    c = 
        \begin{tikzpicture}[baseline=-1mm, scale=1.0] 
        \begin{scope}[every node/.style={circle,thick,draw,scale=0.5, color=fgraph}]
            \node(B) at (0,0) {};
            \node[coordinate] (C1) at (0.4,0.3) {};
            \node[coordinate] (C2) at (0.4,-0.3) {};
        \end{scope}
        \node[above] at (B.north) {$a_i$};
        \begin{scope}[every edge/.style={draw=black,very thick, color=fgraph}]
            \path [-] (C1) edge (B);
            \path [-] (C2) edge (B);
            \draw [-] ($(C1) + (-1pt, -3pt)$) edge[dashed, dash pattern=on 1pt off 1pt] ($(C2) + (-1pt, 3pt)$) ;
        \end{scope}
        \end{tikzpicture}
        \quad \to \quad a_i \cdot c = 
        \begin{tikzpicture}[baseline=-1mm, scale=1.0] 
        \begin{scope}[every node/.style={circle,thick,draw,scale=0.5, color=fgraph}]
            \node[fill=fgraph] (B) at (0,0) {};
            \node[coordinate] (C1) at (0.4,0.3) {};
            \node[coordinate] (C2) at (0.4,-0.3) {};
        \end{scope}
        \node[above] at (B.north) {$a_i$};
        \begin{scope}[every edge/.style={draw=black,very thick, color=fgraph}]
            \path [-] (C1) edge (B);
            \path [-] (C2) edge (B);
            \draw [-] ($(C1) + (-1pt, -3pt)$) edge[dashed, dash pattern=on 1pt off 1pt] ($(C2) + (-1pt, 3pt)$) ;
        \end{scope}
        \end{tikzpicture}\,.
    \end{align*}
    \item \textbf{Removing $a_i$.}  If $a_i\in\mathcal{C}$ is a colored vertex, we can remove it from the colored vertices $\mathcal{C}{:}$
    \begin{align*}
    c = 
        \begin{tikzpicture}[baseline=-1mm, scale=1.0] 
        \begin{scope}[every node/.style={circle,thick,draw,scale=0.5, color=fgraph}]
            \node[fill=fgraph] (B) at (0,0) {};
            \node[coordinate] (C1) at (0.4,0.3) {};
            \node[coordinate] (C2) at (0.4,-0.3) {};
        \end{scope}
        \node[above] at (B.north) {$a_i$};
        \begin{scope}[every edge/.style={draw=black,very thick, color=fgraph}]
            \path [-] (C1) edge (B);
            \path [-] (C2) edge (B);
            \draw [-] ($(C1) + (-1pt, -3pt)$) edge[dashed, dash pattern=on 1pt off 1pt] ($(C2) + (-1pt, 3pt)$) ;
        \end{scope}
        \end{tikzpicture}
        \quad \to \quad a_i \cdot c =
        \begin{tikzpicture}[baseline=-1mm, scale=1.0] 
        \begin{scope}[every node/.style={circle,thick,draw,scale=0.5, color=fgraph}]
            \node (B) at (0,0) {};
            \node[coordinate] (C1) at (0.4,0.3) {};
            \node[coordinate] (C2) at (0.4,-0.3) {};
        \end{scope}
        \node[above] at (B.north) {$a_i$};
        \begin{scope}[every edge/.style={draw=black,very thick, color=fgraph}]
            \path [-] (C1) edge (B);
            \path [-] (C2) edge (B);
            \draw [-] ($(C1) + (-1pt, -3pt)$) edge[dashed, dash pattern=on 1pt off 1pt] ($(C2) + (-1pt, 3pt)$) ;
        \end{scope}
        \end{tikzpicture}\,.
    \end{align*} 
\end{enumerate}
\end{definition}

Now we can characterize the Pauli strings in the DLA $\Lie{\mcA}$ in terms of colorings of the frustration graph $\Gamma(\mcA)$.

\begin{lemma}\label{lem:dla_col_fr}
    Given a set of Pauli strings $\mcA = \{a_1,\dots,a_M\}$, a Pauli string $c = a_1^{k_1}\cdot a_2^{k_2} \cdots a_M^{k_M}$ is (up to a scalar) in the dynamical Lie algebra $\Lie{\mcA}$ if and only if its corresponding coloring $\mathcal{C} \subseteq\mcA$ of the frustration graph $\Gamma(\mcA)$ can be obtained by a sequence of adding and removing operations as in Definition \ref{def:frust_ops}, starting from a coloring with a single vertex.
\end{lemma}

\begin{proof}
First, note that a Pauli string $a_i$ anti-commutes with those $a_j$ connected to it in $\Gamma(\mcA)$; hence, $a_i$ anti-commutes with $c$ precisely when $a_i$ is connected to an odd number of colored vertices in $\mathcal{C}$. In this case,
$[a_i,c] = 2 a_i \cdot c \equiv a_i \cdot c$ corresponds to either adding $a_i$ to $\mathcal{C}$ (when $a_i\not\in\mathcal{C}$) or removing it (when $a_i\in\mathcal{C}$).

In the proof of Lemma \ref{lem:generator_product}, we showed that a Pauli string $c$ is in the DLA $\Lie{\mcA}$ if and only if $c$ is (up to a scalar) of the form in Eq.\ \eqref{eq:basis_element}. Such an element is obtained by a sequence of operations as in Definition \ref{def:frust_ops}, where we start with the single colored vertex $a_{m_{r+1}}$, then add the vertex $a_{m_r}$, then add/remove $a_{m_{r-1}}$, and so on until we add/remove $a_{m_1}$.
\end{proof}

As an example, consider again the frustration graph \eqref{eq:example_frustration_graph}. 
%
Starting from the colored vertex $a_1$, we can add vertices $a_2$ and $a_3$ as follows:
\begin{align*}
  \begin{tikzpicture}[scale=1.0] 
\begin{scope}[every node/.style={circle,thick,draw,scale=0.5, color=fgraph}]
    \node[fill=fgraph] (A) at (0,0) {};
    \node(B) at (1,0) {};
    \node(C) at (2,0) {};
\end{scope}
\node[above] at (A.north) {$a_1$};
\node[above] at (B.north) {$a_2$};
\node[above] at (C.north) {$a_3$};
\begin{scope}[every edge/.style={draw=black,very thick, color=fgraph}]
    \path [-] (A) edge (B);
    \path [-] (C) edge (B);
\end{scope}
\end{tikzpicture} \quad \rightarrow \quad 
\begin{tikzpicture}[scale=1.0] 
\begin{scope}[every node/.style={circle,thick,draw,scale=0.5, color=fgraph}]
    \node [fill=fgraph](A) at (0,0) {};
    \node[fill=fgraph] (B) at (1,0) {};
    \node(C) at (2,0) {};
\end{scope}
\node[above] at (A.north) {$a_1$};
\node[above] at (B.north) {$a_2$};
\node[above] at (C.north) {$a_3$};
\begin{scope}[every edge/.style={draw=black,very thick, color=fgraph}]
    \path [-] (A) edge (B);
    \path [-] (C) edge (B);
\end{scope}
\end{tikzpicture} 
\quad \rightarrow \quad
\begin{tikzpicture}[scale=1.0] 
\begin{scope}[every node/.style={circle,thick,draw,scale=0.5, color=fgraph}]
    \node[fill = fgraph] (A) at (0,0) {};
    \node[fill = fgraph] (B) at (1,0) {};
    \node[fill = fgraph] (C) at (2,0) {};
\end{scope}
\node[above] at (A.north) {$a_1$};
\node[above] at (B.north) {$a_2$};
\node[above] at (C.north) {$a_3$};
\begin{scope}[every edge/.style={draw=black,very thick, color=fgraph}]
    \path [-] (A) edge (B);
    \path [-] (B) edge (C);
\end{scope}
\end{tikzpicture}\,.
\end{align*}
Therefore, $a_1 \cdot a_2 \cdot a_3 \in \Lie{a_1,a_2,a_3}$.
On the other hand, one can check that $a_1 \cdot a_3 \not\in \Lie{a_1,a_2,a_3}$.
By a similar reasoning, the DLA of a line frustration graph was determined in \cite[Proposition C.1]{wiersema2023classification}.

In the following subsections \ref{frust-eq4} and \ref{frust-eq567}, 
we will apply the operations of adding and removing a vertex to the coloring of 
the frustration graphs $\Gamma(a_k^\Omega)$ and $\Gamma(a_k^\Sigma)$ 
to obtain proofs of Eqs.\ \eqref{eq:X1Y4}--\eqref{eq:Z1Z3}.



\subsection{An alternative proof of Eq.\ \eqref{eq:X1Y4} \label{frust-eq4}} 

In this subsection, using the colored frustration graph operations of adding and removing a vertex given in  Definition~\ref{def:frust_ops},
we will derive two claims from the proof of Lemma \ref{lem:deltak23}: $X_1 Y_4 \in \aa^\Sigma_2$ and $X_1 X_4 \in \aa^\Sigma_{14}$.

\begin{lemma}\label{lemmatree}
We have    $X_1 Y_4 \in \aa^\Sigma_2$.
\end{lemma}
\begin{proof}
Recall that the Lie algebra $\aa^\Sigma_2$ is generated by all $X_iY_j$, where $(i,j)$ are the edges of the interaction graph $\Sigma$ from Lemma \ref{lem:deltak23}.
The frustration graph $\Gamma(\aa^\Sigma_2)$ of this generating set is given by:
\begin{align*}
\Gamma(\aa^\Sigma_2) =\:
\begin{tikzpicture}[baseline=4mm, scale=2.0] 
\begin{scope}[every node/.style={circle,thick,draw,scale=0.5, color=fgraph}]
    \node (X4_Y3) at (0,-0.5) {};
    \node (X2_Y3_) at (0,0) {};
    \node (X2_Y5_) at (0,0.5) {};
    \node (X2_Y1_) at (0,1) {};
    \node (Y4_X3) at (1,-0.5) {};
    \node (Y2_X3_) at (1,0) {};
    \node (Y2_X5_) at (1,0.5) {};
    \node (Y2_X1_) at (1,1) {};
\end{scope}
\node[left] at (X4_Y3.west) {$X_4Y_3$};
\node[left] at (X2_Y3_.west) {$X_2Y_3$};
\node[left] at (X2_Y5_.west) {$X_2Y_5$};
\node[left] at (X2_Y1_.west) {$X_2Y_1$};
\node[right] at (Y4_X3.east) {$X_3Y_4$};
\node[right] at (Y2_X3_.east) {$X_3Y_2$};
\node[right] at (Y2_X5_.east) {$X_5Y_2$};
\node[right] at (Y2_X1_.east) {$X_1Y_2$};
\begin{scope}[every edge/.style={draw=black,very thick, color=fgraph}]
    \path [-] (X2_Y1_) edge (Y2_X3_);
    \path [-] (X2_Y1_) edge (Y2_X5_);
    \path [-] (X2_Y5_) edge (Y2_X3_);
    \path [-] (X2_Y3_) edge (Y2_X5_);
    \path [-] (X2_Y5_) edge (Y2_X1_);
    \path [-] (X2_Y3_) edge (Y2_X1_);
    \path [-] (X2_Y3_) edge (Y4_X3);
    \path [-] (X4_Y3) edge (Y2_X3_);
\end{scope}
\end{tikzpicture}\;.
\end{align*}
%
%
We note that
$X_1 Y_4 = X_3 Y_4  \cdot X_3 Y_2 \cdot X_1 Y_2$ can be written as a product of generators, which corresponds to a coloring
of the frustration graph. Our goal is to obtain this coloring by applying add/remove vertex operations starting from a single colored vertex (cf.\ Lemma \ref{lem:dla_col_fr}). We start with
$X_1 Y_2$ and perform the following operations:
\begin{align*}
\begin{tikzpicture}[baseline=4mm, scale=1.0] 
\begin{scope}[every node/.style={circle,thick,draw,scale=0.5, color=fgraph}]
    \node (X4_Y3) at (0,-0.5) {};
    \node (X2_Y3_) at (0,0) {};
    \node (X2_Y5_) at (0,0.5) {};
    \node (X2_Y1_) at (0,1) {};
    \node (Y4_X3) at (1,-0.5) {};
    \node (Y2_X3_) at (1,0) {};
    \node (Y2_X5_) at (1,0.5) {};
    \node[fill=fgraph] (Y2_X1_) at (1,1) {};
\end{scope}
\node[left] at (X4_Y3.west) {$X_4Y_3$};
\node[left] at (X2_Y3_.west) {$X_2Y_3$};
\node[left] at (X2_Y5_.west) {$X_2Y_5$};
\node[left] at (X2_Y1_.west) {$X_2Y_1$};
\node[right] at (Y4_X3.east) {$X_3Y_4$};
\node[right] at (Y2_X3_.east) {$X_3Y_2$};
\node[right] at (Y2_X5_.east) {$X_5Y_2$};
\node[right] at (Y2_X1_.east) {$X_1Y_2$};
\begin{scope}[every edge/.style={draw=black,very thick, color=fgraph}]
    \path [-] (X2_Y1_) edge (Y2_X3_);
    \path [-] (X2_Y1_) edge (Y2_X5_);
    \path [-] (X2_Y5_) edge (Y2_X3_);
    \path [-] (X2_Y3_) edge (Y2_X5_);
    \path [-] (X2_Y5_) edge (Y2_X1_);
    \path [-] (X2_Y3_) edge (Y2_X1_);
    \path [-] (X2_Y3_) edge (Y4_X3);
    \path [-] (X4_Y3) edge (Y2_X3_);
\end{scope}
\end{tikzpicture}
\to
\begin{tikzpicture}[baseline=4mm, scale=1.0] 
\begin{scope}[every node/.style={circle,thick,draw,scale=0.5, color=fgraph}]
    \node (X4_Y3) at (0,-0.5) {};
    \node[fill=fgraph] (X2_Y3_) at (0,0) {};
    \node (X2_Y5_) at (0,0.5) {};
    \node (X2_Y1_) at (0,1) {};
    \node (Y4_X3) at (1,-0.5) {};
    \node (Y2_X3_) at (1,0) {};
    \node (Y2_X5_) at (1,0.5) {};
    \node[fill=fgraph] (Y2_X1_) at (1,1) {};
\end{scope}
\node[left] at (X4_Y3.west) {$X_4Y_3$};
\node[left] at (X2_Y3_.west) {$X_2Y_3$};
\node[left] at (X2_Y5_.west) {$X_2Y_5$};
\node[left] at (X2_Y1_.west) {$X_2Y_1$};
\node[right] at (Y4_X3.east) {$X_3Y_4$};
\node[right] at (Y2_X3_.east) {$X_3Y_2$};
\node[right] at (Y2_X5_.east) {$X_5Y_2$};
\node[right] at (Y2_X1_.east) {$X_1Y_2$};
\begin{scope}[every edge/.style={draw=black,very thick, color=fgraph}]
    \path [-] (X2_Y1_) edge (Y2_X3_);
    \path [-] (X2_Y1_) edge (Y2_X5_);
    \path [-] (X2_Y5_) edge (Y2_X3_);
    \path [-] (X2_Y3_) edge (Y2_X5_);
    \path [-] (X2_Y5_) edge (Y2_X1_);
    \path [-] (X2_Y3_) edge (Y2_X1_);
    \path [-] (X2_Y3_) edge (Y4_X3);
    \path [-] (X4_Y3) edge (Y2_X3_);
\end{scope}
\end{tikzpicture}
\to
\begin{tikzpicture}[baseline=4mm, scale=1.0] 
\begin{scope}[every node/.style={circle,thick,draw,scale=0.5, color=fgraph}]
    \node (X4_Y3) at (0,-0.5) {};
    \node[fill=fgraph] (X2_Y3_) at (0,0) {};
    \node (X2_Y5_) at (0,0.5) {};
    \node (X2_Y1_) at (0,1) {};
    \node[fill=fgraph] (Y4_X3) at (1,-0.5) {};
    \node (Y2_X3_) at (1,0) {};
    \node (Y2_X5_) at (1,0.5) {};
    \node[fill=fgraph] (Y2_X1_) at (1,1) {};
\end{scope}
\node[left] at (X4_Y3.west) {$X_4Y_3$};
\node[left] at (X2_Y3_.west) {$X_2Y_3$};
\node[left] at (X2_Y5_.west) {$X_2Y_5$};
\node[left] at (X2_Y1_.west) {$X_2Y_1$};
\node[right] at (Y4_X3.east) {$X_3Y_4$};
\node[right] at (Y2_X3_.east) {$X_3Y_2$};
\node[right] at (Y2_X5_.east) {$X_5Y_2$};
\node[right] at (Y2_X1_.east) {$X_1Y_2$};
\begin{scope}[every edge/.style={draw=black,very thick, color=fgraph}]
    \path [-] (X2_Y1_) edge (Y2_X3_);
    \path [-] (X2_Y1_) edge (Y2_X5_);
    \path [-] (X2_Y5_) edge (Y2_X3_);
    \path [-] (X2_Y3_) edge (Y2_X5_);
    \path [-] (X2_Y5_) edge (Y2_X1_);
    \path [-] (X2_Y3_) edge (Y2_X1_);
    \path [-] (X2_Y3_) edge (Y4_X3);
    \path [-] (X4_Y3) edge (Y2_X3_);
\end{scope}
\end{tikzpicture}
\to
\begin{tikzpicture}[baseline=4mm, scale=1.0] 
\begin{scope}[every node/.style={circle,thick,draw,scale=0.5, color=fgraph}]
    \node (X4_Y3) at (0,-0.5) {};
    \node[fill=fgraph] (X2_Y3_) at (0,0) {};
    \node (X2_Y5_) at (0,0.5) {};
    \node (X2_Y1_) at (0,1) {};
    \node[fill=fgraph] (Y4_X3) at (1,-0.5) {};
    \node (Y2_X3_) at (1,0) {};
    \node[fill=fgraph] (Y2_X5_) at (1,0.5) {};
    \node[fill=fgraph] (Y2_X1_) at (1,1) {};
\end{scope}
\node[left] at (X4_Y3.west) {$X_4Y_3$};
\node[left] at (X2_Y3_.west) {$X_2Y_3$};
\node[left] at (X2_Y5_.west) {$X_2Y_5$};
\node[left] at (X2_Y1_.west) {$X_2Y_1$};
\node[right] at (Y4_X3.east) {$X_3Y_4$};
\node[right] at (Y2_X3_.east) {$X_3Y_2$};
\node[right] at (Y2_X5_.east) {$X_5Y_2$};
\node[right] at (Y2_X1_.east) {$X_1Y_2$};
\begin{scope}[every edge/.style={draw=black,very thick, color=fgraph}]
    \path [-] (X2_Y1_) edge (Y2_X3_);
    \path [-] (X2_Y1_) edge (Y2_X5_);
    \path [-] (X2_Y5_) edge (Y2_X3_);
    \path [-] (X2_Y3_) edge (Y2_X5_);
    \path [-] (X2_Y5_) edge (Y2_X1_);
    \path [-] (X2_Y3_) edge (Y2_X1_);
    \path [-] (X2_Y3_) edge (Y4_X3);
    \path [-] (X4_Y3) edge (Y2_X3_);
\end{scope}
\end{tikzpicture}
\to
\begin{tikzpicture}[baseline=4mm, scale=1.0] 
\begin{scope}[every node/.style={circle,thick,draw,scale=0.5, color=fgraph}]
    \node (X4_Y3) at (0,-0.5) {};
    \node (X2_Y3_) at (0,0) {};
    \node (X2_Y5_) at (0,0.5) {};
    \node (X2_Y1_) at (0,1) {};
    \node[fill=fgraph] (Y4_X3) at (1,-0.5) {};
    \node (Y2_X3_) at (1,0) {};
    \node[fill=fgraph] (Y2_X5_) at (1,0.5) {};
    \node[fill=fgraph] (Y2_X1_) at (1,1) {};
\end{scope}
\node[left] at (X4_Y3.west) {$X_4Y_3$};
\node[left] at (X2_Y3_.west) {$X_2Y_3$};
\node[left] at (X2_Y5_.west) {$X_2Y_5$};
\node[left] at (X2_Y1_.west) {$X_2Y_1$};
\node[right] at (Y4_X3.east) {$X_3Y_4$};
\node[right] at (Y2_X3_.east) {$X_3Y_2$};
\node[right] at (Y2_X5_.east) {$X_5Y_2$};
\node[right] at (Y2_X1_.east) {$X_1Y_2$};
\begin{scope}[every edge/.style={draw=black,very thick, color=fgraph}]
    \path [-] (X2_Y1_) edge (Y2_X3_);
    \path [-] (X2_Y1_) edge (Y2_X5_);
    \path [-] (X2_Y5_) edge (Y2_X3_);
    \path [-] (X2_Y3_) edge (Y2_X5_);
    \path [-] (X2_Y5_) edge (Y2_X1_);
    \path [-] (X2_Y3_) edge (Y2_X1_);
    \path [-] (X2_Y3_) edge (Y4_X3);
    \path [-] (X4_Y3) edge (Y2_X3_);
\end{scope}
\end{tikzpicture}
\end{align*}
Note that we could remove $X_2Y_3$ in the last step because it was connected to three other colored vertices. Continuing,
we get:
\begin{align*}
\to
\begin{tikzpicture}[baseline=4mm, scale=1.0] 
\begin{scope}[every node/.style={circle,thick,draw,scale=0.5, color=fgraph}]
    \node (X4_Y3) at (0,-0.5) {};
    \node (X2_Y3_) at (0,0) {};
    \node (X2_Y5_) at (0,0.5) {};
    \node[fill = fgraph] (X2_Y1_) at (0,1) {};
    \node[fill=fgraph] (Y4_X3) at (1,-0.5) {};
    \node (Y2_X3_) at (1,0) {};
    \node[fill=fgraph] (Y2_X5_) at (1,0.5) {};
    \node[fill=fgraph] (Y2_X1_) at (1,1) {};
\end{scope}
\node[left] at (X4_Y3.west) {$X_4Y_3$};
\node[left] at (X2_Y3_.west) {$X_2Y_3$};
\node[left] at (X2_Y5_.west) {$X_2Y_5$};
\node[left] at (X2_Y1_.west) {$X_2Y_1$};
\node[right] at (Y4_X3.east) {$X_3Y_4$};
\node[right] at (Y2_X3_.east) {$X_3Y_2$};
\node[right] at (Y2_X5_.east) {$X_5Y_2$};
\node[right] at (Y2_X1_.east) {$X_1Y_2$};
\begin{scope}[every edge/.style={draw=black,very thick, color=fgraph}]
    \path [-] (X2_Y1_) edge (Y2_X3_);
    \path [-] (X2_Y1_) edge (Y2_X5_);
    \path [-] (X2_Y5_) edge (Y2_X3_);
    \path [-] (X2_Y3_) edge (Y2_X5_);
    \path [-] (X2_Y5_) edge (Y2_X1_);
    \path [-] (X2_Y3_) edge (Y2_X1_);
    \path [-] (X2_Y3_) edge (Y4_X3);
    \path [-] (X4_Y3) edge (Y2_X3_);
\end{scope}
\end{tikzpicture}
\to
\begin{tikzpicture}[baseline=4mm, scale=1.0] 
\begin{scope}[every node/.style={circle,thick,draw,scale=0.5, color=fgraph}]
    \node (X4_Y3) at (0,-0.5) {};
    \node (X2_Y3_) at (0,0) {};
    \node (X2_Y5_) at (0,0.5) {};
    \node[fill = fgraph] (X2_Y1_) at (0,1) {};
    \node[fill=fgraph] (Y4_X3) at (1,-0.5) {};
    \node (Y2_X3_) at (1,0) {};
    \node (Y2_X5_) at (1,0.5) {};
    \node[fill=fgraph] (Y2_X1_) at (1,1) {};
\end{scope}
\node[left] at (X4_Y3.west) {$X_4Y_3$};
\node[left] at (X2_Y3_.west) {$X_2Y_3$};
\node[left] at (X2_Y5_.west) {$X_2Y_5$};
\node[left] at (X2_Y1_.west) {$X_2Y_1$};
\node[right] at (Y4_X3.east) {$X_3Y_4$};
\node[right] at (Y2_X3_.east) {$X_3Y_2$};
\node[right] at (Y2_X5_.east) {$X_5Y_2$};
\node[right] at (Y2_X1_.east) {$X_1Y_2$};
\begin{scope}[every edge/.style={draw=black,very thick, color=fgraph}]
    \path [-] (X2_Y1_) edge (Y2_X3_);
    \path [-] (X2_Y1_) edge (Y2_X5_);
    \path [-] (X2_Y5_) edge (Y2_X3_);
    \path [-] (X2_Y3_) edge (Y2_X5_);
    \path [-] (X2_Y5_) edge (Y2_X1_);
    \path [-] (X2_Y3_) edge (Y2_X1_);
    \path [-] (X2_Y3_) edge (Y4_X3);
    \path [-] (X4_Y3) edge (Y2_X3_);
\end{scope}
\end{tikzpicture}
\to
\begin{tikzpicture}[baseline=4mm, scale=1.0] 
\begin{scope}[every node/.style={circle,thick,draw,scale=0.5, color=fgraph}]
    \node (X4_Y3) at (0,-0.5) {};
    \node (X2_Y3_) at (0,0) {};
    \node (X2_Y5_) at (0,0.5) {};
    \node[fill = fgraph] (X2_Y1_) at (0,1) {};
    \node[fill=fgraph] (Y4_X3) at (1,-0.5) {};
    \node[fill=fgraph] (Y2_X3_) at (1,0) {};
    \node (Y2_X5_) at (1,0.5) {};
    \node[fill=fgraph] (Y2_X1_) at (1,1) {};
\end{scope}
\node[left] at (X4_Y3.west) {$X_4Y_3$};
\node[left] at (X2_Y3_.west) {$X_2Y_3$};
\node[left] at (X2_Y5_.west) {$X_2Y_5$};
\node[left] at (X2_Y1_.west) {$X_2Y_1$};
\node[right] at (Y4_X3.east) {$X_3Y_4$};
\node[right] at (Y2_X3_.east) {$X_3Y_2$};
\node[right] at (Y2_X5_.east) {$X_5Y_2$};
\node[right] at (Y2_X1_.east) {$X_1Y_2$};
\begin{scope}[every edge/.style={draw=black,very thick, color=fgraph}]
    \path [-] (X2_Y1_) edge (Y2_X3_);
    \path [-] (X2_Y1_) edge (Y2_X5_);
    \path [-] (X2_Y5_) edge (Y2_X3_);
    \path [-] (X2_Y3_) edge (Y2_X5_);
    \path [-] (X2_Y5_) edge (Y2_X1_);
    \path [-] (X2_Y3_) edge (Y2_X1_);
    \path [-] (X2_Y3_) edge (Y4_X3);
    \path [-] (X4_Y3) edge (Y2_X3_);
\end{scope}
\end{tikzpicture}
\to
\begin{tikzpicture}[baseline=4mm, scale=1.0] 
\begin{scope}[every node/.style={circle,thick,draw,scale=0.5, color=fgraph}]
    \node (X4_Y3) at (0,-0.5) {};
    \node (X2_Y3_) at (0,0) {};
    \node (X2_Y5_) at (0,0.5) {};
    \node (X2_Y1_) at (0,1) {};
    \node[fill=fgraph] (Y4_X3) at (1,-0.5) {};
    \node[fill=fgraph] (Y2_X3_) at (1,0) {};
    \node (Y2_X5_) at (1,0.5) {};
    \node[fill=fgraph] (Y2_X1_) at (1,1) {};
\end{scope}
\node[left] at (X4_Y3.west) {$X_4Y_3$};
\node[left] at (X2_Y3_.west) {$X_2Y_3$};
\node[left] at (X2_Y5_.west) {$X_2Y_5$};
\node[left] at (X2_Y1_.west) {$X_2Y_1$};
\node[right] at (Y4_X3.east) {$X_3Y_4$};
\node[right] at (Y2_X3_.east) {$X_3Y_2$};
\node[right] at (Y2_X5_.east) {$X_5Y_2$};
\node[right] at (Y2_X1_.east) {$X_1Y_2$};
\begin{scope}[every edge/.style={draw=black,very thick, color=fgraph}]
    \path [-] (X2_Y1_) edge (Y2_X3_);
    \path [-] (X2_Y1_) edge (Y2_X5_);
    \path [-] (X2_Y5_) edge (Y2_X3_);
    \path [-] (X2_Y3_) edge (Y2_X5_);
    \path [-] (X2_Y5_) edge (Y2_X1_);
    \path [-] (X2_Y3_) edge (Y2_X1_);
    \path [-] (X2_Y3_) edge (Y4_X3);
    \path [-] (X4_Y3) edge (Y2_X3_);
\end{scope}
\end{tikzpicture}.
\end{align*}
The last colored frustration graph corresponds to the product $X_1 Y_4 = X_3 Y_4  \cdot X_3Y_2 \cdot X_1Y_2$; hence, $X_1 Y_4  \in \aa^\Sigma_2$. 
\end{proof}

\begin{lemma}
We have $X_1 X_4 \in \aa^\Sigma_{14}$.
\end{lemma}

\begin{proof}
The Lie algebra $\aa_{14}$ is generated by $\{XX,XY,YX\}$, or equivalently by $\{XX,ZI,IZ\}$, since $\aa_{14}=\Lie{XX,XY,YX} = \Lie{XX, ZI, IZ}$. 
Using the latter set of generators, the frustration graph $\Gamma(\aa^\Sigma_{14})$ is given by:
\begin{align*}
\Gamma(\aa^\Sigma_{14}) =\:
\begin{tikzpicture}[baseline=4mm, scale=2.0] 
\begin{scope}[every node/.style={circle,thick,draw,scale=0.5, color=fgraph}]
    \node (X4-X3) at (0,-0.5) {};
    \node (X2-X3) at (0,0) {};
    \node (X2-X5) at (0,0.5) {};
    \node (X2-X1) at (0,1) {};
    \node (Z-4) at (1,-1) {};
    \node (Z-3) at (1,-0.5) {};
    \node (Z-5) at (1,0) {};
    \node (Z-1) at (1,0.5) {};
    \node (Z-2) at (1,1) {};
\end{scope}
\node[left] at (X4-X3.west) {$X_3X_4$};
\node[left] at (X2-X3.west) {$X_2X_3$};
\node[left] at (X2-X5.west) {$X_2X_5$};
\node[left] at (X2-X1.west) {$X_1X_2$};
\node[right] at (Z-4.east) {$Z_4$};
\node[right] at (Z-3.east) {$Z_3$};
\node[right] at (Z-5.east) {$Z_5$};
\node[right] at (Z-1.east) {$Z_1$};
\node[right] at (Z-2.east) {$Z_2$};
\begin{scope}[every edge/.style={draw=black,very thick, color=fgraph}]
    \path [-] (X2-X1) edge (Z-2);
    \path [-] (X2-X1) edge (Z-1);
    \path [-] (X2-X5) edge (Z-2);
    \path [-] (X2-X5) edge (Z-5);
    \path [-] (X2-X3) edge (Z-2);
    \path [-] (X2-X3) edge (Z-3);
    \path [-] (X4-X3) edge (Z-3);
    \path [-] (X4-X3) edge (Z-4);
\end{scope}
\end{tikzpicture}\;.
\end{align*}
%
Note that $X_1 X_4 = X_3 X_4 \cdot X_2 X_3 \cdot X_1 X_2$ corresponds to a coloring of the frustration graph.
We start with the single colored vertex $X_1 X_2$, and perform the following adding/removing operations:
\begin{align*}
\begin{tikzpicture}[baseline=4mm, scale=1.0] 
\begin{scope}[every node/.style={circle,thick,draw,scale=0.5, color=fgraph}]
    \node (X4-X3) at (0,-0.5) {};
    \node (X2-X3) at (0,0) {};
    \node (X2-X5) at (0,0.5) {};
    \node[fill=fgraph] (X2-X1) at (0,1) {};
    \node (Z-4) at (1,-1) {};
    \node (Z-3) at (1,-0.5) {};
    \node (Z-5) at (1,0) {};
    \node (Z-1) at (1,0.5) {};
    \node (Z-2) at (1,1) {};
\end{scope}
\node[left] at (X4-X3.west) {$X_3X_4$};
\node[left] at (X2-X3.west) {$X_2X_3$};
\node[left] at (X2-X5.west) {$X_2X_5$};
\node[left] at (X2-X1.west) {$X_1X_2$};
\node[right] at (Z-4.east) {$Z_4$};
\node[right] at (Z-3.east) {$Z_3$};
\node[right] at (Z-5.east) {$Z_5$};
\node[right] at (Z-1.east) {$Z_1$};
\node[right] at (Z-2.east) {$Z_2$};
\begin{scope}[every edge/.style={draw=black,very thick, color=fgraph}]
    \path [-] (X2-X1) edge (Z-2);
    \path [-] (X2-X1) edge (Z-1);
    \path [-] (X2-X5) edge (Z-2);
    \path [-] (X2-X5) edge (Z-5);
    \path [-] (X2-X3) edge (Z-2);
    \path [-] (X2-X3) edge (Z-3);
    \path [-] (X4-X3) edge (Z-3);
    \path [-] (X4-X3) edge (Z-4);
\end{scope}
\end{tikzpicture}
\to
\begin{tikzpicture}[baseline=4mm, scale=1.0] 
\begin{scope}[every node/.style={circle,thick,draw,scale=0.5, color=fgraph}]
    \node (X4-X3) at (0,-0.5) {};
    \node (X2-X3) at (0,0) {};
    \node (X2-X5) at (0,0.5) {};
    \node[fill=fgraph] (X2-X1) at (0,1) {};
    \node (Z-4) at (1,-1) {};
    \node (Z-3) at (1,-0.5) {};
    \node (Z-5) at (1,0) {};
    \node (Z-1) at (1,0.5) {};
    \node[fill=fgraph] (Z-2) at (1,1) {};
\end{scope}
\node[left] at (X4-X3.west) {$X_3X_4$};
\node[left] at (X2-X3.west) {$X_2X_3$};
\node[left] at (X2-X5.west) {$X_2X_5$};
\node[left] at (X2-X1.west) {$X_1X_2$};
\node[right] at (Z-4.east) {$Z_4$};
\node[right] at (Z-3.east) {$Z_3$};
\node[right] at (Z-5.east) {$Z_5$};
\node[right] at (Z-1.east) {$Z_1$};
\node[right] at (Z-2.east) {$Z_2$};
\begin{scope}[every edge/.style={draw=black,very thick, color=fgraph}]
    \path [-] (X2-X1) edge (Z-2);
    \path [-] (X2-X1) edge (Z-1);
    \path [-] (X2-X5) edge (Z-2);
    \path [-] (X2-X5) edge (Z-5);
    \path [-] (X2-X3) edge (Z-2);
    \path [-] (X2-X3) edge (Z-3);
    \path [-] (X4-X3) edge (Z-3);
    \path [-] (X4-X3) edge (Z-4);
\end{scope}
\end{tikzpicture}
\to
\begin{tikzpicture}[baseline=4mm, scale=1.0] 
\begin{scope}[every node/.style={circle,thick,draw,scale=0.5, color=fgraph}]
    \node (X4-X3) at (0,-0.5) {};
    \node[fill=fgraph] (X2-X3) at (0,0) {};
    \node (X2-X5) at (0,0.5) {};
    \node[fill=fgraph] (X2-X1) at (0,1) {};
    \node (Z-4) at (1,-1) {};
    \node (Z-3) at (1,-0.5) {};
    \node (Z-5) at (1,0) {};
    \node (Z-1) at (1,0.5) {};
    \node[fill=fgraph] (Z-2) at (1,1) {};
\end{scope}
\node[left] at (X4-X3.west) {$X_3X_4$};
\node[left] at (X2-X3.west) {$X_2X_3$};
\node[left] at (X2-X5.west) {$X_2X_5$};
\node[left] at (X2-X1.west) {$X_1X_2$};
\node[right] at (Z-4.east) {$Z_4$};
\node[right] at (Z-3.east) {$Z_3$};
\node[right] at (Z-5.east) {$Z_5$};
\node[right] at (Z-1.east) {$Z_1$};
\node[right] at (Z-2.east) {$Z_2$};
\begin{scope}[every edge/.style={draw=black,very thick, color=fgraph}]
    \path [-] (X2-X1) edge (Z-2);
    \path [-] (X2-X1) edge (Z-1);
    \path [-] (X2-X5) edge (Z-2);
    \path [-] (X2-X5) edge (Z-5);
    \path [-] (X2-X3) edge (Z-2);
    \path [-] (X2-X3) edge (Z-3);
    \path [-] (X4-X3) edge (Z-3);
    \path [-] (X4-X3) edge (Z-4);
\end{scope}
\end{tikzpicture}
\to
\begin{tikzpicture}[baseline=4mm, scale=1.0] 
\begin{scope}[every node/.style={circle,thick,draw,scale=0.5, color=fgraph}]
    \node (X4-X3) at (0,-0.5) {};
    \node[fill=fgraph] (X2-X3) at (0,0) {};
    \node[fill=fgraph] (X2-X5) at (0,0.5) {};
    \node[fill=fgraph] (X2-X1) at (0,1) {};
    \node (Z-4) at (1,-1) {};
    \node (Z-3) at (1,-0.5) {};
    \node (Z-5) at (1,0) {};
    \node (Z-1) at (1,0.5) {};
    \node[fill=fgraph] (Z-2) at (1,1) {};
\end{scope}
\node[left] at (X4-X3.west) {$X_3X_4$};
\node[left] at (X2-X3.west) {$X_2X_3$};
\node[left] at (X2-X5.west) {$X_2X_5$};
\node[left] at (X2-X1.west) {$X_1X_2$};
\node[right] at (Z-4.east) {$Z_4$};
\node[right] at (Z-3.east) {$Z_3$};
\node[right] at (Z-5.east) {$Z_5$};
\node[right] at (Z-1.east) {$Z_1$};
\node[right] at (Z-2.east) {$Z_2$};
\begin{scope}[every edge/.style={draw=black,very thick, color=fgraph}]
    \path [-] (X2-X1) edge (Z-2);
    \path [-] (X2-X1) edge (Z-1);
    \path [-] (X2-X5) edge (Z-2);
    \path [-] (X2-X5) edge (Z-5);
    \path [-] (X2-X3) edge (Z-2);
    \path [-] (X2-X3) edge (Z-3);
    \path [-] (X4-X3) edge (Z-3);
    \path [-] (X4-X3) edge (Z-4);
\end{scope}
\end{tikzpicture}
\to
\begin{tikzpicture}[baseline=4mm, scale=1.0] 
\begin{scope}[every node/.style={circle,thick,draw,scale=0.5, color=fgraph}]
    \node (X4-X3) at (0,-0.5) {};
    \node[fill=fgraph] (X2-X3) at (0,0) {};
    \node[fill=fgraph] (X2-X5) at (0,0.5) {};
    \node[fill=fgraph] (X2-X1) at (0,1) {};
    \node (Z-4) at (1,-1) {};
    \node (Z-3) at (1,-0.5) {};
    \node (Z-5) at (1,0) {};
    \node (Z-1) at (1,0.5) {};
    \node (Z-2) at (1,1) {};
\end{scope}
\node[left] at (X4-X3.west) {$X_3X_4$};
\node[left] at (X2-X3.west) {$X_2X_3$};
\node[left] at (X2-X5.west) {$X_2X_5$};
\node[left] at (X2-X1.west) {$X_1X_2$};
\node[right] at (Z-4.east) {$Z_4$};
\node[right] at (Z-3.east) {$Z_3$};
\node[right] at (Z-5.east) {$Z_5$};
\node[right] at (Z-1.east) {$Z_1$};
\node[right] at (Z-2.east) {$Z_2$};
\begin{scope}[every edge/.style={draw=black,very thick, color=fgraph}]
    \path [-] (X2-X1) edge (Z-2);
    \path [-] (X2-X1) edge (Z-1);
    \path [-] (X2-X5) edge (Z-2);
    \path [-] (X2-X5) edge (Z-5);
    \path [-] (X2-X3) edge (Z-2);
    \path [-] (X2-X3) edge (Z-3);
    \path [-] (X4-X3) edge (Z-3);
    \path [-] (X4-X3) edge (Z-4);
\end{scope}
\end{tikzpicture}.
\end{align*}
In order to shorten the remaining part of this proof, observe that we can move the color from the vertex $X_1X_2$ to $Z_1$ by first adding $Z_1$ and then removing $X_1 X_2$:
\begin{align}\label{eq:move_vertex}
\begin{tikzpicture}[baseline=4mm, scale=1.0] 
\begin{scope}[every node/.style={circle,thick,draw,scale=0.5, color=fgraph}]
    \node (X4-X3) at (0,-0.5) {};
    \node[fill=fgraph] (X2-X3) at (0,0) {};
    \node[fill=fgraph] (X2-X5) at (0,0.5) {};
    \node[fill=fgraph] (X2-X1) at (0,1) {};
    \node (Z-4) at (1,-1) {};
    \node (Z-3) at (1,-0.5) {};
    \node (Z-5) at (1,0) {};
    \node (Z-1) at (1,0.5) {};
    \node (Z-2) at (1,1) {};
\end{scope}
\node[left] at (X4-X3.west) {$X_3X_4$};
\node[left] at (X2-X3.west) {$X_2X_3$};
\node[left] at (X2-X5.west) {$X_2X_5$};
\node[left] at (X2-X1.west) {$X_1X_2$};
\node[right] at (Z-4.east) {$Z_4$};
\node[right] at (Z-3.east) {$Z_3$};
\node[right] at (Z-5.east) {$Z_5$};
\node[right] at (Z-1.east) {$Z_1$};
\node[right] at (Z-2.east) {$Z_2$};
\begin{scope}[every edge/.style={draw=black,very thick, color=fgraph}]
    \path [-] (X2-X1) edge (Z-2);
    \path [-] (X2-X1) edge (Z-1);
    \path [-] (X2-X5) edge (Z-2);
    \path [-] (X2-X5) edge (Z-5);
    \path [-] (X2-X3) edge (Z-2);
    \path [-] (X2-X3) edge (Z-3);
    \path [-] (X4-X3) edge (Z-3);
    \path [-] (X4-X3) edge (Z-4);
\end{scope}
\end{tikzpicture}
\to
\begin{tikzpicture}[baseline=4mm, scale=1.0] 
\begin{scope}[every node/.style={circle,thick,draw,scale=0.5, color=fgraph}]
    \node (X4-X3) at (0,-0.5) {};
    \node[fill=fgraph] (X2-X3) at (0,0) {};
    \node[fill=fgraph] (X2-X5) at (0,0.5) {};
    \node[fill=fgraph] (X2-X1) at (0,1) {};
    \node (Z-4) at (1,-1) {};
    \node (Z-3) at (1,-0.5) {};
    \node (Z-5) at (1,0) {};
    \node[fill=fgraph] (Z-1) at (1,0.5) {};
    \node (Z-2) at (1,1) {};
\end{scope}
\node[left] at (X4-X3.west) {$X_3X_4$};
\node[left] at (X2-X3.west) {$X_2X_3$};
\node[left] at (X2-X5.west) {$X_2X_5$};
\node[left] at (X2-X1.west) {$X_1X_2$};
\node[right] at (Z-4.east) {$Z_4$};
\node[right] at (Z-3.east) {$Z_3$};
\node[right] at (Z-5.east) {$Z_5$};
\node[right] at (Z-1.east) {$Z_1$};
\node[right] at (Z-2.east) {$Z_2$};
\begin{scope}[every edge/.style={draw=black,very thick, color=fgraph}]
    \path [-] (X2-X1) edge (Z-2);
    \path [-] (X2-X1) edge (Z-1);
    \path [-] (X2-X5) edge (Z-2);
    \path [-] (X2-X5) edge (Z-5);
    \path [-] (X2-X3) edge (Z-2);
    \path [-] (X2-X3) edge (Z-3);
    \path [-] (X4-X3) edge (Z-3);
    \path [-] (X4-X3) edge (Z-4);
\end{scope}
\end{tikzpicture}
\to
\begin{tikzpicture}[baseline=4mm, scale=1.0] 
\begin{scope}[every node/.style={circle,thick,draw,scale=0.5, color=fgraph}]
    \node (X4-X3) at (0,-0.5) {};
    \node[fill=fgraph] (X2-X3) at (0,0) {};
    \node[fill=fgraph] (X2-X5) at (0,0.5) {};
    \node (X2-X1) at (0,1) {};
    \node (Z-4) at (1,-1) {};
    \node (Z-3) at (1,-0.5) {};
    \node (Z-5) at (1,0) {};
    \node[fill=fgraph] (Z-1) at (1,0.5) {};
    \node (Z-2) at (1,1) {};
\end{scope}
\node[left] at (X4-X3.west) {$X_3X_4$};
\node[left] at (X2-X3.west) {$X_2X_3$};
\node[left] at (X2-X5.west) {$X_2X_5$};
\node[left] at (X2-X1.west) {$X_1X_2$};
\node[right] at (Z-4.east) {$Z_4$};
\node[right] at (Z-3.east) {$Z_3$};
\node[right] at (Z-5.east) {$Z_5$};
\node[right] at (Z-1.east) {$Z_1$};
\node[right] at (Z-2.east) {$Z_2$};
\begin{scope}[every edge/.style={draw=black,very thick, color=fgraph}]
    \path [-] (X2-X1) edge (Z-2);
    \path [-] (X2-X1) edge (Z-1);
    \path [-] (X2-X5) edge (Z-2);
    \path [-] (X2-X5) edge (Z-5);
    \path [-] (X2-X3) edge (Z-2);
    \path [-] (X2-X3) edge (Z-3);
    \path [-] (X4-X3) edge (Z-3);
    \path [-] (X4-X3) edge (Z-4);
\end{scope}
\end{tikzpicture}.
\end{align}
%
Using this shortcut move,  we move the colored vertices until we
obtain the term $X_1 X_4$ as follows: 
\begin{align*}
\begin{tikzpicture}[baseline=4mm, scale=1.0] 
\begin{scope}[every node/.style={circle,thick,draw,scale=0.5, color=fgraph}]
    \node (X4-X3) at (0,-0.5) {};
    \node[fill=fgraph] (X2-X3) at (0,0) {};
    \node[fill=fgraph] (X2-X5) at (0,0.5) {};
    \node (X2-X1) at (0,1) {};
    \node (Z-4) at (1,-1) {};
    \node (Z-3) at (1,-0.5) {};
    \node (Z-5) at (1,0) {};
    \node[fill=fgraph] (Z-1) at (1,0.5) {};
    \node (Z-2) at (1,1) {};
\end{scope}
\node[left] at (X4-X3.west) {$X_3X_4$};
\node[left] at (X2-X3.west) {$X_2X_3$};
\node[left] at (X2-X5.west) {$X_2X_5$};
\node[left] at (X2-X1.west) {$X_1X_2$};
\node[right] at (Z-4.east) {$Z_4$};
\node[right] at (Z-3.east) {$Z_3$};
\node[right] at (Z-5.east) {$Z_5$};
\node[right] at (Z-1.east) {$Z_1$};
\node[right] at (Z-2.east) {$Z_2$};
\begin{scope}[every edge/.style={draw=black,very thick, color=fgraph}]
    \path [-] (X2-X1) edge (Z-2);
    \path [-] (X2-X1) edge (Z-1);
    \path [-] (X2-X5) edge (Z-2);
    \path [-] (X2-X5) edge (Z-5);
    \path [-] (X2-X3) edge (Z-2);
    \path [-] (X2-X3) edge (Z-3);
    \path [-] (X4-X3) edge (Z-3);
    \path [-] (X4-X3) edge (Z-4);
\end{scope}
\end{tikzpicture}
\to
\begin{tikzpicture}[baseline=4mm, scale=1.0] 
\begin{scope}[every node/.style={circle,thick,draw,scale=0.5, color=fgraph}]
    \node (X4-X3) at (0,-0.5) {};
    \node (X2-X3) at (0,0) {};
    \node[fill=fgraph] (X2-X5) at (0,0.5) {};
    \node (X2-X1) at (0,1) {};
    \node (Z-4) at (1,-1) {};
    \node[fill=fgraph] (Z-3) at (1,-0.5) {};
    \node (Z-5) at (1,0) {};
    \node[fill=fgraph] (Z-1) at (1,0.5) {};
    \node (Z-2) at (1,1) {};
\end{scope}
\node[left] at (X4-X3.west) {$X_3X_4$};
\node[left] at (X2-X3.west) {$X_2X_3$};
\node[left] at (X2-X5.west) {$X_2X_5$};
\node[left] at (X2-X1.west) {$X_1X_2$};
\node[right] at (Z-4.east) {$Z_4$};
\node[right] at (Z-3.east) {$Z_3$};
\node[right] at (Z-5.east) {$Z_5$};
\node[right] at (Z-1.east) {$Z_1$};
\node[right] at (Z-2.east) {$Z_2$};
\begin{scope}[every edge/.style={draw=black,very thick, color=fgraph}]
    \path [-] (X2-X1) edge (Z-2);
    \path [-] (X2-X1) edge (Z-1);
    \path [-] (X2-X5) edge (Z-2);
    \path [-] (X2-X5) edge (Z-5);
    \path [-] (X2-X3) edge (Z-2);
    \path [-] (X2-X3) edge (Z-3);
    \path [-] (X4-X3) edge (Z-3);
    \path [-] (X4-X3) edge (Z-4);
\end{scope}
\end{tikzpicture}
\to
\begin{tikzpicture}[baseline=4mm, scale=1.0] 
\begin{scope}[every node/.style={circle,thick,draw,scale=0.5, color=fgraph}]
    \node[fill=fgraph] (X4-X3) at (0,-0.5) {};
    \node (X2-X3) at (0,0) {};
    \node (X2-X5) at (0,0.5) {};
    \node (X2-X1) at (0,1) {};
    \node (Z-4) at (1,-1) {};
    \node (Z-3) at (1,-0.5) {};
    \node (Z-5) at (1,0) {};
    \node[fill=fgraph] (Z-1) at (1,0.5) {};
    \node[fill=fgraph] (Z-2) at (1,1) {};
\end{scope}
\node[left] at (X4-X3.west) {$X_3X_4$};
\node[left] at (X2-X3.west) {$X_2X_3$};
\node[left] at (X2-X5.west) {$X_2X_5$};
\node[left] at (X2-X1.west) {$X_1X_2$};
\node[right] at (Z-4.east) {$Z_4$};
\node[right] at (Z-3.east) {$Z_3$};
\node[right] at (Z-5.east) {$Z_5$};
\node[right] at (Z-1.east) {$Z_1$};
\node[right] at (Z-2.east) {$Z_2$};
\begin{scope}[every edge/.style={draw=black,very thick, color=fgraph}]
    \path [-] (X2-X1) edge (Z-2);
    \path [-] (X2-X1) edge (Z-1);
    \path [-] (X2-X5) edge (Z-2);
    \path [-] (X2-X5) edge (Z-5);
    \path [-] (X2-X3) edge (Z-2);
    \path [-] (X2-X3) edge (Z-3);
    \path [-] (X4-X3) edge (Z-3);
    \path [-] (X4-X3) edge (Z-4);
\end{scope}
\end{tikzpicture}
\to
\begin{tikzpicture}[baseline=4mm, scale=1.0] 
\begin{scope}[every node/.style={circle,thick,draw,scale=0.5, color=fgraph}]
    \node[fill=fgraph] (X4-X3) at (0,-0.5) {};
    \node[fill=fgraph] (X2-X3) at (0,0) {};
    \node (X2-X5) at (0,0.5) {};
    \node (X2-X1) at (0,1) {};
    \node (Z-4) at (1,-1) {};
    \node (Z-3) at (1,-0.5) {};
    \node (Z-5) at (1,0) {};
    \node[fill=fgraph] (Z-1) at (1,0.5) {};
    \node (Z-2) at (1,1) {};
\end{scope}
\node[left] at (X4-X3.west) {$X_3X_4$};
\node[left] at (X2-X3.west) {$X_2X_3$};
\node[left] at (X2-X5.west) {$X_2X_5$};
\node[left] at (X2-X1.west) {$X_1X_2$};
\node[right] at (Z-4.east) {$Z_4$};
\node[right] at (Z-3.east) {$Z_3$};
\node[right] at (Z-5.east) {$Z_5$};
\node[right] at (Z-1.east) {$Z_1$};
\node[right] at (Z-2.east) {$Z_2$};
\begin{scope}[every edge/.style={draw=black,very thick, color=fgraph}]
    \path [-] (X2-X1) edge (Z-2);
    \path [-] (X2-X1) edge (Z-1);
    \path [-] (X2-X5) edge (Z-2);
    \path [-] (X2-X5) edge (Z-5);
    \path [-] (X2-X3) edge (Z-2);
    \path [-] (X2-X3) edge (Z-3);
    \path [-] (X4-X3) edge (Z-3);
    \path [-] (X4-X3) edge (Z-4);
\end{scope}
\end{tikzpicture}
\to
\begin{tikzpicture}[baseline=4mm, scale=1.0] 
\begin{scope}[every node/.style={circle,thick,draw,scale=0.5, color=fgraph}]
    \node[fill=fgraph] (X4-X3) at (0,-0.5) {};
    \node[fill=fgraph] (X2-X3) at (0,0) {};
    \node (X2-X5) at (0,0.5) {};
    \node[fill=fgraph] (X2-X1) at (0,1) {};
    \node (Z-4) at (1,-1) {};
    \node (Z-3) at (1,-0.5) {};
    \node (Z-5) at (1,0) {};
    \node (Z-1) at (1,0.5) {};
    \node (Z-2) at (1,1) {};
\end{scope}
\node[left] at (X4-X3.west) {$X_3X_4$};
\node[left] at (X2-X3.west) {$X_2X_3$};
\node[left] at (X2-X5.west) {$X_2X_5$};
\node[left] at (X2-X1.west) {$X_1X_2$};
\node[right] at (Z-4.east) {$Z_4$};
\node[right] at (Z-3.east) {$Z_3$};
\node[right] at (Z-5.east) {$Z_5$};
\node[right] at (Z-1.east) {$Z_1$};
\node[right] at (Z-2.east) {$Z_2$};
\begin{scope}[every edge/.style={draw=black,very thick, color=fgraph}]
    \path [-] (X2-X1) edge (Z-2);
    \path [-] (X2-X1) edge (Z-1);
    \path [-] (X2-X5) edge (Z-2);
    \path [-] (X2-X5) edge (Z-5);
    \path [-] (X2-X3) edge (Z-2);
    \path [-] (X2-X3) edge (Z-3);
    \path [-] (X4-X3) edge (Z-3);
    \path [-] (X4-X3) edge (Z-4);
\end{scope}
\end{tikzpicture}.
\end{align*}
We conclude that $X_1 X_4 \in \aa^\Sigma_{14}$. 
\end{proof}

\ifTikz

\subsection{Alternative proofs of Eqs. \eqref{eq:Y1X3}, \eqref{eq:X1X3} and \eqref{eq:Z1Z3}}\label{frust-eq567}

In this subsection, using the colored frustration graph operations of adding and removing a vertex given in  Definition~\ref{def:frust_ops},
we will derive the claims from the proof of Lemma \ref{lem:k3plusk4}: 
$X_3 Y_1 \in \aa^\Omega_2$, $X_1X_3 \in \aa^\Omega_4$, $Z_1Z_3 \in \aa^\Omega_6$, and $X_1X_3 \in \aa^\Omega_{14}$.

\begin{lemma}\label{lemmacyc}
We have $X_3 Y_1 \in \aa^\Omega_2$.
\end{lemma}
\begin{proof}
Recall that the Lie algebra $\aa^\Omega_2$ is generated by all $X_iY_j$, where $(i,j)$ are the edges of the interaction graph $\Omega$ from Lemma \ref{lem:k3plusk4}. These generators give the frustration graph
\begin{align*}
\Gamma(\aa^\Omega_2) =\:
\begin{tikzpicture}[baseline=4mm, scale=2.0] 
\begin{scope}[every node/.style={circle,thick,draw,scale=0.5, color=fgraph}]
    \node (X1Y4) at (0,-0.5) {};
    \node (X3Y4) at (0,0) {};
    \node (X1Y3) at (0,0.5) {};
    \node (X1Y2) at (0,1) {};
    \node (Y1X4) at (1,-0.5) {};
    \node (Y3X4) at (1,0.0) {};
    \node (Y1X3) at (1,0.5) {};
    \node (Y1X2) at (1,1) {};
\end{scope}
\node[left] at (X1Y4.west) {$X_2Y_4$};
\node[left] at (X3Y4.west) {$X_1Y_4$};
\node[left] at (X1Y3.west) {$X_2Y_1$};
\node[left] at (X1Y2.west) {$X_2Y_3$};
\node[right] at (Y1X4.east) {$X_4Y_2$};
\node[right] at (Y3X4.east) {$X_4Y_1$};
\node[right] at (Y1X3.east) {$X_1Y_2$};
\node[right] at (Y1X2.east) {$X_3Y_2$};
\begin{scope}[every edge/.style={draw=black,very thick, color=fgraph}]
    \path [-] (X1Y2) edge (Y1X3);
    \path [-] (X1Y2) edge (Y1X4);
    \path [-] (X1Y3) edge (Y1X2);
    \path [-] (X1Y3) edge (Y1X4);
    \path [-] (X1Y4) edge (Y1X2);
    \path [-] (X1Y4) edge (Y1X3);
    \path [-] (X1Y4) edge (Y3X4);
    \path [-] (Y1X4) edge (X3Y4);
    \path [-] (X1Y3) edge(X3Y4);
    \path [-] (Y1X3) edge (Y3X4);
\end{scope}
\end{tikzpicture}\;.
\end{align*}
We note that
$X_3 Y_1 = X_1Y_2 \cdot X_2Y_4 \cdot X_1Y_4 \cdot X_3Y_2 \cdot X_2Y_1$. 
We can obtain the corresponding colored frustration graph from a graph with a single colored vertex by adding vertices as follows:
\begin{align*}
\begin{tikzpicture}[baseline=4mm, scale=1.0] 
\begin{scope}[every node/.style={circle,thick,draw,scale=0.5, color=fgraph}]
    \node (X1Y4) at (0,-0.5) {};
    \node (X3Y4) at (0,0) {};
    \node[fill=fgraph] (X1Y3) at (0,0.5) {};
    \node (X1Y2) at (0,1) {};
    \node (Y1X4) at (1,-0.5) {};
    \node (Y3X4) at (1,0.0) {};
    \node (Y1X3) at (1,0.5) {};
    \node (Y1X2) at (1,1) {};
\end{scope}
\node[left] at (X1Y4.west) {$X_2Y_4$};
\node[left] at (X3Y4.west) {$X_1Y_4$};
\node[left] at (X1Y3.west) {$X_2Y_1$};
\node[left] at (X1Y2.west) {$X_2Y_3$};
\node[right] at (Y1X4.east) {$X_4Y_2$};
\node[right] at (Y3X4.east) {$X_4Y_1$};
\node[right] at (Y1X3.east) {$X_1Y_2$};
\node[right] at (Y1X2.east) {$X_3Y_2$};
\begin{scope}[every edge/.style={draw=black,very thick, color=fgraph}]
    \path [-] (X1Y2) edge (Y1X3);
    \path [-] (X1Y2) edge (Y1X4);
    \path [-] (X1Y3) edge (Y1X2);
    \path [-] (X1Y3) edge (Y1X4);
    \path [-] (X1Y4) edge (Y1X2);
    \path [-] (X1Y4) edge (Y1X3);
    \path [-] (X1Y4) edge (Y3X4);
    \path [-] (Y1X4) edge (X3Y4);
    \path [-] (X1Y3) edge(X3Y4);
    \path [-] (Y1X3) edge (Y3X4);
\end{scope}
\end{tikzpicture}
\to
\begin{tikzpicture}[baseline=4mm, scale=1.0] 
\begin{scope}[every node/.style={circle,thick,draw,scale=0.5, color=fgraph}]
    \node (X1Y4) at (0,-0.5) {};
    \node (X3Y4) at (0,0) {};
    \node[fill=fgraph] (X1Y3) at (0,0.5) {};
    \node (X1Y2) at (0,1) {};
    \node (Y1X4) at (1,-0.5) {};
    \node (Y3X4) at (1,0.0) {};
    \node (Y1X3) at (1,0.5) {};
    \node[fill=fgraph] (Y1X2) at (1,1) {};
\end{scope}
\node[left] at (X1Y4.west) {$X_2Y_4$};
\node[left] at (X3Y4.west) {$X_1Y_4$};
\node[left] at (X1Y3.west) {$X_2Y_1$};
\node[left] at (X1Y2.west) {$X_2Y_3$};
\node[right] at (Y1X4.east) {$X_4Y_2$};
\node[right] at (Y3X4.east) {$X_4Y_1$};
\node[right] at (Y1X3.east) {$X_1Y_2$};
\node[right] at (Y1X2.east) {$X_3Y_2$};
\begin{scope}[every edge/.style={draw=black,very thick, color=fgraph}]
    \path [-] (X1Y2) edge (Y1X3);
    \path [-] (X1Y2) edge (Y1X4);
    \path [-] (X1Y3) edge (Y1X2);
    \path [-] (X1Y3) edge (Y1X4);
    \path [-] (X1Y4) edge (Y1X2);
    \path [-] (X1Y4) edge (Y1X3);
    \path [-] (X1Y4) edge (Y3X4);
    \path [-] (Y1X4) edge (X3Y4);
    \path [-] (X1Y3) edge(X3Y4);
    \path [-] (Y1X3) edge (Y3X4);
\end{scope}
\end{tikzpicture}
\to
\begin{tikzpicture}[baseline=4mm, scale=1.0] 
\begin{scope}[every node/.style={circle,thick,draw,scale=0.5, color=fgraph}]
    \node (X1Y4) at (0,-0.5) {};
    \node[fill=fgraph] (X3Y4) at (0,0) {};
    \node[fill=fgraph] (X1Y3) at (0,0.5) {};
    \node (X1Y2) at (0,1) {};
    \node (Y1X4) at (1,-0.5) {};
    \node (Y3X4) at (1,0.0) {};
    \node (Y1X3) at (1,0.5) {};
    \node[fill=fgraph] (Y1X2) at (1,1) {};
\end{scope}
\node[left] at (X1Y4.west) {$X_2Y_4$};
\node[left] at (X3Y4.west) {$X_1Y_4$};
\node[left] at (X1Y3.west) {$X_2Y_1$};
\node[left] at (X1Y2.west) {$X_2Y_3$};
\node[right] at (Y1X4.east) {$X_4Y_2$};
\node[right] at (Y3X4.east) {$X_4Y_1$};
\node[right] at (Y1X3.east) {$X_1Y_2$};
\node[right] at (Y1X2.east) {$X_3Y_2$};
\begin{scope}[every edge/.style={draw=black,very thick, color=fgraph}]
    \path [-] (X1Y2) edge (Y1X3);
    \path [-] (X1Y2) edge (Y1X4);
    \path [-] (X1Y3) edge (Y1X2);
    \path [-] (X1Y3) edge (Y1X4);
    \path [-] (X1Y4) edge (Y1X2);
    \path [-] (X1Y4) edge (Y1X3);
    \path [-] (X1Y4) edge (Y3X4);
    \path [-] (Y1X4) edge (X3Y4);
    \path [-] (X1Y3) edge(X3Y4);
    \path [-] (Y1X3) edge (Y3X4);
\end{scope}
\end{tikzpicture}
\to
\begin{tikzpicture}[baseline=4mm, scale=1.0] 
\begin{scope}[every node/.style={circle,thick,draw,scale=0.5, color=fgraph}]
    \node[fill=fgraph] (X1Y4) at (0,-0.5) {};
    \node[fill=fgraph] (X3Y4) at (0,0) {};
    \node[fill=fgraph] (X1Y3) at (0,0.5) {};
    \node (X1Y2) at (0,1) {};
    \node (Y1X4) at (1,-0.5) {};
    \node (Y3X4) at (1,0.0) {};
    \node (Y1X3) at (1,0.5) {};
    \node[fill=fgraph] (Y1X2) at (1,1) {};
\end{scope}
\node[left] at (X1Y4.west) {$X_2Y_4$};
\node[left] at (X3Y4.west) {$X_1Y_4$};
\node[left] at (X1Y3.west) {$X_2Y_1$};
\node[left] at (X1Y2.west) {$X_2Y_3$};
\node[right] at (Y1X4.east) {$X_4Y_2$};
\node[right] at (Y3X4.east) {$X_4Y_1$};
\node[right] at (Y1X3.east) {$X_1Y_2$};
\node[right] at (Y1X2.east) {$X_3Y_2$};
\begin{scope}[every edge/.style={draw=black,very thick, color=fgraph}]
    \path [-] (X1Y2) edge (Y1X3);
    \path [-] (X1Y2) edge (Y1X4);
    \path [-] (X1Y3) edge (Y1X2);
    \path [-] (X1Y3) edge (Y1X4);
    \path [-] (X1Y4) edge (Y1X2);
    \path [-] (X1Y4) edge (Y1X3);
    \path [-] (X1Y4) edge (Y3X4);
    \path [-] (Y1X4) edge (X3Y4);
    \path [-] (X1Y3) edge(X3Y4);
    \path [-] (Y1X3) edge (Y3X4);
\end{scope}
\end{tikzpicture}
\to
\begin{tikzpicture}[baseline=4mm, scale=1.0] 
\begin{scope}[every node/.style={circle,thick,draw,scale=0.5, color=fgraph}]
    \node[fill=fgraph] (X1Y4) at (0,-0.5) {};
    \node[fill=fgraph] (X3Y4) at (0,0) {};
    \node[fill=fgraph] (X1Y3) at (0,0.5) {};
    \node (X1Y2) at (0,1) {};
    \node (Y1X4) at (1,-0.5) {};
    \node (Y3X4) at (1,0.0) {};
    \node[fill=fgraph] (Y1X3) at (1,0.5) {};
    \node[fill=fgraph] (Y1X2) at (1,1) {};
\end{scope}
\node[left] at (X1Y4.west) {$X_2Y_4$};
\node[left] at (X3Y4.west) {$X_1Y_4$};
\node[left] at (X1Y3.west) {$X_2Y_1$};
\node[left] at (X1Y2.west) {$X_2Y_3$};
\node[right] at (Y1X4.east) {$X_4Y_2$};
\node[right] at (Y3X4.east) {$X_4Y_1$};
\node[right] at (Y1X3.east) {$X_1Y_2$};
\node[right] at (Y1X2.east) {$X_3Y_2$};
\begin{scope}[every edge/.style={draw=black,very thick, color=fgraph}]
    \path [-] (X1Y2) edge (Y1X3);
    \path [-] (X1Y2) edge (Y1X4);
    \path [-] (X1Y3) edge (Y1X2);
    \path [-] (X1Y3) edge (Y1X4);
    \path [-] (X1Y4) edge (Y1X2);
    \path [-] (X1Y4) edge (Y1X3);
    \path [-] (X1Y4) edge (Y3X4);
    \path [-] (Y1X4) edge (X3Y4);
    \path [-] (X1Y3) edge(X3Y4);
    \path [-] (Y1X3) edge (Y3X4);
\end{scope}
\end{tikzpicture}
\end{align*}
Therefore, $X_3 Y_1 \in \aa^\Omega_2$. 
\end{proof}

\begin{lemma}
We have $X_1X_3 \in \aa^\Omega_4$.
\end{lemma}

\begin{proof}
The Lie algebra $\aa^\Omega_4$ is generated by placing $XX$ and $YY$ on the edges of the interaction graph $\Omega$, giving rise to
the frustration graph
\begin{align*}
\Gamma(\aa_4^\Omega) =\:
\begin{tikzpicture}[baseline=4mm, scale=2.0] 
\begin{scope}[every node/.style={circle,thick,draw,scale=0.5, color=fgraph}]
    \node (X1X4) at (0,-0.5) {};
    \node (X3X4) at (0,0) {};
    \node (X1X3) at (0,0.5) {};
    \node (X1X2) at (0,1) {};
    \node (Y1Y4) at (1,-0.5) {};
    \node (Y3Y4) at (1,0) {};
    \node (Y1Y3) at (1,0.5) {};
    \node (Y1Y2) at (1,1) {};
\end{scope}
\node[left] at (X1X4.west) {$X_2X_4$};
\node[left] at (X3X4.west) {$X_1X_4$};
\node[left] at (X1X3.west) {$X_1X_2$};
\node[left] at (X1X2.west) {$X_2X_3$};
\node[right] at (Y1Y4.east) {$Y_2Y_4$};
\node[right] at (Y3Y4.east) {$Y_1Y_4$};
\node[right] at (Y1Y3.east) {$Y_1Y_2$};
\node[right] at (Y1Y2.east) {$Y_2Y_3$};
\begin{scope}[every edge/.style={draw=black,very thick, color=fgraph}]
    \path [-] (X1X2) edge (Y1Y3);
    \path [-] (X1X2) edge (Y1Y4);
    \path [-] (X1X3) edge (Y1Y2);
    \path [-] (X1X3) edge (Y3Y4);
    \path [-] (X1X3) edge (Y1Y4);
    \path [-] (X3X4) edge (Y1Y3);
    \path [-] (X1X4) edge (Y1Y2);
    \path [-] (X1X4) edge (Y1Y3);
    \path [-] (X3X4) edge (Y1Y4);
    \path [-] (X1X4) edge (Y3Y4);
\end{scope}
\end{tikzpicture}\;.
\end{align*}
We want to obtain the colored frustration graph of
$X_1X_3 = X_2 X_3 \cdot X_1 X_4 \cdot X_2 X_4$ by add/remove vertex operations. Starting with the generator $X_2 X_3$, we perform the following sequence:
\begin{align*}
\begin{tikzpicture}[baseline=4mm, scale=1.0] 
\begin{scope}[every node/.style={circle,thick,draw,scale=0.5, color=fgraph}]
    \node (X1X4) at (0,-0.5) {};
    \node (X3X4) at (0,0) {};
    \node (X1X3) at (0,0.5) {};
    \node[fill=fgraph] (X1X2) at (0,1) {};
    \node (Y1Y4) at (1,-0.5) {};
    \node (Y3Y4) at (1,0) {};
    \node (Y1Y3) at (1,0.5) {};
    \node (Y1Y2) at (1,1) {};
\end{scope}
\node[left] at (X1X4.west) {$X_2X_4$};
\node[left] at (X3X4.west) {$X_1X_4$};
\node[left] at (X1X3.west) {$X_1X_2$};
\node[left] at (X1X2.west) {$X_2X_3$};
\node[right] at (Y1Y4.east) {$Y_2Y_4$};
\node[right] at (Y3Y4.east) {$Y_1Y_4$};
\node[right] at (Y1Y3.east) {$Y_1Y_2$};
\node[right] at (Y1Y2.east) {$Y_2Y_3$};
\begin{scope}[every edge/.style={draw=black,very thick, color=fgraph}]
    \path [-] (X1X2) edge (Y1Y3);
    \path [-] (X1X2) edge (Y1Y4);
    \path [-] (X1X3) edge (Y1Y2);
    \path [-] (X1X3) edge (Y3Y4);
    \path [-] (X1X3) edge (Y1Y4);
    \path [-] (X3X4) edge (Y1Y3);
    \path [-] (X1X4) edge (Y1Y2);
    \path [-] (X1X4) edge (Y1Y3);
    \path [-] (X3X4) edge (Y1Y4);
    \path [-] (X1X4) edge (Y3Y4);
\end{scope}
\end{tikzpicture}
\to
\begin{tikzpicture}[baseline=4mm, scale=1.0] 
\begin{scope}[every node/.style={circle,thick,draw,scale=0.5, color=fgraph}]
    \node (X1X4) at (0,-0.5) {};
    \node (X3X4) at (0,0) {};
    \node (X1X3) at (0,0.5) {};
    \node[fill=fgraph] (X1X2) at (0,1) {};
    \node (Y1Y4) at (1,-0.5) {};
    \node (Y3Y4) at (1,0) {};
    \node[fill=fgraph] (Y1Y3) at (1,0.5) {};
    \node (Y1Y2) at (1,1) {};
\end{scope}
\node[left] at (X1X4.west) {$X_2X_4$};
\node[left] at (X3X4.west) {$X_1X_4$};
\node[left] at (X1X3.west) {$X_1X_2$};
\node[left] at (X1X2.west) {$X_2X_3$};
\node[right] at (Y1Y4.east) {$Y_2Y_4$};
\node[right] at (Y3Y4.east) {$Y_1Y_4$};
\node[right] at (Y1Y3.east) {$Y_1Y_2$};
\node[right] at (Y1Y2.east) {$Y_2Y_3$};
\begin{scope}[every edge/.style={draw=black,very thick, color=fgraph}]
    \path [-] (X1X2) edge (Y1Y3);
    \path [-] (X1X2) edge (Y1Y4);
    \path [-] (X1X3) edge (Y1Y2);
    \path [-] (X1X3) edge (Y3Y4);
    \path [-] (X1X3) edge (Y1Y4);
    \path [-] (X3X4) edge (Y1Y3);
    \path [-] (X1X4) edge (Y1Y2);
    \path [-] (X1X4) edge (Y1Y3);
    \path [-] (X3X4) edge (Y1Y4);
    \path [-] (X1X4) edge (Y3Y4);
\end{scope}
\end{tikzpicture}
\to
\begin{tikzpicture}[baseline=4mm, scale=1.0] 
\begin{scope}[every node/.style={circle,thick,draw,scale=0.5, color=fgraph}]
    \node[fill=fgraph] (X1X4) at (0,-0.5) {};
    \node (X3X4) at (0,0) {};
    \node (X1X3) at (0,0.5) {};
    \node[fill=fgraph] (X1X2) at (0,1) {};
    \node (Y1Y4) at (1,-0.5) {};
    \node (Y3Y4) at (1,0) {};
    \node[fill=fgraph] (Y1Y3) at (1,0.5) {};
    \node (Y1Y2) at (1,1) {};
\end{scope}
\node[left] at (X1X4.west) {$X_2X_4$};
\node[left] at (X3X4.west) {$X_1X_4$};
\node[left] at (X1X3.west) {$X_1X_2$};
\node[left] at (X1X2.west) {$X_2X_3$};
\node[right] at (Y1Y4.east) {$Y_2Y_4$};
\node[right] at (Y3Y4.east) {$Y_1Y_4$};
\node[right] at (Y1Y3.east) {$Y_1Y_2$};
\node[right] at (Y1Y2.east) {$Y_2Y_3$};
\begin{scope}[every edge/.style={draw=black,very thick, color=fgraph}]
    \path [-] (X1X2) edge (Y1Y3);
    \path [-] (X1X2) edge (Y1Y4);
    \path [-] (X1X3) edge (Y1Y2);
    \path [-] (X1X3) edge (Y3Y4);
    \path [-] (X1X3) edge (Y1Y4);
    \path [-] (X3X4) edge (Y1Y3);
    \path [-] (X1X4) edge (Y1Y2);
    \path [-] (X1X4) edge (Y1Y3);
    \path [-] (X3X4) edge (Y1Y4);
    \path [-] (X1X4) edge (Y3Y4);
\end{scope}
\end{tikzpicture}
\to
\begin{tikzpicture}[baseline=4mm, scale=1.0] 
\begin{scope}[every node/.style={circle,thick,draw,scale=0.5, color=fgraph}]
    \node[fill=fgraph] (X1X4) at (0,-0.5) {};
    \node[fill=fgraph] (X3X4) at (0,0) {};
    \node (X1X3) at (0,0.5) {};
    \node[fill=fgraph] (X1X2) at (0,1) {};
    \node (Y1Y4) at (1,-0.5) {};
    \node (Y3Y4) at (1,0) {};
    \node[fill=fgraph] (Y1Y3) at (1,0.5) {};
    \node (Y1Y2) at (1,1) {};
\end{scope}
\node[left] at (X1X4.west) {$X_2X_4$};
\node[left] at (X3X4.west) {$X_1X_4$};
\node[left] at (X1X3.west) {$X_1X_2$};
\node[left] at (X1X2.west) {$X_2X_3$};
\node[right] at (Y1Y4.east) {$Y_2Y_4$};
\node[right] at (Y3Y4.east) {$Y_1Y_4$};
\node[right] at (Y1Y3.east) {$Y_1Y_2$};
\node[right] at (Y1Y2.east) {$Y_2Y_3$};
\begin{scope}[every edge/.style={draw=black,very thick, color=fgraph}]
    \path [-] (X1X2) edge (Y1Y3);
    \path [-] (X1X2) edge (Y1Y4);
    \path [-] (X1X3) edge (Y1Y2);
    \path [-] (X1X3) edge (Y3Y4);
    \path [-] (X1X3) edge (Y1Y4);
    \path [-] (X3X4) edge (Y1Y3);
    \path [-] (X1X4) edge (Y1Y2);
    \path [-] (X1X4) edge (Y1Y3);
    \path [-] (X3X4) edge (Y1Y4);
    \path [-] (X1X4) edge (Y3Y4);
\end{scope}
\end{tikzpicture}
\to
\begin{tikzpicture}[baseline=4mm, scale=1.0] 
\begin{scope}[every node/.style={circle,thick,draw,scale=0.5, color=fgraph}]
    \node[fill=fgraph] (X1X4) at (0,-0.5) {};
    \node[fill=fgraph] (X3X4) at (0,0) {};
    \node (X1X3) at (0,0.5) {};
    \node[fill=fgraph] (X1X2) at (0,1) {};
    \node (Y1Y4) at (1,-0.5) {};
    \node (Y3Y4) at (1,0) {};
    \node (Y1Y3) at (1,0.5) {};
    \node (Y1Y2) at (1,1) {};
\end{scope}
\node[left] at (X1X4.west) {$X_2X_4$};
\node[left] at (X3X4.west) {$X_1X_4$};
\node[left] at (X1X3.west) {$X_1X_2$};
\node[left] at (X1X2.west) {$X_2X_3$};
\node[right] at (Y1Y4.east) {$Y_2Y_4$};
\node[right] at (Y3Y4.east) {$Y_1Y_4$};
\node[right] at (Y1Y3.east) {$Y_1Y_2$};
\node[right] at (Y1Y2.east) {$Y_2Y_3$};
\begin{scope}[every edge/.style={draw=black,very thick, color=fgraph}]
    \path [-] (X1X2) edge (Y1Y3);
    \path [-] (X1X2) edge (Y1Y4);
    \path [-] (X1X3) edge (Y1Y2);
    \path [-] (X1X3) edge (Y3Y4);
    \path [-] (X1X3) edge (Y1Y4);
    \path [-] (X3X4) edge (Y1Y3);
    \path [-] (X1X4) edge (Y1Y2);
    \path [-] (X1X4) edge (Y1Y3);
    \path [-] (X3X4) edge (Y1Y4);
    \path [-] (X1X4) edge (Y3Y4);
\end{scope}
\end{tikzpicture}
\end{align*}
which yields that $X_1X_3 \in \aa^\Omega_4$.
%
\end{proof}

\begin{lemma}
We have $Z_1Z_3 \in \aa^\Omega_6$, if we define $\aa_6$ as $\Lie{XY,YX,ZZ}$.
\end{lemma}

\begin{proof}
The Lie algebra $\aa^\Omega_6$ is generated by placing
$XY,ZZ$ on all edges of the interaction graph $\Omega$.
We will consider the subset of generators consisting of $Z_1 Z_4$ and all $XY$-edges:
\begin{align*}
    \mcA = \{ X_2Y_3, X_3Y_2, X_2Y_1, X_1Y_2, X_2Y_4, X_4Y_2, X_1Y_4, X_4Y_1, Z_1 Z_4 \}. 
\end{align*}
The frustration graph of $\mcA$ is
\begin{align*} \Gamma(\mcA) = 
\begin{tikzpicture}[baseline=4mm, scale=2.0] 
\begin{scope}[every node/.style={circle,thick,draw,scale=0.5, color=fgraph}]
    \node (X1Y4) at (0,-0.5) {};
    \node (X3Y4) at (0,0.0) {};
    \node (X1Y3) at (0,0.5) {};
    \node (X1Y2) at (0,1) {};
    \node (Y1X4) at (1,-0.5) {};
    \node (Y3X4) at (1,0.0) {};
    \node (Y1X3) at (1,0.5) {};
    \node (Y1X2) at (1,1) {};
    \node (Z3Z4) at (0.5,-1.) {};
\end{scope}
\node[left] at (X1Y4.west) {$X_2Y_4$};
\node[left] at (X3Y4.west) {$X_1Y_4$};
\node[left] at (X1Y3.west) {$X_2Y_1$};
\node[left] at (X1Y2.west) {$X_2Y_3$};
\node[right] at (Y1X4.east) {$X_4Y_2$};
\node[right] at (Y3X4.east) {$X_4Y_1$};
\node[right] at (Y1X3.east) {$X_1Y_2$};
\node[right] at (Y1X2.east) {$X_3Y_2$};
\node[below] at (Z3Z4.south) {$Z_1Z_4$};
\begin{scope}[every edge/.style={draw=black,very thick, color=fgraph}]
    \path [-] (X1Y2) edge (Y1X3);
    \path [-] (X1Y2) edge (Y1X4);
    \path [-] (X1Y3) edge (Y1X2);
    \path [-] (X1Y3) edge (Y1X4);
    \path [-] (X1Y4) edge (Y1X2);
    \path [-] (X1Y4) edge (Y1X3);
    \path [-] (X1Y4) edge (Y3X4);
    \path [-] (Y1X4) edge (X3Y4);
    \path [-] (Z3Z4) edge (Y1X4);
    \path [-] (Z3Z4) edge (Y1X3);
    \path [-] (Z3Z4) edge (X1Y4);
    \path [-] (Z3Z4) edge (X1Y3);
    \path [-] (X1Y3) edge (X3Y4);
    \path [-] (Y1X3) edge (Y3X4);
\end{scope}
\end{tikzpicture}
\end{align*}
We will show that we can obtain $Z_1Z_3$ from this set of generators $\mcA$, which then will prove that we can obtain it from the generators of $\aa^\Omega_6$.
Note that $Z_1Z_3 = X_2 Y_3\cdot X_3Y_2 \cdot X_4Y_2 \cdot X_2 Y_4 \cdot Z_1 Z_4$. We start from $Z_1 Z_4$ and do the following adding operations:
\begin{align*}
\begin{tikzpicture}[baseline=4mm, scale=1.0] 
\begin{scope}[every node/.style={circle,thick,draw,scale=0.5, color=fgraph}]
    \node (X1Y4) at (0,-0.5) {};
    \node (X3Y4) at (0,0.0) {};
    \node (X1Y3) at (0,0.5) {};
    \node (X1Y2) at (0,1) {};
    \node (Y1X4) at (1,-0.5) {};
    \node (Y3X4) at (1,0.0) {};
    \node (Y1X3) at (1,0.5) {};
    \node (Y1X2) at (1,1) {};
    \node[fill=fgraph] (Z3Z4) at (0.5,-1.) {};
\end{scope}
\node[left] at (X1Y4.west) {$X_2Y_4$};
\node[left] at (X3Y4.west) {$X_1Y_4$};
\node[left] at (X1Y3.west) {$X_2Y_1$};
\node[left] at (X1Y2.west) {$X_2Y_3$};
\node[right] at (Y1X4.east) {$X_4Y_2$};
\node[right] at (Y3X4.east) {$X_4Y_1$};
\node[right] at (Y1X3.east) {$X_1Y_2$};
\node[right] at (Y1X2.east) {$X_3Y_2$};
\node[below] at (Z3Z4.south) {$Z_1Z_4$};
\begin{scope}[every edge/.style={draw=black,very thick, color=fgraph}]
    \path [-] (X1Y2) edge (Y1X3);
    \path [-] (X1Y2) edge (Y1X4);
    \path [-] (X1Y3) edge (Y1X2);
    \path [-] (X1Y3) edge (Y1X4);
    \path [-] (X1Y4) edge (Y1X2);
    \path [-] (X1Y4) edge (Y1X3);
    \path [-] (X1Y4) edge (Y3X4);
    \path [-] (Y1X4) edge (X3Y4);
    \path [-] (Z3Z4) edge (Y1X4);
    \path [-] (Z3Z4) edge (Y1X3);
    \path [-] (Z3Z4) edge (X1Y4);
    \path [-] (Z3Z4) edge (X1Y3);
    \path [-] (X1Y3) edge (X3Y4);
    \path [-] (Y1X3) edge (Y3X4);
\end{scope}
\end{tikzpicture}
\to
\begin{tikzpicture}[baseline=4mm, scale=1.0] 
\begin{scope}[every node/.style={circle,thick,draw,scale=0.5, color=fgraph}]
    \node[fill=fgraph] (X1Y4) at (0,-0.5) {};
    \node (X3Y4) at (0,0.0) {};
    \node (X1Y3) at (0,0.5) {};
    \node (X1Y2) at (0,1) {};
    \node[fill=fgraph] (Y1X4) at (1,-0.5) {};
    \node (Y3X4) at (1,0.0) {};
    \node (Y1X3) at (1,0.5) {};
    \node (Y1X2) at (1,1) {};
    \node[fill=fgraph] (Z3Z4) at (0.5,-1.) {};
\end{scope}
\node[left] at (X1Y4.west) {$X_2Y_4$};
\node[left] at (X3Y4.west) {$X_1Y_4$};
\node[left] at (X1Y3.west) {$X_2Y_1$};
\node[left] at (X1Y2.west) {$X_2Y_3$};
\node[right] at (Y1X4.east) {$X_4Y_2$};
\node[right] at (Y3X4.east) {$X_4Y_1$};
\node[right] at (Y1X3.east) {$X_1Y_2$};
\node[right] at (Y1X2.east) {$X_3Y_2$};
\node[below] at (Z3Z4.south) {$Z_1Z_4$};
\begin{scope}[every edge/.style={draw=black,very thick, color=fgraph}]
    \path [-] (X1Y2) edge (Y1X3);
    \path [-] (X1Y2) edge (Y1X4);
    \path [-] (X1Y3) edge (Y1X2);
    \path [-] (X1Y3) edge (Y1X4);
    \path [-] (X1Y4) edge (Y1X2);
    \path [-] (X1Y4) edge (Y1X3);
    \path [-] (X1Y4) edge (Y3X4);
    \path [-] (Y1X4) edge (X3Y4);
    \path [-] (Z3Z4) edge (Y1X4);
    \path [-] (Z3Z4) edge (Y1X3);
    \path [-] (Z3Z4) edge (X1Y4);
    \path [-] (Z3Z4) edge (X1Y3);
    \path [-] (X1Y3) edge (X3Y4);
    \path [-] (Y1X3) edge (Y3X4);
\end{scope}
\end{tikzpicture}
\to
\begin{tikzpicture}[baseline=4mm, scale=1.0] 
\begin{scope}[every node/.style={circle,thick,draw,scale=0.5, color=fgraph}]
    \node[fill=fgraph] (X1Y4) at (0,-0.5) {};
    \node (X3Y4) at (0,0.0) {};
    \node (X1Y3) at (0,0.5) {};
    \node[fill=fgraph] (X1Y2) at (0,1) {};
    \node[fill=fgraph] (Y1X4) at (1,-0.5) {};
    \node (Y3X4) at (1,0.0) {};
    \node (Y1X3) at (1,0.5) {};
    \node[fill=fgraph] (Y1X2) at (1,1) {};
    \node [fill=fgraph](Z3Z4) at (0.5,-1.) {};
\end{scope}
\node[left] at (X1Y4.west) {$X_2Y_4$};
\node[left] at (X3Y4.west) {$X_1Y_4$};
\node[left] at (X1Y3.west) {$X_2Y_1$};
\node[left] at (X1Y2.west) {$X_2Y_3$};
\node[right] at (Y1X4.east) {$X_4Y_2$};
\node[right] at (Y3X4.east) {$X_4Y_1$};
\node[right] at (Y1X3.east) {$X_1Y_2$};
\node[right] at (Y1X2.east) {$X_3Y_2$};
\node[below] at (Z3Z4.south) {$Z_1Z_4$};
\begin{scope}[every edge/.style={draw=black,very thick, color=fgraph}]
    \path [-] (X1Y2) edge (Y1X3);
    \path [-] (X1Y2) edge (Y1X4);
    \path [-] (X1Y3) edge (Y1X2);
    \path [-] (X1Y3) edge (Y1X4);
    \path [-] (X1Y4) edge (Y1X2);
    \path [-] (X1Y4) edge (Y1X3);
    \path [-] (X1Y4) edge (Y3X4);
    \path [-] (Y1X4) edge (X3Y4);
    \path [-] (Z3Z4) edge (Y1X4);
    \path [-] (Z3Z4) edge (Y1X3);
    \path [-] (Z3Z4) edge (X1Y4);
    \path [-] (Z3Z4) edge (X1Y3);
    \path [-] (X1Y3) edge (X3Y4);
    \path [-] (Y1X3) edge (Y3X4);
\end{scope}
\end{tikzpicture}\;,
\end{align*}
which produce the desired product. 
\end{proof}

\begin{lemma}
We have $X_1X_3 \in \aa^\Omega_{14}$.
\end{lemma}
\begin{proof}
As $\aa_{14}=\Lie{XX,XY,YX} = \Lie{XX, ZI, IZ}$, the Lie algebra
$\aa^\Omega_{14}$ can be generated by placing $XX$ on every edge of $\Omega$ and $Z$ on every vertex of $\Omega$.
These generators have the frustration graph
\begin{align*}
\Gamma(\aa_{14}^\Omega) =\:
\begin{tikzpicture}[baseline=4mm, scale=2.0] 
\begin{scope}[every node/.style={circle,thick,draw,scale=0.5, color=fgraph}]
    \node (X1X4) at (0,-0.5) {};
    \node (X3X4) at (0,0) {};
    \node (X1X3) at (0,0.5) {};
    \node (X1X2) at (0,1) {};
    \node (Z4) at (1,-0.5) {};
    \node (Z3) at (1,0) {};
    \node (Z2) at (1,0.5) {};
    \node (Z1) at (1,1) {};
\end{scope}
\node[left] at (X1X4.west) {$X_2X_4$};
\node[left] at (X3X4.west) {$X_1X_4$};
\node[left] at (X1X3.west) {$X_1X_2$};
\node[left] at (X1X2.west) {$X_2X_3$};
\node[right] at (Z4.east) {$Z_4$};
\node[right] at (Z3.east) {$Z_1$};
\node[right] at (Z2.east) {$Z_3$};
\node[right] at (Z1.east) {$Z_2$};
\begin{scope}[every edge/.style={draw=black,very thick, color=fgraph}]
    \path [-] (X1X2) edge (Z1);
    \path [-] (X1X2) edge (Z2);
    \path [-] (X1X3) edge (Z1);
    \path [-] (X1X3) edge (Z3);
    \path [-] (X3X4) edge (Z3);
    \path [-] (X3X4) edge (Z4);
    \path [-] (X1X4) edge (Z1);
    \path [-] (X1X4) edge (Z4);
\end{scope}
\end{tikzpicture}\;.
\end{align*}
Note that $X_1X_3 = X_2 X_3 \cdot X_2 X_4 \cdot X_1 X_4$ can be represented by a colored frustration graph with three blue vertices. 
First, we create a colored frustration graph with three blue vertices by the sequence of add/remove operations:
\begin{align*}
\begin{tikzpicture}[baseline=4mm, scale=1.0] 
\begin{scope}[every node/.style={circle,thick,draw,scale=0.5, color=fgraph}]
    \node (X1X4) at (0,-0.5) {};
    \node (X3X4) at (0,0) {};
    \node (X1X3) at (0,0.5) {};
    \node[fill=fgraph] (X1X2) at (0,1) {};
    \node (Z4) at (1,-0.5) {};
    \node (Z3) at (1,0) {};
    \node (Z2) at (1,0.5) {};
    \node (Z1) at (1,1) {};
\end{scope}
\node[left] at (X1X4.west) {$X_2X_4$};
\node[left] at (X3X4.west) {$X_1X_4$};
\node[left] at (X1X3.west) {$X_1X_2$};
\node[left] at (X1X2.west) {$X_2X_3$};
\node[right] at (Z4.east) {$Z_4$};
\node[right] at (Z3.east) {$Z_1$};
\node[right] at (Z2.east) {$Z_3$};
\node[right] at (Z1.east) {$Z_2$};
\begin{scope}[every edge/.style={draw=black,very thick, color=fgraph}]
    \path [-] (X1X2) edge (Z1);
    \path [-] (X1X2) edge (Z2);
    \path [-] (X1X3) edge (Z1);
    \path [-] (X1X3) edge (Z3);
    \path [-] (X3X4) edge (Z3);
    \path [-] (X3X4) edge (Z4);
    \path [-] (X1X4) edge (Z1);
    \path [-] (X1X4) edge (Z4);
\end{scope}
\end{tikzpicture}
\to
\begin{tikzpicture}[baseline=4mm, scale=1.0] 
\begin{scope}[every node/.style={circle,thick,draw,scale=0.5, color=fgraph}]
    \node (X1X4) at (0,-0.5) {};
    \node (X3X4) at (0,0) {};
    \node (X1X3) at (0,0.5) {};
    \node[fill=fgraph] (X1X2) at (0,1) {};
    \node (Z4) at (1,-0.5) {};
    \node (Z3) at (1,0) {};
    \node (Z2) at (1,0.5) {};
    \node[fill=fgraph] (Z1) at (1,1) {};
\end{scope}
\node[left] at (X1X4.west) {$X_2X_4$};
\node[left] at (X3X4.west) {$X_1X_4$};
\node[left] at (X1X3.west) {$X_1X_2$};
\node[left] at (X1X2.west) {$X_2X_3$};
\node[right] at (Z4.east) {$Z_4$};
\node[right] at (Z3.east) {$Z_1$};
\node[right] at (Z2.east) {$Z_3$};
\node[right] at (Z1.east) {$Z_2$};
\begin{scope}[every edge/.style={draw=black,very thick, color=fgraph}]
    \path [-] (X1X2) edge (Z1);
    \path [-] (X1X2) edge (Z2);
    \path [-] (X1X3) edge (Z1);
    \path [-] (X1X3) edge (Z3);
    \path [-] (X3X4) edge (Z3);
    \path [-] (X3X4) edge (Z4);
    \path [-] (X1X4) edge (Z1);
    \path [-] (X1X4) edge (Z4);
\end{scope}
\end{tikzpicture}
\to
\begin{tikzpicture}[baseline=4mm, scale=1.0] 
\begin{scope}[every node/.style={circle,thick,draw,scale=0.5, color=fgraph}]
    \node (X1X4) at (0,-0.5) {};
    \node (X3X4) at (0,0) {};
    \node[fill=fgraph] (X1X3) at (0,0.5) {};
    \node[fill=fgraph] (X1X2) at (0,1) {};
    \node (Z4) at (1,-0.5) {};
    \node (Z3) at (1,0) {};
    \node (Z2) at (1,0.5) {};
    \node[fill=fgraph] (Z1) at (1,1) {};
\end{scope}
\node[left] at (X1X4.west) {$X_2X_4$};
\node[left] at (X3X4.west) {$X_1X_4$};
\node[left] at (X1X3.west) {$X_1X_2$};
\node[left] at (X1X2.west) {$X_2X_3$};
\node[right] at (Z4.east) {$Z_4$};
\node[right] at (Z3.east) {$Z_1$};
\node[right] at (Z2.east) {$Z_3$};
\node[right] at (Z1.east) {$Z_2$};
\begin{scope}[every edge/.style={draw=black,very thick, color=fgraph}]
    \path [-] (X1X2) edge (Z1);
    \path [-] (X1X2) edge (Z2);
    \path [-] (X1X3) edge (Z1);
    \path [-] (X1X3) edge (Z3);
    \path [-] (X3X4) edge (Z3);
    \path [-] (X3X4) edge (Z4);
    \path [-] (X1X4) edge (Z1);
    \path [-] (X1X4) edge (Z4);
\end{scope}
\end{tikzpicture}
\to
\begin{tikzpicture}[baseline=4mm, scale=1.0] 
\begin{scope}[every node/.style={circle,thick,draw,scale=0.5, color=fgraph}]
    \node[fill=fgraph] (X1X4) at (0,-0.5) {};
    \node (X3X4) at (0,0) {};
    \node[fill=fgraph] (X1X3) at (0,0.5) {};
    \node[fill=fgraph] (X1X2) at (0,1) {};
    \node (Z4) at (1,-0.5) {};
    \node (Z3) at (1,0) {};
    \node (Z2) at (1,0.5) {};
    \node[fill=fgraph] (Z1) at (1,1) {};
\end{scope}
\node[left] at (X1X4.west) {$X_2X_4$};
\node[left] at (X3X4.west) {$X_1X_4$};
\node[left] at (X1X3.west) {$X_1X_2$};
\node[left] at (X1X2.west) {$X_2X_3$};
\node[right] at (Z4.east) {$Z_4$};
\node[right] at (Z3.east) {$Z_1$};
\node[right] at (Z2.east) {$Z_3$};
\node[right] at (Z1.east) {$Z_2$};
\begin{scope}[every edge/.style={draw=black,very thick, color=fgraph}]
    \path [-] (X1X2) edge (Z1);
    \path [-] (X1X2) edge (Z2);
    \path [-] (X1X3) edge (Z1);
    \path [-] (X1X3) edge (Z3);
    \path [-] (X3X4) edge (Z3);
    \path [-] (X3X4) edge (Z4);
    \path [-] (X1X4) edge (Z1);
    \path [-] (X1X4) edge (Z4);
\end{scope}
\end{tikzpicture}
\to
\begin{tikzpicture}[baseline=4mm, scale=1.0] 
\begin{scope}[every node/.style={circle,thick,draw,scale=0.5, color=fgraph}]
    \node[fill=fgraph] (X1X4) at (0,-0.5) {};
    \node (X3X4) at (0,0) {};
    \node[fill=fgraph] (X1X3) at (0,0.5) {};
    \node[fill=fgraph] (X1X2) at (0,1) {};
    \node (Z4) at (1,-0.5) {};
    \node (Z3) at (1,0) {};
    \node (Z2) at (1,0.5) {};
    \node (Z1) at (1,1) {};
\end{scope}
\node[left] at (X1X4.west) {$X_2X_4$};
\node[left] at (X3X4.west) {$X_1X_4$};
\node[left] at (X1X3.west) {$X_1X_2$};
\node[left] at (X1X2.west) {$X_2X_3$};
\node[right] at (Z4.east) {$Z_4$};
\node[right] at (Z3.east) {$Z_1$};
\node[right] at (Z2.east) {$Z_3$};
\node[right] at (Z1.east) {$Z_2$};
\begin{scope}[every edge/.style={draw=black,very thick, color=fgraph}]
    \path [-] (X1X2) edge (Z1);
    \path [-] (X1X2) edge (Z2);
    \path [-] (X1X3) edge (Z1);
    \path [-] (X1X3) edge (Z3);
    \path [-] (X3X4) edge (Z3);
    \path [-] (X3X4) edge (Z4);
    \path [-] (X1X4) edge (Z1);
    \path [-] (X1X4) edge (Z4);
\end{scope}
\end{tikzpicture}\;.
\end{align*}
Then we utilize 
the shortcut move as in Eq.\ \eqref{eq:move_vertex} to move the colored vertex from $X_1X_2$ to $X_1 X_4$:
\begin{align*}
\begin{tikzpicture}[baseline=4mm, scale=1.0] 
\begin{scope}[every node/.style={circle,thick,draw,scale=0.5, color=fgraph}]
    \node[fill=fgraph] (X1X4) at (0,-0.5) {};
    \node (X3X4) at (0,0) {};
    \node[fill=fgraph] (X1X3) at (0,0.5) {};
    \node[fill=fgraph] (X1X2) at (0,1) {};
    \node (Z4) at (1,-0.5) {};
    \node (Z3) at (1,0) {};
    \node (Z2) at (1,0.5) {};
    \node (Z1) at (1,1) {};
\end{scope}
\node[left] at (X1X4.west) {$X_2X_4$};
\node[left] at (X3X4.west) {$X_1X_4$};
\node[left] at (X1X3.west) {$X_1X_2$};
\node[left] at (X1X2.west) {$X_2X_3$};
\node[right] at (Z4.east) {$Z_4$};
\node[right] at (Z3.east) {$Z_1$};
\node[right] at (Z2.east) {$Z_3$};
\node[right] at (Z1.east) {$Z_2$};
\begin{scope}[every edge/.style={draw=black,very thick, color=fgraph}]
    \path [-] (X1X2) edge (Z1);
    \path [-] (X1X2) edge (Z2);
    \path [-] (X1X3) edge (Z1);
    \path [-] (X1X3) edge (Z3);
    \path [-] (X3X4) edge (Z3);
    \path [-] (X3X4) edge (Z4);
    \path [-] (X1X4) edge (Z1);
    \path [-] (X1X4) edge (Z4);
\end{scope}
\end{tikzpicture}
\to
\begin{tikzpicture}[baseline=4mm, scale=1.0] 
\begin{scope}[every node/.style={circle,thick,draw,scale=0.5, color=fgraph}]
    \node[fill=fgraph] (X1X4) at (0,-0.5) {};
    \node (X3X4) at (0,0) {};
    \node (X1X3) at (0,0.5) {};
    \node[fill=fgraph] (X1X2) at (0,1) {};
    \node (Z4) at (1,-0.5) {};
    \node[fill=fgraph] (Z3) at (1,0) {};
    \node (Z2) at (1,0.5) {};
    \node (Z1) at (1,1) {};
\end{scope}
\node[left] at (X1X4.west) {$X_2X_4$};
\node[left] at (X3X4.west) {$X_1X_4$};
\node[left] at (X1X3.west) {$X_1X_2$};
\node[left] at (X1X2.west) {$X_2X_3$};
\node[right] at (Z4.east) {$Z_4$};
\node[right] at (Z3.east) {$Z_1$};
\node[right] at (Z2.east) {$Z_3$};
\node[right] at (Z1.east) {$Z_2$};
\begin{scope}[every edge/.style={draw=black,very thick, color=fgraph}]
    \path [-] (X1X2) edge (Z1);
    \path [-] (X1X2) edge (Z2);
    \path [-] (X1X3) edge (Z1);
    \path [-] (X1X3) edge (Z3);
    \path [-] (X3X4) edge (Z3);
    \path [-] (X3X4) edge (Z4);
    \path [-] (X1X4) edge (Z1);
    \path [-] (X1X4) edge (Z4);
\end{scope}
\end{tikzpicture}
\to
\begin{tikzpicture}[baseline=4mm, scale=1.0] 
\begin{scope}[every node/.style={circle,thick,draw,scale=0.5, color=fgraph}]
    \node[fill=fgraph] (X1X4) at (0,-0.5) {};
    \node[fill=fgraph] (X3X4) at (0,0) {};
    \node (X1X3) at (0,0.5) {};
    \node[fill=fgraph] (X1X2) at (0,1) {};
    \node (Z4) at (1,-0.5) {};
    \node (Z3) at (1,0) {};
    \node (Z2) at (1,0.5) {};
    \node (Z1) at (1,1) {};
\end{scope}
\node[left] at (X1X4.west) {$X_2X_4$};
\node[left] at (X3X4.west) {$X_1X_4$};
\node[left] at (X1X3.west) {$X_1X_2$};
\node[left] at (X1X2.west) {$X_2X_3$};
\node[right] at (Z4.east) {$Z_4$};
\node[right] at (Z3.east) {$Z_1$};
\node[right] at (Z2.east) {$Z_3$};
\node[right] at (Z1.east) {$Z_2$};
\begin{scope}[every edge/.style={draw=black,very thick, color=fgraph}]
    \path [-] (X1X2) edge (Z1);
    \path [-] (X1X2) edge (Z2);
    \path [-] (X1X3) edge (Z1);
    \path [-] (X1X3) edge (Z3);
    \path [-] (X3X4) edge (Z3);
    \path [-] (X3X4) edge (Z4);
    \path [-] (X1X4) edge (Z1);
    \path [-] (X1X4) edge (Z4);
\end{scope}
\end{tikzpicture}\;.
\end{align*}
This gives that $X_1X_3 \in \aa^\Omega_{14}$.
\end{proof}
\fi

\end{document}